\documentclass[reqno,11pt]{amsart}

\usepackage{xspace}
\usepackage[utf8]{inputenc}
\usepackage{graphicx}
\usepackage{amssymb}
\usepackage{mathrsfs}
\usepackage[active]{srcltx}
\usepackage{url}
\usepackage{hyperref}
\usepackage{amsmath,amstext,amsxtra,amsgen,amsbsy,amsopn,amscd,amsthm,amsfonts}
\usepackage{latexsym}
\usepackage{dsfont}
\usepackage{enumitem}
\usepackage{color}

\newtheorem{proposition}{Proposition}
\newtheorem{lemma}[proposition]{Lemma}
\newtheorem{corollary}[proposition]{Corollary}
\newtheorem{theorem}[proposition]{Theorem}

\theoremstyle{definition}

\newtheorem{definition}[proposition]{Definition}
\newtheorem{example}[proposition]{Example}

\theoremstyle{remark}

\newtheorem{remark}[proposition]{Remark}

\numberwithin{proposition}{section}

\newcommand\R{{\ensuremath {\mathbb R} }}
\newcommand\C{{\ensuremath {\mathbb C} }}
\newcommand\N{{\ensuremath {\mathbb N} }}
\newcommand\Z{{\ensuremath {\mathbb Z} }}

\renewcommand\phi{\varphi}
\renewcommand\iint{\int \hspace{-0.3cm} \int}
\renewcommand\le{\leqslant}
\renewcommand\ge{\geqslant}
\renewcommand\epsilon{\varepsilon}
\renewcommand\hat{\widehat}
\renewcommand\tilde{\widetilde}

\newcommand{\gH}{\mathfrak{H}}

\newcommand{\gS}{\mathfrak{S}}
\newcommand{\wto}{\rightharpoonup}

\newcommand{\cD}{\mathcal{D}}
\newcommand{\cH}{\mathcal{H}}
\newcommand\ii{{\ensuremath {\infty}}}

\newcommand{\cK}{\mathcal{K}}

\newcommand{\cO}{\mathcal{O}}

\newcommand{\cL}{\mathcal{L}}
\newcommand{\cF}{\mathcal{F}}

\newcommand\1{{\ensuremath {\mathds 1} }}

\newcommand{\hcl}{\cH^{{\rm cl}}}
\newcommand{\dps}{\displaystyle}
\newcommand{\opw}{{\rm Op}_{\rm W}}
\newcommand{\oph}{{\rm Op}_{\rm H}}

\newcommand\cS{\mathcal{S}}
\newcommand\ccS{\mathscr{S}}

\DeclareMathOperator{\Div}{div}

\DeclareMathOperator{\ran}{ran}

\DeclareMathOperator{\supp}{supp}

\DeclareMathOperator{\tr}{Tr}

\title{The Hartree and Vlasov equations at positive density}

\author[M. Lewin]{Mathieu Lewin}
\address{CNRS \& CEREMADE, Universit\'e Paris-Dauphine, PSL University, 75016 Paris, France} 
\email{mathieu.lewin@math.cnrs.fr}

\author[J. Sabin]{Julien Sabin}
\address{Laboratoire de Math\'ematiques d'Orsay, Univ. Paris-Sud, CNRS, Universit\'e Paris-Saclay, 91405 Orsay, France} 
\email{julien.sabin@math.u-psud.fr}

\date{\today}

\begin{document}

\begin{abstract}
We consider the nonlinear Hartree and Vlasov equations around a translation-invariant (homogeneous) stationary state in infinite volume, for a short range interaction potential. For both models, we consider time-dependent solutions which have a finite relative energy with respect to the reference translation-invariant state. We prove the convergence of the Hartree solutions to the Vlasov ones in a semi-classical limit and obtain as a by-product global well-posedness of the Vlasov equation in the (relative) energy space. 
\end{abstract}

\maketitle

\tableofcontents

\newpage
\section{Introduction}

In this paper we study two mean-field-type models describing the dynamics of a gas of infinitely many interacting particles. These are the Hartree model where the particles are quantum and the kinetic Vlasov model in which they are purely classical. In both cases we give ourselves a reference stationary state which is translation-invariant (homogeneous) in infinite volume. Our goal is to give a proper mathematical meaning to the dynamics of the infinite gas when it is placed at initial time in a state which is a \emph{local}, but not necessarily small, perturbation of the reference stationary state. By ``local'' we mean that it has a finite relative (free) energy with respect to the reference state. As we will see, this implies in particular that the spatial distribution of particles is an $L^p$ perturbation of a constant (the positive density of the reference state), which is a sort of locality.

In a previous work~\cite{LewSab-13a} we have proved the global well-posedness of the Hartree equation, in the relative (free) energy space, for short range smooth-enough interactions in dimensions $d=1,2,3$. This was extended to delta interaction potentials in~\cite{CheHonPav-16}. The asymptotic stability of (linearly stable) stationary states was established in dimension $d=2$ in~\cite{LewSab-13b} and in higher dimensions in~\cite{CheHonPav-17,ColSuz-19_ppt}.

In the context of the Vlasov equation, the analogous asymptotic stability of homogeneous states is usually called \emph{Landau damping}~\cite{Landau-46}. This was proved in finite volume in~\cite{MouVil-11,BedMouMas-16}, and recently established in infinite volume in~\cite{BedMouMas-18,HanNguRou-19}. As explained in \cite[Remark 2.2]{BedMouMas-18}, the global well-posedness of the Vlasov equation for any (possibly large) local perturbation of the reference stationary state has not yet been proved in infinite volume. The main difficulty is to handle the reference state which has an infinite number of particles and an infinite energy (which is not the case in finite volume). 

In this paper we are able to show the \emph{global well-posedness of the Vlasov equation in the relative (free) energy space}, using our previous results about the Hartree equation~\cite{LewSab-13a}. Our strategy is to pass to the semi-classical limit in the Hartree equation in order to construct solutions to the Vlasov problem. This requires a uniform control of the Hartree solution in the semi-classical parameter $\hbar$ which, in our case, follows from the conservation of the relative (free) energy. 

\medskip

Let us now describe more precisely the two models we consider in this paper. In the Hartree model the state of the system is described by a \emph{density matrix} $\gamma$ which is a non-negative and bounded self-adjoint operator acting on $L^2(\R^d)$, where $d$ is the space dimension. The number of particles is given by the trace $\tr(\gamma)$ which will always be infinite for us. A translation-invariant state corresponds to an operator $\gamma$ commuting with space translations, that is, a Fourier multiplier $\gamma=g(-i\hbar\nabla)$ for some bounded velocity distribution $g:\R^d\to\R_+$. This means that $\widehat{\gamma u}(\xi)=g(\hbar\xi)\widehat{u}(\xi)$ for all $u\in L^2(\R^d)$. The parameter $\hbar>0$ plays the role of Planck's constant but will be taken to zero in a semi-classical limit. The density of particles $\rho_\gamma:\R^d\to\R_+$ is formally defined by $\rho_\gamma(x)=\gamma(x,x)$ where $\gamma(x,y)$ is the integral kernel of $\gamma$, when it exists. A Fourier multiplier $\gamma=g(-i\hbar\nabla)$ has the constant density
\begin{equation}
\forall x\in\R^d,\ g(-i\hbar\nabla)(x,x)=\frac1{(2\pi\hbar)^d}\int_{\R^d}g(v)\,dv 
 \label{eq:density_g_inabla}
\end{equation}
which is finite under the assumption that $g$ is integrable. The time evolution of a system of particles interacting with a pair potential $w$ is governed by the Hartree equation
\begin{equation}\label{eq:hartree}
 \begin{cases}
  i\hbar\partial_t\gamma=\big[-\hbar^2\Delta+\hbar^dw\ast\rho_\gamma,\gamma(t)\big],\\
  \gamma_{|t=0}=\gamma_0.
 \end{cases}
\end{equation}
Here $[\cdot,\cdot]$ denotes the commutator.  That the nonlinear interaction is multiplied by the small parameter $\hbar^d$ is typical of the mean-field scaling and compatible with the fact that the density of $g(-i\hbar\nabla)$ in~\eqref{eq:density_g_inabla} is of order $\hbar^{-d}$. Note that the translation-invariant states $\gamma=g(-i\hbar\nabla)$ are all stationary states for~\eqref{eq:hartree}, since they commute both with the Laplacian $-\Delta$ and the constant $\hbar^dw\ast\rho_\gamma=(2\pi)^{-d}\int_{\R^d}w\int_{\R^d}g$. The operator $\gamma_0$ is the one-body density matrix of the system at the initial time $t=0$, which is taken so as to have a finite relative energy with respect to such a reference state 
$$\gamma_{\rm ref}=g(-i\hbar\nabla).$$ 
The precise concept of relative energy will be discussed later. 

In the corresponding classical Vlasov model, the system is described by a non-negative bounded function $m$ on the phase space $\R^d\times\R^d$, and its spatial density is given by 
$$\rho_m(x)=\frac1{(2\pi)^d}\int_{\R^d}m(x,v)\,dv.$$
Translation-invariant states satisfy $m(x+\tau,v)=m(x,v)$ for all $\tau$ and therefore only depend on the velocity $v$. The classical equivalent of the quantum homogeneous state $g(-i\hbar \nabla)$ is simply the function $(x,v)\mapsto g(v)$. Its density is the same as in the quantum case (up to the factor $\hbar^{-d}$), namely $(2\pi)^{-d}\int_{\R^d}g(v)\,dv$. The time evolution is governed by the Vlasov equation
\begin{equation}\label{eq:vlasov}
 \begin{cases}
  \partial_t m + 2v\cdot\nabla_x m -\nabla_x\big(w\ast \rho_m\big)\cdot\nabla_v m=0,\\
  m_{|t=0}=m_0.
 \end{cases}
\end{equation}
Here again we think of $m_0$ as a local perturbation of some reference state 
$$m_{\rm ref}(x,v)=g(v).$$
Notice that the quantum evolution \eqref{eq:hartree} and the classical one \eqref{eq:vlasov} are formally very close, in the sense that \eqref{eq:vlasov} can be rewritten as 
$$\partial_tm=-\{H,m\}$$
where $H$ is the classical Hamiltonian $H(x,v)=|v|^2+w\ast \rho_m$ and $\{\cdot,\cdot\}$ denotes the Poisson bracket $\{a,b\}:=\nabla_v a\cdot\nabla_x b-\nabla_x a\cdot\nabla_v b$. Hence, passing from \eqref{eq:hartree} to \eqref{eq:vlasov} amounts formally to replace the commutator $(i/\hbar)[H_\hbar,\cdot]$ by the Poisson bracket $\{H,\cdot\}$ where $H_\hbar=-\hbar^2\Delta+\hbar^dw\ast\rho_\gamma$ is the mean-field quantum Hamiltonian, provided that $\rho_\gamma$ is of order $\hbar^{-d}$, as expected. 

Relevant examples of translation-invariant states $\gamma_{\rm ref}=g(-i\hbar\nabla)$ and $m_{\rm ref}(x,v)=g(v)$ are those given by the following choices of functions $g$.
\begin{itemize}[leftmargin=15pt]
 \item Fermi gas at zero temperature and chemical potential $\mu>0$:
 $$g(v)=\1(|v|^2\le\mu).$$
 \item Fermi gas at temperature $T>0$ and chemical potential $\mu\in\R$:
 $$g(v)=\frac{1}{e^{\frac{1}{T}(|v|^2-\mu)}+1}.$$
 \item Bose gas at temperature $T>0$ and chemical potential $\mu<0$:
 $$g(v)=\frac{1}{e^{\frac{1}{T}(|v|^2-\mu)}-1}.$$
 \item Boltzmann gas at temperature $T>0$ and chemical potential $\mu\in\R$:
 $$g(v)=e^{-\frac{1}{T}(|v|^2-\mu)}.$$
\end{itemize}
Our main results cover all these cases, except the first one where the singularity of $g$ creates some complications. A key property of all these cases is their variational interpretation as \emph{formal} minimizers of the \emph{free energy}
$$\gamma\mapsto \tr\big((-\hbar^2\Delta-\mu)\gamma-T S(\gamma)\big)$$
in the quantum case, respectively 
$$m\mapsto \frac1{(2\pi)^d}\iint_{\R^d\times\R^d}\Big((|v|^2-\mu)m(x,v)-TS\big(m(x,v)\big)\Big)\,dx\,dv$$
in the classical case, where the entropy $S$ is given by
\begin{equation}
S(m)=\begin{cases}
   -m\log(m)-(1-m)\log(1-m)&\text{(Fermi gas),}\\
   -m\log(m)+(1+m)\log(1+m)&\text{(Bose gas),}\\
   -m\log(m)+m&\text{(Boltzmann gas).}\\
   \end{cases} 
   \label{eq:def_entropy_intro}
\end{equation}
Of course, the minimum of the free energy is $-\ii$ in all cases. 

It is convenient to resort to the concept of \emph{relative entropy} which is the formal difference of the free energy of a state and that of the translation-invariant state. After dividing by $T$, in the quantum case this is just (formally) equal to
$$\cH(\gamma,\gamma_{\rm ref}):=\tr\Big(\frac{-\hbar^2\Delta-\mu}{T}\big(\gamma-\gamma_{\rm ref}\big)- S(\gamma)+S(\gamma_{\rm ref})\Big).$$
It is possible to give a clear mathematical meaning to $\cH(\gamma,\gamma_{\rm ref})$ for all possible $\gamma$, allowing the value $+\ii$~\cite{LewSab-13,DeuHaiSei-15}. This naturally defines a class of states such that $\cH(\gamma,\gamma_{\rm ref})$ is finite, which is our definition of the \emph{(relative) energy space}. In the defocusing case, this space is stable under the Hartree flow, with the relative entropy $\cH(\gamma(t),\gamma_{\rm ref})$ staying uniformly bounded for all times. This is the main property that can be used to obtain global solutions~\cite{LewSab-13a}.

In the classical case the relative entropy $\hcl(m,m_{\rm ref})$ is defined similarly by
\begin{multline*}
  \hcl(m,m_{\rm ref})=\\
  \frac{1}{(2\pi)^d}\iint_{\R^d\times\R^d}\left(\frac{|v|^2-\mu}{T}\big(m(x,v)-g(v)\big)-S\big(m(x,v)\big)+S\big(g(v)\big)\right)\,dx\,dv.
\end{multline*}
In the previous formula, the integrand is a non-negative function since $S'(g(v))=(|v|^2-\mu)/T$ and $S$ is concave. Therefore, $\hcl(m,m_{\rm ref})$ always makes sense in $[0,+\ii]$. Again, the relevant (relative) energy space in the classical case is the set of $m$ such that $\hcl(m,m_{\rm ref})<+\ii$. In this article, we show that the Vlasov equation is globally well-posed in this space for short-range defocusing interactions, from the corresponding result in the quantum case of~\cite{LewSab-13a} and performing the semi-classical limit $\hbar\to0$.

In statistical physics, translation-invariant equilibrium states describe fluid phases. At given values of thermodynamic parameters (temperature and chemical potential for instance), a fundamental question is to determine whether the system is in a fluid phase or in another phase such as a solid. Hence, studying the stability of fluid equilibrium states under perturbations breaking translation-invariance is a first step towards an answer to this question. In that respect, mean-field models (like the Hartree and Vlasov above) are very special. Indeed, the class of translation-invariant equilibrium states is very large and independent of the interaction potential $w$! The only effect of the interaction is to renormalize the chemical potential in the manner $\mu\to\mu-\rho\int_{\R^d} w$. Another feature of mean-field models is that the (density-density) linear response under local perturbations can be computed easily. This leads to the linear stability criterion which bears the name of Penrose~\cite{Penrose-60} for the Vlasov equation and has a well-known quantum analogue (see for instance \cite[Sec. 4.7]{GiuVig-05} and \cite[Ch. IV, Sec. 2.4]{LunMar-83}). Extending some of our results beyond mean-field theory is a challenging open problem.

\medskip

The paper is organized as follows. In the next section, we summarize the main findings of \cite{LewSab-13a} and state our main results about the semi-classical limit of the Hartree equation as well as the global well-posedness of the Vlasov equation in the energy space. In Section \ref{sec:sc-limit}, we follow the strategy of \cite{LioPau-93} to perform the semi-classical limit, which is based on the study of the Wigner function, a natural phase-space distribution associated with the density matrix $\gamma(t)$. The semi-classical limit of the Wigner distribution satisfies a transport equation with a potential that remains to be identified with the mean-field one. This is done in several steps. In Section \ref{sec:BL}, we prove Berezin-Lieb-type inequalities \cite{Berezin-72,Lieb-73} which relate the quantum relative entropy and the classical relative entropy of the Husimi transform. While they are well-known in the context of the \emph{entropy} (see for instance the review \cite{Simon-81}), Berezin-Lieb inequalities seem to be less understood in the context of \emph{relative} entropy. We mention several open problems in this direction. These Berezin-Lieb inequalities imply that the Husimi transform has a classical relative entropy which is uniformly bounded in $\hbar$. In Section~\ref{sec:high-freq}, we exploit this uniform boundedness to derive high-velocity estimates on the Husimi transform which enable us to identify the limiting potential in the Vlasov equation as the mean-field one. At this step we have constructed one solution to the Vlasov equation to the semi-classical limit. In Section \ref{sec:uniqueness}, we establish uniqueness of solutions to the Vlasov equations in the relative energy space. To this end, we follow the strategy of Loeper~\cite{Loeper-06} and use optimal transport methods. In our positive density case, we have to adapt these techniques to \emph{infinite mass measures}. This may be of independent interest, and is quickly explained in Appendix \ref{app:opt-transp}. Finally, in Section \ref{sec:wp-vlasov}, we solve the Cauchy problem for the Vlasov equation in the relative energy space by constructing a suitable sequence of quantum states converging towards the initial datum of the Vlasov equation.

\bigskip

\noindent{\textbf{Acknowledgements:}} This project has received funding from the European Research Council (ERC) under the European Union's Horizon 2020 research and innovation programme (grant agreement MDFT No 725528 of M.L.). J.S. is
supported by ANR DYRAQ ANR-17-CE40-001.

\section{Main results}

Since we want to study the limit as $\hbar\to0$ of solutions to \eqref{eq:hartree}, it is natural to ask that these solutions exist up to a time that is independent of $\hbar$. The most natural setting where it is the case is the context of global solutions (infinite time of existence). Global solutions are usually obtained either in a perturbative regime where the nonlinearity remains small for all times, or using coercive conservation laws that prevent the solution to explode in finite time in adequate function spaces. In the positive density context, the standard conservation laws (number of particles, energy) are all infinite. They have to be replaced by relative quantitites measured with respect to the infinite reference state $\gamma_{\rm ref}$, a concept that we recall below. In this setting, the existence of global solutions (under a defocusing condition on the interaction potential) to the Hartree equation \eqref{eq:hartree} has been proved in \cite{LewSab-13a}. Let us also mention \cite{CheHonPav-16} for an extension of these results to a Dirac interaction potential, in the particular case of the Fermi gas at zero temperature.

In order to state the main theorem of \cite{LewSab-13a}, which is the starting point of our analysis, we need to recall its functional setting. In particular, we start by recalling the important concept of relative energies developed in \cite{LewSab-13a}.

\subsection{Relative entropies}

As we recalled in the introduction, the Fermi, Bose, and Boltzmann gas at positive temperature have the variational property to minimize a relative entropy functional\footnote{The Fermi gas at zero temperature also has a variational interpretation that we used in \cite{LewSab-13a}, although a bit different. Here we only treat the positive temperature case.}. In \cite{LewSab-13a}, we considered a large class of translation-invariant states $\gamma_{\rm ref}$ which all share the variational property to minimize some relative entropy functional. Hence, let us first define these entropy functionals. 

\begin{definition}[Admissible entropy functional]\label{def:entropy}
 Let $d\ge1$. We say that $S:[0,1]\to\R_+$ is an \emph{admissible entropy functional} on $\R^d$ if:
 \begin{enumerate}
  \item $S$ is continuous on $[0,1]$ and differentiable on $(0,1)$;
  \item $-S'$ is operator monotone on $(0,1)$ and not constant (in particular, $S$ is $C^\ii$ on $(0,1)$ and $S$ is concave);
  \item $\lim_{m\to0}S'(m)=+\ii$ and $\lim_{m\to 1}S'(m)\leq0$. 
  \item $S''(S'_+)^{d/2-1}\in L^1(0,1)$.
 \end{enumerate}
\end{definition}
As we proved in \cite{LewSab-13}, the first two properties of the definition allow to define the \emph{relative entropy} 
\begin{equation}\label{eq:relative-entropy}
  \cH_S(A,B):=-\tr\Big(S(A)-S(B)-S'(B)(A-B)\Big)\in[0,+\ii]
\end{equation}
for any operators $0\le A,B\le1$ on any Hilbert space $\gH$ (see \cite{LewSab-13,DeuHaiSei-15} for the correct interpretation of formula \eqref{eq:relative-entropy}). The last two properties (which are adapted to the case $\gH=L^2(\R^d)$) allow to consider the translation invariant state $\gamma_{\rm ref}=(S')^{-1}(-\hbar^2\Delta)$ with a finite density. The main physical examples of admissible entropy functionals $S$ have already been mentioned in~\eqref{eq:def_entropy_intro} and will be denoted in the following by 
\begin{equation}\label{eq:physical-entropies}
\begin{cases}
   S_0(m)=-m\log(m)+m,\\
   S_b(m)=-m\log(m)+(1+m)\log(1+m),\\
   S_f(m)=-m\log(m)-(1-m)\log(1-m).
\end{cases} 
\end{equation}
In these cases, we have 
 $$(S_0')^{-1}(y)=e^{-y},\qquad (S_b')^{-1}(y)=\frac{1}{e^y-1},\qquad (S_f')^{-1}(y)=\frac{1}{e^y+1},$$
where one recognizes the Boltzmann, Bose-Einstein, and Fermi-Dirac distributions. Notice that $S_b$ is not an admissible entropy functional since $S_b'(1)=\log2>0$. The standard way to correct this slight issue is to introduce a chemical potential, as explained in the following remark.

\begin{remark}[Changing the temperature and the chemical potential]\label{rk:temp-mu}
 Introducing a temperature $T>0$ or a chemical potential $\mu\in\R$ amounts to changing $S(m)$ to $TS(m)+\mu m$. The properties of Definition \ref{def:entropy} are then preserved, except perhaps the fact that $\lim_{m\to 1}S'(m)\leq0$ which introduces constraints on $T$ and $\mu$.
\hfill$\diamond$\end{remark}

\begin{remark}
 We only consider entropy functionals defined on $[0,1]$, meaning that we will only consider solutions $\gamma$ to \eqref{eq:hartree} such that $0\le\gamma\le1$. By scaling, we can consider solutions satisfying $0\le\gamma\le M$ (notice that such a constraint is preserved along the Hartree flow), with perhaps bounds that are not optimal in $M$. In order not to overload the exposition, we stick to the case $M=1$.
\hfill$\diamond$\end{remark}

\begin{remark}
 Another interesting example of an admissible entropy functional is $S(m)=m^p-am$ for some $\max(0,1-2/d)<p<1$ and $a\ge p$, in which case 
 $$(S')^{-1}(-\hbar^2\Delta)=\left(\frac{p}{a-\hbar^2\Delta}\right)^\frac{1}{1-p}.$$
\hfill$\diamond$\end{remark}

For any operator $B$, the functional $A\mapsto\cH_S(A,B)$ is non-negative and uniquely minimized at $A=B$. There are obviously many functionals satisfying these properties; in our case the choice of $\cH_S$ is motivated by two additional properties when $\gH=L^2(\R^d)$ and $B=\gamma_{\rm ref}=(S')^{-1}(-\hbar^2\Delta)$: (i) it plays the role of a kinetic energy, in the sense that the sum of $\cH_S$ and the potential energy is conserved along the Hartree flow and (ii) it is sufficiently coercive so that $\cH_S(\gamma,\gamma_{\rm ref})$ controls the density $\rho_\gamma - \rho_{\gamma_{\rm ref}}$ appearing in the mean-field potential. As one can guess, these properties are crucial for the analysis of the well-posedness of the equation.

As we said, the relative entropy $\cH_S$ defined in~\eqref{eq:relative-entropy} plays the role of a kinetic energy and thus determines the energy space in which we will solve the Hartree equation:

\begin{definition}[Semi-classical free energy space]
 Let $\hbar>0$. Let $S$ be an admissible entropy functional. Define $\gamma_{\rm ref}:=(S')^{-1}(-\hbar^2\Delta)$. The \emph{semi-classical free energy space} is 
 $$\cK_S^\hbar:=\Big\{0\le\gamma=\gamma^*\le 1,\ \cH_S(\gamma,\gamma_{\rm ref})<+\ii\Big\}.$$
 \end{definition}

 For any $\gamma\in\cK_S^\hbar$, we define its \emph{semi-classical free energy} by
 \begin{multline}\label{eq:free-energy}
    \cF_S^\hbar(\gamma,\gamma_{\rm ref})\\=\cH_S(\gamma,\gamma_{\rm ref})+\frac{\hbar^d}{2}\int_{\R^d}\int_{\R^d}\rho_{\gamma-\gamma_{\rm ref}}(x)w(x-y)\rho_{\gamma-\gamma_{\rm ref}}(y)\,dx\,dy.
 \end{multline}
As we will see, the semi-classical free energy is conserved along the Hartree flow. For the potential energy (the second term in \eqref{eq:free-energy}) to be well-defined, one needs some assumptions on $w$ and some integrability of $\rho_{\gamma-\gamma_{\rm ref}}$ when $\gamma\in\cK_S^\hbar$. This last condition is ensured by a Lieb-Thirring inequality at positive density, that was first introduced and proved in \cite{FraLewLieSei-12}:

\begin{theorem}[Lieb-Thirring inequality for the relative entropy {\cite[Thm. 7 \& Lem. 9]{LewSab-13a}}]\label{thm:LT-S}
 Let $d\ge1$ and $\hbar>0$. Let $S$ be an admissible entropy functional and $\gamma_{\rm ref}=(S')^{-1}(-\hbar^2\Delta)$. Then, there exists $K_{{\rm LT,S}}>0$ independent of $\hbar$ such that for any $\gamma\in\cK_S^\hbar$, we have 
 \begin{equation}\label{eq:LT-S}
    \cH_S(\gamma,\gamma_{\rm ref})\ge K_{{\rm LT,S}}\begin{cases}
	  \dps \hbar^{-d}\int_{\R^d}\delta\Big(\hbar^d\rho_\gamma(x),\hbar^d\rho_{\gamma_{\rm ref}}(x)\Big)\,dx & {\rm if}\ d\ge2,\\
	  \\
	  \dps \hbar^{-d}\int_\R|\hbar^d\rho_{\gamma-\gamma_{\rm ref}}(x)|^2\,dx & {\rm if}\ d=1,
       \end{cases}
 \end{equation}
 where we used the notation
\begin{equation}
\delta(\rho,\rho_0):=\rho^{1+2/d}-\rho_0^{1+2/d}-\frac{d+2}{d}\rho_0^{2/d}(\rho-\rho_0),\qquad \rho,\rho_0\ge0.
\label{eq:def_delta}
\end{equation}
 \end{theorem}

 \begin{remark}
  In \cite{LewSab-13a}, Theorem \ref{thm:LT-S} was proved for $\hbar=1$. A simple scaling argument gives the version we stated. 
 \hfill$\diamond$\end{remark}
 
 \begin{remark}
 By the Lieb-Thirring inequality for the relative entropy \eqref{eq:LT-S}, the free energy \eqref{eq:free-energy} is well-defined in the energy space $\cK_S^\hbar$ if $w\in L^1(\R^d)$ when $d=1,2$ and if $w\in (L^1\cap L^{(d+2)/4})(\R^d)$ when $d\ge3$.
\hfill$\diamond$\end{remark}

It is useful to introduce one last property of the relative entropy. The Lieb-Thirring inequality \eqref{eq:LT-S} states that the relative entropy controls the relative density $\rho_\gamma-\rho_{\gamma_{\rm ref}}$. The next result shows that it also controls the relative operator $\gamma-\gamma_{\rm ref}$. 

\begin{lemma}[Klein inequality {\cite[Lem. 8]{LewSab-13a}}]
\label{lem:klein-LT-h}
 Let $d\ge1$ and $\hbar>0$. Let $S$ be an admissible entropy functional and $\gamma_{\rm ref}=(S')^{-1}(-\hbar^2\Delta)$. Then, there exists $C>0$ independent of $\hbar$ such that for any $\gamma\in\cK_S^\hbar$ we have
 $$\cH_S(\gamma,\gamma_{\rm ref})\ge C\tr(1-\hbar^2\Delta)(\gamma-\gamma_{\rm ref})^2.$$
\end{lemma}

\begin{remark}\label{rk:regularity-operator}
 In the language of \cite{LewSab-13a}, this means that the relative entropy controls the $\gS^{2,1}$-norm of $\gamma-\gamma_{\rm ref}$, where for any $s\in\R$ and any bounded operator $Q$ we set
 $$\|Q\|_{\gS^{2,s}}:=\left\| (1-\Delta)^{s/4}Q\right\|_{\gS^2},$$
 where $\gS^2$ denotes the Hilbert-Schmidt class.
\hfill$\diamond$\end{remark}

\subsection{Well-posedness of the Hartree equation at positive density}

We now state the global well-posedness result from \cite{LewSab-13a} which is the basis of our work. We use conservation of the free energy \eqref{eq:free-energy} to globalize solutions, hence we ensure that the energy controls the relative entropy by the following defocusing criterion on the interaction potential. 

\begin{definition}[Defocusing interaction potential]
 Let $w\in L^1(\R^d,\R)$. We say that $w$ is \emph{defocusing} if $\hat{w}\ge0$ when $d\ge3$, or if 
 $$\|(\hat{w})_-\|_{L^\ii}<2(2\pi)^{-d/2}K_{{\rm LT,S}},$$
 when $d=1,2$, where $K_{{\rm LT,S}}$ is the (sharp) Lieb-Thirring constant of Theorem~\ref{thm:LT-S}.
\end{definition}

\begin{theorem}[Global well-posedness of the Hartree equation at positive density in the energy space~{\cite[Thm. 9]{LewSab-13a}}]\label{thm:gwp-hartree-h}
 Let $d\in\{1,2,3\}$ and $\hbar>0$. Let $S$ be an admissible entropy functional and 
 $$\gamma_{\rm ref}=(S')^{-1}(-\hbar^2\Delta).$$ 
 Assume 
$$w\in\begin{cases}
(L^1\cap L^\ii)(\R^d) &\text{if $d=1,2$}\\
(W^{1,1}\cap W^{1,\ii})(\R^d) & \text{if $d=3$}
\end{cases}$$
is even and defocusing. Let $\gamma_0\in\cK_S^\hbar$. Then, there exists a unique 
$$\gamma\in\gamma_{\rm ref}+\begin{cases}
C^0_t(\R,\gS^2) & \text{if $d=1,2$}\\
C^0_t(\R,\gS^{2,1}) &\text{if $d=3$}
\end{cases}$$
solution to the time-dependent equation~\eqref{eq:hartree} such that $\gamma(t)\in\cK_S^\hbar$ for all $t\in\R$. Furthermore, the semi-classical free energy is conserved along the flow: we have 
 $$\cF_S^\hbar(\gamma(t),\gamma_{\rm ref})=\cF_S^\hbar(\gamma_0,\gamma_{\rm ref})$$
 for all $t\in\R$.
\end{theorem}

\begin{remark}
 Again, in \cite{LewSab-13a}, Theorem \ref{thm:gwp-hartree-h} is proved for $\hbar=1$. Our statement follows by a simple scaling.
\hfill$\diamond$\end{remark}

\begin{remark}
The statement of Theorem \ref{thm:gwp-hartree-h} remains valid if the defocusing assumption on $w$ is dropped. We then only get a unique local-in-time solution for which the energy is conserved, but which may blow-up in finite time. The time of existence might depend on $\hbar$. 
\hfill$\diamond$\end{remark}

As a consequence of Theorem \ref{thm:gwp-hartree-h}, we obtain bounds on the solutions that are uniform in $\hbar$ and thus will be the starting point of our analysis when $\hbar\to0$. 

\begin{corollary}[Semi-classical bounds for solutions]\label{prop:sc-bounds}
 Let $\gamma\in \gamma_{\rm ref}+C^0_t(\R,\gS^2)$ be any solution to \eqref{eq:hartree}, as given by Theorem \ref{thm:gwp-hartree-h}. Assume that $\gamma_0$ is such that 
 $$\hbar^d \cH_S(\gamma_0,\gamma_{\rm ref})\leq K$$
 for some $K>0$. Then, there exists $C=C_S>0$ such that for all $t\in\R$ and for all $\hbar>0$,
 \begin{equation}\label{eq:control-entropy-time}
    \hbar^d\cH_S(\gamma(t),\gamma_{\rm ref})    \le C_S(1+\|w\|_{L^1\cap L^{(d+2)/4}})\left(K+K^{\frac{2d}{d+2}}\right).
 \end{equation}
 By Lemma \ref{lem:klein-LT-h}, we have, in particular,
 $$\tr(\gamma(t)-\gamma_{\rm ref})^2\le C\hbar^{-d},\quad \int_{\R^d}\delta\Big(\hbar^d\rho_{\gamma(t)}(x),\hbar^d\rho_{\gamma_{\rm ref}}(x)\Big)\,dx \le C,$$
 where the function $\delta$ is defined in~\eqref{eq:def_delta}.
\end{corollary}

\begin{proof}
 Since $w$ is defocusing, we have that 
 $$\cH_S(\gamma(t),\gamma_{\rm ref})\le \cF_S^\hbar(\gamma(t),\gamma_{\rm ref})$$
 when $d\ge3$, while for $d=1,2$ we have by the Lieb-Thirring inequality \eqref{eq:LT-S},
 $$\cF_S^\hbar(\gamma(t),\gamma_{\rm ref})\ge\left(1-\frac{(2\pi)^{d/2}\|\hat{w}_-\|_{L^\ii}}{2K_{{\rm LT,S}}}\right)\cH_S(\gamma(t),\gamma_{\rm ref})$$
 since $\cH_S(\gamma(t),\gamma_{\rm ref})\ge K_{{\rm LT,S}}\hbar^d\|\rho_{\gamma(t)-\gamma_{\rm ref}}\|_{L^2}^2$ in this case. To prove the result, it is thus enough to prove that there exists $C>0$ such that for all $t\in\R$ and all $\hbar>0$, we have
 $$\cF_S^\hbar(\gamma(t),\gamma_{\rm ref})\le C\hbar^{-d}.$$
 By conservation of free energy, it is enough to prove it at $t=0$. The free energy can rewritten as 
 \begin{multline*}
    \cF_S^\hbar(\gamma_0,\gamma_{\rm ref})=\cH_S(\gamma_0,\gamma_{\rm ref})\\
    +\frac{\hbar^{-d}}{2}\int_{\R^d}\int_{\R^d}\hbar^d\rho_{\gamma_0-\gamma_{\rm ref}}(x)w(x-y)\hbar^d\rho_{\gamma_0-\gamma_{\rm ref}}(y)\,dx\,dy.
 \end{multline*}
 By the Lieb-Thirring inequality \eqref{eq:LT-S}, we have the bound
 \begin{multline*}
    \|\hbar^d\rho_{\gamma_0-\gamma_{\rm ref}}\|_{L^2+L^{\min(2,1+2/d)}}\\
    \le C \left(\hbar^d\cH_S(\gamma_0,\gamma_{\rm ref})\right)^{1/2}+C\left(\hbar^d\cH_S(\gamma_0,\gamma_{\rm ref})\right)^{d/(d+2)},
 \end{multline*}
 which implies the result together with the Young inequality and $w\in (L^1\cap L^{(d+2)/4})(\R^d)$.
\end{proof}

\subsection{Classical objects} 

Before giving our main result about the semi-classical limit of solutions to \eqref{eq:hartree}, we introduce some classical quantities. The first one is the Wigner transform, introduced by Wigner in \cite{Wigner-32} and studied in~\cite{LioPau-93,MarMau-93}. A review about works based on the Wigner transform can be found in \cite{Tatarski-83}.

\begin{definition}[Semi-classical Wigner transform]
Let $d\ge1$ and $\gamma$ be an operator on $L^2(\R^d)$ with integral kernel $\gamma(\cdot,\cdot)$. Let $\hbar>0$. The semi-classical Wigner transform $W^\hbar_\gamma$ of $\gamma$ is defined by
$$W_\gamma^\hbar(x,v):=\int_{\R^d}\gamma(x+y/2,x-y/2)e^{-iy\cdot v/\hbar}\,dy,$$
for all $(x,v)\in\R^d\times\R^d$.
\end{definition}

\begin{remark}
 For $\gamma_{\rm ref}=g(-i\hbar\nabla)$ with $g$ integrable, we have
 $$W_{\gamma_{\rm ref}}^\hbar(x,v)=g(v)$$
 and thus $W^\hbar_{\gamma_\hbar}\in L^\ii_x L^1_v$. This motivates our choice of normalization for the Wigner transform (this will also be useful for Berezin-Lieb inequalities).
\hfill$\diamond$\end{remark}

\begin{remark}
 The Wigner transform can be defined by duality with the formula
 $$\langle W^\hbar_\gamma,\phi\rangle =(2\pi \hbar)^d \tr \Big(\gamma \opw^\hbar(\phi)\Big),$$
 for any $\phi\in \cS_{x,v}(\R^d\times\R^d)$,
 where the Weyl quantization is defined by
 $$\opw^\hbar(a)(x,y):=\frac{1}{(2\pi \hbar)^d}\int_{\R^d}a\left(\frac{x+y}{2},v\right)e^{i(x-y)\cdot v/\hbar}\,dv.$$ 
 This allows for instance to define $W^\hbar_\gamma$ in $\cS'_{x,v}$ for any bounded operator $\gamma$ (since $\opw^\hbar(\phi)\in\gS^1$ for any $\phi$ in the Schwartz class $\cS_{x,v}$). 
\hfill$\diamond$\end{remark}

\begin{remark}
 Notice that we have the following relation between the densities of a density matrix $\gamma$ and the density of its Wigner transform:
 $$\rho_{W^\hbar_\gamma}=\hbar^d\rho_\gamma.$$
\hfill$\diamond$\end{remark}

\begin{definition}[Classical relative entropy]\label{def:class-rel-ent}
 Let $S:[0,1]\to\R$ a continuous, concave function, differentiable on $(0,1)$. Let $m,m_0:\R^d_x\times\R^d_v\to[0,1]$ two measurable functions. We define the classical relative entropy of $m$ and $m_0$ with respect to the function $S$ as 
 $$
 \boxed{
 \hcl_S(m,m_0):=\int_{\R^d}\int_{\R^d}\cH_S \big(m(x,v),m_0(x,v)\big)\frac{\,dx\,dv}{(2\pi)^d}.
 }
 $$
 For numbers $y\in[0,1]$ and $y_0\in(0,1)$ we simply have
$$ \cH_S(y,y_0):=S'(y_0)(y-y_0)-S(y)+S(y_0)\ge0.$$
 When $y_0\in\{0,1\}$, the definition of $\cH_S(y,y_0)$ is 
 $$\cH_S(y,y_0):=\sup_{0<\theta<1}\frac1\theta\Big(S(\theta y + (1-\theta)y_0)-(1-\theta)S(y_0)-\theta S(y)\Big)\in[0,+\ii].$$   
 \end{definition}

\begin{remark}
 If $S$ is an admissible entropy functional, the classical relative entropy $\hcl_S$ possesses the same functional properties as the quantum one; for instance for any $m$, $m_0$ we have a Klein inequality analogue to Lemma \ref{lem:klein-LT-h}:
 \begin{equation}\label{eq:klein-classical}
  \hcl_S(m,m_0)\ge C\|m-m_0\|_{L^2_{x,v}}^2,
 \end{equation}
 and a Lieb-Thirring inequality (valid only for $m_0(x,v)=(S')^{-1}(|v|^2)$) analogue to Theorem \ref{thm:LT-S}
 \begin{equation}\label{eq:LT-classical}
    \hcl_S(m,m_0)\ge K_{{\rm LT,S}}\begin{cases}
	  \dps \int_{\R^d}\delta\Big(\rho_m(x),\rho_{m_0}(x)\Big)\,dx & {\rm if}\ d\ge2,\\
	  \\
	  \dps \int_\R|\rho_{m-m_0}(x)|^2\,dx & {\rm if}\ d=1.
       \end{cases}
 \end{equation}
 These properties may be proved analogously to the quantum ones (and should require less assumptions on $S$ to hold).
\end{remark}

\subsection{Statement of the main results}

We now state our main results. They require more assumptions on the entropy functional $S$ than Theorem \ref{thm:gwp-hartree-h}. For the sake of clarity, we only state them for the three physical cases $S\in\{S_0,S_b,S_f\}$ as in \eqref{eq:physical-entropies}, and make a remark afterwards concerning the more precise properties on $S$ needed for our proof. As mentioned before in Remark \ref{rk:temp-mu}, we might need to tune the chemical potential in order to obtain an admissible entropy functional; as a consequence we introduce the class
\begin{equation}\label{eq:entropies-berezin}
 \ccS:=\Big\{m\mapsto TS(m)+\mu m\ {\rm admissible},\  T>0,\ \mu\in\R,\ S\in\{S_0,S_b,S_f\}\Big\}.
\end{equation}

\subsubsection{The semi-classical limit}

Our first result is about the convergence of global-in-time solutions of the Hartree equation~\eqref{eq:hartree} to the Vlasov equation and about the uniqueness of these limiting solutions.

\begin{theorem}[Semi-classical limit of the Hartree equation at positive density]\label{thm:limit-sc}
 Let $d\in\{1,2,3\}$, $\hbar>0$, $S\in\ccS$, and $\gamma_{\rm ref}=(S')^{-1}(-\hbar^2\Delta)$, $m_{\rm ref}(x,v):=(S')^{-1}(|v|^2)$. Let $w\in (W^{2,1}\cap W^{2,\ii})(\R^d)$ be an even, real-valued function. Let $\gamma\in\gamma_{\rm ref}+C^0_t(\R,\gS^2)$ be the unique solution to the Hartree equation \eqref{eq:hartree} associated with an initial condition $\gamma_0$ satisfying 
 $$\liminf_{\hbar\to0} \hbar^d\cH_S(\gamma_0,\gamma_{\rm ref})<+\ii$$
 and $W^\hbar_{\gamma_0}\to W_0$ as $\hbar\to0$ in the sense of distributions on the phase space $\R^d\times\R^d$. Then, there exists $C_S>0$ and 
 $$W\in \Big\{m_{\rm ref}+L^\ii_t(\R,L^2_{x,v}(\R^d\times\R^d))\Big\}\cap C^0_t(\R,\cD'_{x,v}(\R^d\times\R^d))$$
 such that:
 \begin{enumerate}
  \item $0\le W(t)\le 1$ for all $t\in\R$;
  \item $\hcl_S(W(t),m_{\rm ref})\le C_S \liminf_{\hbar\to0}\hbar^d\cH_S(\gamma(t),\gamma_{\rm ref})$ for all $t\in\R$;
  \item $W_{\gamma(t)}^\hbar\wto W(t)$ as $\hbar\to0$ in the sense of distributions on $\R^{2d}$, uniformly on compact sets in $t$;
  \item $\hbar^d\rho_{\gamma}-\hbar^d\rho_{\gamma_{\rm ref}}\to\rho_W-\rho_{m_{\rm ref}}$ as $\hbar\to0$, weakly-$*$ in $L^\ii_t (L^2+L^{\min(1+2/d,2)})$;
  \item $W$ is the unique solution to the nonlinear Vlasov equation
  $$\begin{cases}
        \partial_t W+2v\cdot\nabla_x W-\nabla_x(w*\rho_W)\cdot\nabla_v W=0,\\
        W_{|t=0}=W_0
    \end{cases}
  $$
  in $\Big\{m_{\rm ref}+L^\ii_t(\R,L^2_{x,v}(\R^d\times\R^d))\Big\}\cap C^0_t(\R,\cD'_{x,v}(\R^d\times\R^d))$ such that $\hcl_S(W(t),m_{\rm ref})\in L^\ii_t$.
 \end{enumerate}
\end{theorem}

\begin{remark}[About the constant $C_S$]
 The constant $C_S$ in (2) is obtained in Corollary \ref{coro:ident-density} and, as its proof shows, comes from two contributions: (i) a Berezin-Lieb-type inequality and (ii) the decrease of the classical relative entropy under weak limits. In particular, our proof shows that if $S\in\ccS$ is constructed from either $S_0$ or $S_f$, one can take $C_S=1$, while if $S$ is constructed from the bosonic entropy $S_b$, one can take $C_S=4$. Having $C_S=1$ is the best case scenario, which means that we are not ``losing'' anything when going from quantum to classical. In particular, it should imply that the (classical) free energy is conserved along the flow of the Vlasov equation \eqref{eq:vlasov}, from the corresponding conservation of the quantum free energy. If we considered operators with spectrum in $[0,M]$ instead of $[0,1]$, we could still take $C_S=1$ for $S=S_0,S_f$ but we could only take $C_S=(1+M)^2$ for $S=S_b$.
\hfill$\diamond$\end{remark}

\begin{remark}[About the assumptions needed on $S$]
 Our proof of Theorem~\ref{thm:limit-sc} applies to general admissible entropy functionals $S$ having the following properties:
 \begin{enumerate}
  \item $S$ has the Berezin-Lieb property with constant $C^{BL}_S$ (see Definition \ref{def:BL} below). 
  \item $\hcl_S$ decays under weak limits, up to a multiplicative constant. This is for instance the case if $(x,y)\in[0,1]^2\mapsto\cH_S(x,y)$ is convex (see Corollary \ref{coro:ident-density} and Remark~\ref{rk:weak-limit-bosons}).
 \end{enumerate}
\hfill$\diamond$\end{remark}

The proof of Theorem \ref{thm:limit-sc} is divided into several steps:
\begin{enumerate}
 \item From the quantum bounds uniform in $\hbar$ of Proposition \ref{prop:sc-bounds}, we obtain uniform bounds on the Wigner transform $W^\hbar_{\gamma(t)}$ and on its density $\hbar^d\rho_{\gamma(t)}$. Hence, they admit weak limits, up to a subsequence. These weak limits are shown to satisfy a Vlasov equation (Proposition \ref{prop:sc-limit-wigner}), where the phase-space distribution and its density are decoupled. This is done in Section \ref{sec:der-vla}.
 \item We then show that the density appearing in the limiting Vlasov equation is indeed the density of the phase-space distribution (in other words, that the weak limit of the density is the density of the weak limit). This is done in Section \ref{sec:ident-dens}, based on a tightness property which amounts to a large velocity bound uniform in $\hbar$.
 \item This tightness result is obtained via two ingredients, that are the main inputs of our work: (i) Berezin-Lieb inequalities, introduced in Section \ref{sec:BL}, which control the classical relative entropy in terms of the quantum entropy and (ii) high velocity estimates, proved in Section \ref{sec:high-freq}, in terms of the classical relative entropy, which are thus uniform in $\hbar$.
 \item We next prove uniqueness of solutions to the above nonlinear Vlasov equation in Section \ref{sec:uniqueness}. This allows us to prove that the semi-classical limit holds not only for a subsequence, but as $\hbar\to0$ (see Remark \ref{rk:uniqueness}).
\end{enumerate}

Theorem \ref{thm:limit-sc} shows that the classical Vlasov equation can be recovered from the quantum Hartree equation in some limit. That quantum mechanics retains some features of classical mechanics has been realized from the early days of the theory, as can be seen from Ehrenfest's theorem, from the work of Schr\"odinger on coherent states for the harmonic oscillator \cite{Schro-26}, or more informally from Bohr's correspondence principle. The 'convergence' of quantum mechanics towards classical mechanics as $\hbar\to0$ was also understood early, as in the work of Wigner \cite{Wigner-32}. Since then, many physical and mathematical works have been devoted to this phenomenon in its various manifestations, so we only mention the works related to our question at hand. Among the pioneering mathematical works in this topic are the works of Maslov \cite{Maslov-72} extending the WKB method and the development of microlocal analysis around H\"ormander, which turned out to be a very fruitful formalism to study semi-classical limits, as exploited and developed for instance by Helffer and Robert \cite{HelRob-81}. Parallel to these works are the articles of Hepp \cite{Hepp-74} where the classical limit is performed also in the context of field theory using the coherent states approach,
and Narnhofer and Sewell \cite{NarSew-81}, where they recover the classical equations of motions as semi-classical limits of the full (linear) Schr\"odinger equation, using the BBGKY approach.
Then, Lions and Paul \cite{LioPau-93} and Markowich and Mauser \cite{MarMau-93} simultaneously studied the semi-classical limit of nonlinear equations of the Hartree-type \eqref{eq:hartree} towards the Vlasov equation. Both these works are concerned with the 'zero-density' case, meaning that they are considering solutions with finite number of particles or finite energy. Our work can be seen as an extension of the results of Lions and Paul \cite{LioPau-93} to the positive density case. Contrary to \cite{LioPau-93,MarMau-93}, we will not be able to treat singular interaction potentials $w$ (like the Coulomb one $w(x)=|x|^{-1}$ in $\R^3$). Finally, let us mention that after \cite{LioPau-93,MarMau-93}, several works have been devoted to strengthening these results, in particular by proving a stronger notion of convergence: see the works of Athanassoulis and Paul \cite{AthPau-11}, Benedikter, Porta, Saffirio and Schlein \cite{BenPorSafSch-16}, and Golse, Mouhot, and Paul \cite{GolMouPau-16}. Of course, it would be very interesting to extend these stronger convergence results to the positive density case.

The Hartree and Vlasov equations \eqref{eq:hartree} and \eqref{eq:vlasov} are effective (mean-field) equations. A related and important question is the discussion of the physical regimes in which the use of such equations is justified, that is, the derivation of these equations from the fundamental laws of quantum or classical mechanics. To our knowledge, such a derivation in our context (infinite volume, positive density) is still an open question. In the 'zero density' case, however, many results are known since the pioneering works of Hepp \cite{Hepp-74} in the bosonic case, Braun-Hepp \cite{BraHep-77} and Dobrushin~\cite{Dobrushin-79} in the classical case and Narnhofer-Sewel~\cite{NarSew-81} in the fermionic case. Important works dealing with fermions include~\cite{Spohn-81,BarGolGotMau-03,ElgErdSchYau-04,AkiMarSpa-08,FroKno-11,BenPorSch-14,BacBrePetPicTza-15,PetPic-16,BenPorSafSch-16,GolMouPau-16,GolPau-17,DieRadSch-18}. The relevant physical regime is called the \emph{mean-field} regime and corresponds to particles that interact often and weakly.

\subsubsection{Well-posedness of the Vlasov equation at positive density in the energy space}

Our second main result concerns the well-posedness of the nonlinear Vlasov equation at positive density. 

\begin{theorem}[Well-posedness of the Vlasov equation at positive density]\label{thm:wp-vlasov}
Let $d\in\{1,2,3\}$ and $S\in\ccS$. Define $m_{\rm ref}(x,v):=(S')^{-1}(|v|^2)$. Let $w\in (W^{2,1}\cap W^{2,\ii})(\R^d)$ be an even, real-valued function. Let $W_0\in m_{\rm ref}+ L^2_{x,v}$ be such that $0\le W_0\le 1$ and such that $\hcl_S(W_0,m_{\rm ref})<+\ii$. Then, there exists a unique 
$$W\in (m_{\rm ref}+L^\ii_t(\R,L^2_{x,v}))\cap C^0_t(\R,\cD'_{x,v})$$ 
such that $0\le W(t)\le 1$ and $\hcl_S(W(t),m_{\rm ref})\le C$ for all $t\in\R$, which solves the nonlinear Vlasov equation
$$\begin{cases}
        \partial_t W+2v\cdot\nabla_x W-\nabla_x(w*\rho_W)\cdot\nabla_v W=0,\\
        W_{|t=0}=W_0.
    \end{cases}
$$
\end{theorem}

\begin{example}
 When $S(m)=-m\log m+m$, the reference profile is the Maxwellian $(S')^{-1}(|v|^2)=e^{-|v|^2}$. The previous theorem thus gives the existence and uniqueness of solutions to the Vlasov equation for initial data $W_0$ which satisfy $0\le W_0\le1$ and 
 $$-\int_{\R^d}\int_{\R^d}\Big[- W_0(x,v)\log\frac{W_0(x,v)}{e^{-|v|^2}}+W_0(x,v)-e^{-|v|^2}\Big]\,dx\,dv<+\ii.$$
 Writing $W_0=e^{-|v|^2}+q$, the last assumption is verified if for instance we have
 $$\int(1+|v|^2)|q(x,v)|\,dx\,dv<+\ii.$$
\end{example}

As the assumptions of Theorem \ref{thm:wp-vlasov} suggest, we prove this result using the quantum result Theorem \ref{thm:limit-sc}. This is done in Section \ref{sec:wp-vlasov}. In particular, there are severe restrictions on the entropy functional $S$ coming from the quantum world, while the classical relative entropy requires much less assumptions. We thus expect Theorem \ref{thm:wp-vlasov} to hold in a much more general setting. 

To our knowledge, Theorem \ref{thm:wp-vlasov} is the first result about the existence and uniqueness of non-perturbative solutions to the nonlinear Vlasov equation around a non-trivial homogeneous state $m_{\rm ref}$ in $\R^d$, in the energy space. Perturbative solutions with high regularity have been obtained in \cite{BedMouMas-18} in the context of Landau damping, leaving the case of global-in-time non-perturbative solutions open (see Remark 2.2 in \cite{BedMouMas-18}). Of course, our result has two drawbacks: (i) we consider only nice interaction potentials $w$ and (ii) our assumptions on the reference state $m_{\rm ref}$ coming from the entropies are quite stringent. Relaxing these assumptions on $w$ and $m_{\rm ref}$ will be the goal of future work. Let us finally notice that Theorem \ref{thm:wp-vlasov} does not imply Landau damping (that is, weak convergence of $W(t)$ to $m_{\rm ref}$ as $t\to+\ii$, in other words the asymptotic stability of $m_{\rm ref}$). However, it implies the orbital (or rather Lyapounov, since the orbit is reduced to a single point) stability of $m_{\rm ref}$ in the sense of the relative entropy: given $\epsilon>0$, there exists $\eta>0$ such that for all initial $W_0$ with $\hcl_S(W_0,m_{\rm ref})<\eta$, then $\hcl_S(W(t),m_{\rm ref})<\epsilon$ for all $t\in\R$.

\section{Semi-classical limit}\label{sec:sc-limit}

In this section, we start the proof of Theorem \ref{thm:limit-sc} by computing (in Proposition \ref{prop:sc-limit-wigner} below) the equation satisfied by weak limits as $\hbar\to0$ of the Wigner transform of solutions to the Hartree equation in the free energy space. We obtain the Vlasov equation \eqref{eq:vlasov-decoupled} where the phase-space distribution and the density are decoupled. We next show that the limiting phase-space distributions are between $0$ and $1$ (Corollary \ref{coro:limit-husimi}) and give a criterion to identify the weak limit of the density with the density of the weak limit (Proposition \ref{prop:ident-density}). In the next two sections we will give the tools to prove this criterion. The uniqueness of solutions to the Vlasov equation will be proved in Section~\ref{sec:uniqueness}.

\subsection{Deriving Vlasov}\label{sec:der-vla}

We start by computing the equation satisfied by the Wigner tranform of a solution to the Hartree equation, for fixed $\hbar>0$. To give sense to this equation, we first recall a classical fact about the Wigner transform.

\begin{lemma}\label{prop:wigner-L2}
 Let $\hbar>0$ and $Q\in\gS^2$. Then, $W^\hbar_Q\in L^2_{x,v}$ and we have 
 $$\|W_Q^\hbar\|_{L^2_{x,v}}^2=(2\pi\hbar)^d\tr|Q|^2.$$
\end{lemma}

\begin{proof}
 This follows from Plancherel's theorem and the fact that $W_Q^\hbar=(2\pi)^{d/2}\cF_{y\to v}[Q(x+\cdot/2,x-\cdot/2)](\cdot/\hbar)$.
\end{proof}

\begin{lemma}[Wigner equation]
Let $\hbar>0$ and $\gamma\in\gamma_{\rm ref}+C^0_t(\R,\gS^2)$ any solution to \eqref{eq:hartree} given by Theorem \ref{thm:gwp-hartree-h}. Then, its Wigner transform $W_\gamma^\hbar\in m_{\rm ref}+C^0_t(\R,L^2_{x,v})$ satisfies the equation
\begin{equation}\label{eq:wigner}
\partial_tW^\hbar_{\gamma(t)}+2v\cdot\nabla_x W^\hbar_{\gamma(t)}+K_\hbar(t,x,\cdot)*_v W^\hbar_{\gamma(t)}=0 
\end{equation}
in $\cD'(\R_t\times\R_x^d\times\R_v^d)$, where
$$K_\hbar(t,x,v)=\frac{i}{(2\pi)^d}\int_{\R^d}\frac{V(t,x+\hbar y/2)-V(t,x-\hbar y/2)}{\hbar}e^{-iv\cdot y}\,dy,$$
with $V(t):=w*\hbar^d\rho_{\gamma(t)}$.
\end{lemma}

\begin{remark}\label{rk:wigner}
 The last term of the Wigner equation is understood in the sense of distributions as follows: let $\phi\in C^\ii_c(\R_t\times\R_x^d\times\R_v^d)$. Then, we have
 $$\langle K_\hbar(t,x,\cdot)*_v W^\hbar_{\gamma(t)},\phi\rangle:=\int_\R\langle W_{\gamma(t)}^\hbar,\psi^\hbar[\phi(t)]\rangle_{x,v}\,dt,$$
 where 
 \begin{multline*}
    \psi^\hbar[\phi](x,v):=\\
    \frac{i}{(2\pi)^{d/2}}\int_{\R^d}\cF_{v\to y}\phi(x,y)e^{i y\cdot v}\frac{V(t,x+\hbar y/2)-V(t,x-\hbar y/2)}{\hbar}\,dy.
 \end{multline*}
 Now since $W^\hbar_\gamma\in C^0_t(L^\ii_x L^1_v+L^2_{x,v})$, we can bound
 \begin{multline*}
    |\langle K_\hbar(t,x,\cdot)*_v W^\hbar_{\gamma(t)},\phi\rangle|\\
    \le\|W^\hbar_\gamma\|_{L^\ii_t(L^\ii_x L^1_v+L^2_{x,v})}\|\cF_{v\to y}\psi^\hbar[\phi(t)]\|_{L^1_t(L^1_{x,y}\cap L^2_{x,y})}.
 \end{multline*}
 Using that 
 $$\cF_{v\to y}\psi^\hbar[\phi(t)](x,y)=i\cF_{v\to y}\phi(t,x,y)\frac{V(t,x+\hbar y/2)-V(t,x-\hbar y/2)}{\hbar},$$
 we deduce that 
 $$\|\cF_{v\to y}\psi^\hbar[\phi(t)]\|_{L^1_t(L^1_{x,y}\cap L^2_{x,y})}\le\|\nabla V\|_{L^\ii_x} \||y|(\cF_{v\to y}\phi)(t,x,y)\|_{L^1_t(L^1_{x,y}\cap L^2_{x,y})},$$
 showing that the Wigner equation is well-defined in the sense of distributions (if $\nabla V\in L^\ii_x$, which is the case if $\nabla w\in (L^2\cap L^{1+d/2})(\R^d)$). 
\hfill$\diamond$\end{remark}

\begin{proof}
  This is done in \cite[Prop. II.1]{LioPau-93}.
\end{proof}

\begin{proposition}[Semi-classical limit of the Wigner equation]\label{prop:sc-limit-wigner}
 Let $\gamma\in\gamma_{\rm ref}+C^0_t(\R,\gS^2)$ be any solution to \eqref{eq:hartree} given by Theorem \ref{thm:gwp-hartree-h}, such that $\hbar^d\cH_S(\gamma_0,\gamma_{\rm ref})$ is bounded as $\hbar\to0$. Assume that $\nabla w\in (L^2\cap L^{1+d/2})(\R^d)$. Then, for any sequence $\hbar_n\to0$, there exists a subsequence (that we still denote by $(\hbar_n)$) and  there exist
 $$W\in \left( m_{\rm ref}+L^\ii_t(\R,L^2_{x,v}) \right) \cap C^0_t(\R,\cD'_{x,v})$$
 and
 $$\nu \in L^\ii_t(\R,(L^2+L^{\min(1+2/d,2)})(\R^d)),$$
 such that 
 $$W_{\gamma}^{\hbar_n}\to W$$
 as $n\to\ii$, weakly-$*$ in $m_{\rm ref}+L^\ii_t(\R,L^2_{x,v})$, in $\cD'_{x,v}$ uniformly on compact sets of $\R$, and such that 
 $$\hbar_n^d\rho_{\gamma}-\hbar_n^d\rho_{\gamma_{\rm ref}}\to \nu$$
 as $n\to\ii$, weakly-$*$ in $L^\ii_t(\R,(L^2+L^{\min(1+2/d,2)})(\R^d))$. Furthermore, $(W,\nu)$ satisfy the transport equation
 \begin{equation}\label{eq:vlasov-decoupled}
    \partial_t W+2v\cdot\nabla_x W-\nabla_x (w*\nu(t))\cdot\nabla_v W=0,
 \end{equation}
 in $\cD'(\R_t\times\R^d_x\times\R^d_v)$.
\end{proposition}

\begin{proof}
 The existence of the weak limits follows from the boundedness of $W_{\gamma}^\hbar-m_{\rm ref}$ in $L^\ii_t L^2_{x,v}$ and of $\hbar^d\rho_{\gamma}-\hbar^d\rho_{\gamma_{\rm ref}}$ in $L^\ii_t(L^2+L^{\min(1+2/d,2)})$ as $\hbar\to0$, as a consequence of Proposition \ref{prop:sc-bounds} and Lemma~\ref{prop:wigner-L2}. The fact that $W\in C^0_t\cD'_{x,v}$ and that $W^{\hbar_n}_{\gamma}\to W$ in $\cD'_{x,v}$ uniformly on compact sets follows from the boundedness of $\partial_t W^\hbar_\gamma$ in $\cD'_{x,v}$ as $\hbar\to0$ using the Wigner equation and Remark \ref{rk:wigner}. It thus remains to derive the limiting equation, by taking the limit $\hbar_n\to0$ in the Wigner equation \eqref{eq:wigner} in the sense of distributions. The limit of the first two terms follows from the fact that the map $T\in\cD'_{t,x,v}\mapsto \partial_t T+2v\cdot\nabla_x T\in\cD'_{t,x,v}$ is continuous. We only have to deal with the last term. To do so, Let $\phi\in C^\ii_0(\R_t\times\R^d_x\times\R^d_v)$ and write
 $$W_{\gamma(t)}^\hbar=m_{\rm ref}+q^\hbar(t),\ \hbar^d\rho_{\gamma(t)}=\hbar^d\rho_{\gamma_{\rm ref}}+\nu^\hbar(t).$$
 Up to a subsequence, we may assume that $q^{\hbar_n}(t)\to q(t):=W(t)-m_{\rm ref}$ weakly-$*$ in $L^2_{x,v}$ and $\nu^{\hbar_n}(t)\to \nu(t)$ weakly-$*$ in $L^2+L^{\min(1+2/d,2)}$ as $n\to\ii$, for a.e. $t\in\R$. Since 
 $$\langle K_\hbar(t,x,\cdot)*_v W^\hbar_{\gamma(t)},\phi\rangle:=\int_\R\langle W_{\gamma(t)}^\hbar,\psi^\hbar[\phi(t)]\rangle_{x,v}\,dt$$
 and since $\phi$ is compactly supported in $t$, if we show that $\psi^{\hbar_n}[\phi(t)]\to\psi^0(t)$ strongly in $L^1_xL^\ii_v\cap L^2_{x,v}$ as $n\to\ii$, for a.e. $t\in\R$, where 
 $$\psi_0(t,x,v):=\nabla_v\phi(t,x,v)\cdot\nabla_x(w*\nu(t)),$$
 then by dominated convergence we indeed have that 
 $$\langle K_{\hbar_n}(t,x,\cdot)*_v W^{\hbar_n}_{\gamma(t)},\phi\rangle\to\langle-\nabla_x(w*\nu(t))\cdot\nabla_v W,\phi\rangle$$
 as $n\to\ii$. Let us thus show that $\psi^{\hbar_n}[\phi(t)]\to\psi^0(t)$ strongly in $L^1_xL^\ii_v\cap L^2_{x,v}$, by showing that 
 $$\|\cF_{v\to y}\psi^{\hbar_n}[\phi(t)]-\cF_{v\to y}\psi^0(t)\|_{L^1_{x,y}\cap L^2_{x,y}}\to 0$$
 as $n\to\ii$, for a.e. $t\in\R$. Notice that for all $t,x,v$ we have
 \begin{multline*}
    \cF_{v\to y}\psi^{\hbar_n}[\phi(t)](x,y)-\cF_{v\to y}\psi^0(t,x,v)\\
    =i\cF_{v\to y}\phi(t,x,y)\left(\frac{V(x+\hbar_ny/2)-V(x-\hbar_ny/2)}{\hbar_n}-y\cdot\nabla(w*\nu(t))(x)\right),
 \end{multline*}
 By dominated convergence, it is enough to show that for a.e. $x,y\in\R^d$, we have
 $$\frac{V(x+\hbar_ny/2)-V(x-\hbar_ny/2)}{\hbar_n}-y\cdot\nabla(w*\nu(t))(x)\to0$$
 as $n\to\ii$. This term can rewritten as
 $$\frac{V(x+\hbar_ny/2)-V(x-\hbar_ny/2)}{\hbar_n}-y\cdot\nabla(w*\nu(t))(x)=\langle\nu^{\hbar_n}(t),g_{\hbar_n}\rangle-\langle\nu(t),g_0\rangle,$$
 where 
 $$g_\hbar(z):=\frac{w(x+\hbar y/2-z)-w(x-\hbar y/2-z)}{\hbar},\ g_0(z):=y\cdot\nabla w(x-z).$$
 We split it as 
 $$\langle\nu^{\hbar_n}(t),g_{\hbar_n}\rangle-\langle\nu_t,g_0\rangle=\langle\nu^{\hbar_n}(t),g_{\hbar_n}-g_0\rangle-\langle\nu^{\hbar_n}(t)-\nu(t),g_0\rangle.$$
 Since $\nabla w\in L^2\cap L^{1+d/2}$, it is clear that $g_h\to g_0$ strongly as $\hbar\to0$ in $L^2\cap L^{1+d/2}$, and since $\nu^{\hbar_n}(t)\to \nu(t)$ weakly in $L^2+L^{\min(1+2/d),2}$, we conclude that 
 $$\langle\nu^{\hbar_n}(t),g_{\hbar_n}-g_0\rangle\to0,\ \langle\nu^{\hbar_n}(t)-\nu(t),g_0\rangle\to0$$
 as $n\to\ii$, finishing the proof.
\end{proof}

\begin{remark}\label{rk:uniqueness}
 Under the additional assumption that $W^\hbar_{\gamma(0)}$ converges as $\hbar\to0$ in $\cD'_{x,v}$, we will show in Corollary \ref{coro:ident-density} and Theorem \ref{thm:uniqueness} that the pair $(W,\nu)$ in Proposition \ref{prop:sc-limit-wigner} is the \emph{same} for all sequences $\hbar_n\to0$. In this case, the conclusions of Proposition \ref{prop:sc-limit-wigner} hold as $\hbar\to0$, and not only up to subsequence (thus recovering the statement of Theorem \ref{thm:limit-sc}).
\hfill$\diamond$\end{remark}

\subsection{A tightness condition allowing to identify the limiting density}\label{sec:ident-dens}

We now introduce the well-known notion of Husimi (or coherent states) transform, and use it to prove that the phase-space distribution takes values between $0$ and $1$ and to give a criterion on the identification of the density $\nu(t)$ in Proposition \ref{prop:sc-limit-wigner}. We first recall the definition and properties of Husimi transforms. 

\begin{definition}[Husimi transform]\label{def:husimi}
 Let $\chi\in L^2(\R^d,\R)$ with $\int_{\R^d}\chi^2=1$. For any $\hbar>0$ and for any $(x,v)\in\R^d\times\R^d$, define the coherent state 
 $$\chi_{x,v}^\hbar(z):=\chi^\hbar(z-x)e^{iv\cdot z/\hbar},\ \chi^\hbar(z):=\hbar^{-d/4}\chi(\hbar^{-1/2}z),\ z\in\R^d.$$
 Let $\gamma$ a bounded operator on $L^2(\R^d)$. Its Husimi transform $m_\gamma^\hbar$ is defined by 
 $$m^\hbar_\gamma(x,v):=\langle \chi^\hbar_{x,v},\gamma \chi^\hbar_{x,v}\rangle,\ (x,v)\in\R^d\times\R^d.$$
\end{definition}

\begin{remark}
 When $\gamma_{\rm ref}=g(-i\hbar\nabla)$ for some $0\le g\le 1$, we have
 $$m_{\gamma_{\rm ref}}^\hbar(x,v)=g*_v \hbar^{-d/2}|\hat{\chi}|^2(\cdot/\sqrt{\hbar}),\ (x,v)\in\R^d\times\R^d.$$
\hfill$\diamond$\end{remark}

\begin{lemma}[Relation between Wigner and Husimi]\label{lem:husimi-wigner}
 Assume that $\chi\in\cS_z(\R^d)$. Then, for any $\hbar>0$, for any bounded operator $\gamma$ on $L^2(\R^d)$, and for any $\phi\in C^\ii_0(\R^d\times\R^d)$, we have
 $$\langle m_\gamma^\hbar,\phi\rangle=\left\langle W^\hbar_\gamma,\phi * \hbar^{-d}W^1_{|\chi\rangle\langle\chi|}(\cdot/\sqrt{\hbar},\cdot/\sqrt{\hbar})\right\rangle.$$
\end{lemma}

\begin{remark}
 The right side of the previous equation makes sense since $\chi\in\cS_z$ implies $W^1_{|\chi\rangle\langle\chi|}\in\cS_{x,v}$ and thus $\phi * \hbar^{-d}W^1_{|\chi\rangle\langle\chi|}(\cdot/\sqrt{\hbar},\cdot/\sqrt{\hbar})\in\cS_{x,v}$ for all $\hbar>0$ and all $\phi\in\cS_{x,v}$.
\hfill$\diamond$\end{remark}

\begin{corollary}\label{coro:limit-husimi}
 Let $\hbar_n\to0$ be the sequence given by Proposition \ref{prop:sc-limit-wigner}. Then, we have $m^{\hbar_n}_{\gamma(t)}\to W(t)$ in $\cD'_{x,v}$ for all $t\in\R$ and in particular, $0\le W(t)\le 1$ for all $t\in\R$. 
\end{corollary}

\begin{proof}
 We know that $W^{\hbar_n}_{\gamma(t)}-m_{\rm ref}\to W(t)-m_{\rm ref}$ as $n\to\ii$, weakly in $L^2_{x,v}$ for a.e. $t\in\R$. Let $\phi\in C^\ii_0(\R^d\times\R^d)$. Now take $\chi\in\cS_z(\R^d)$. By Lemma \ref{lem:husimi-wigner}, we have for all $n$
 \begin{multline*}
    \langle m_{\gamma(t)}^{\hbar_n},\phi\rangle=\left\langle m_{\rm ref},\phi * \hbar_n^{-d}W^1_{|\chi\rangle\langle\chi|}(\cdot/\sqrt{\hbar_n},\cdot/\sqrt{\hbar_n})\right\rangle\\
    +\left\langle W^{\hbar_n}_{\gamma(t)}-m_{\rm ref},\phi * \hbar_n^{-d}W^1_{|\chi\rangle\langle\chi|}(\cdot/\sqrt{\hbar_n},\cdot/\sqrt{\hbar_n})\right\rangle.
 \end{multline*}
 Since $\chi\in\cS_z$, we have $W^1_{|\chi\rangle\langle\chi|}\in L^1_{x,v}$ and thus $(\hbar^{-d}W^1_{|\chi\rangle\langle\chi|}(\cdot/\sqrt{\hbar},\cdot/\sqrt{\hbar}))_{\hbar>0}$ is an approximation of the identity (since $\int W^1_{|\chi\rangle\langle\chi|}\,dx\,dv=\int\chi^2=1$). As a consequence, 
 $$\phi * \hbar_n^{-d}W^1_{|\chi\rangle\langle\chi|}(\cdot/\sqrt{\hbar_n},\cdot/\sqrt{\hbar_n})\to\phi$$
 as $n\to\ii$ in $L^2_{x,v}\cap L^1_{x,v}$, from which we infer that 
 $$\langle m_{\gamma(t)}^{\hbar_n},\phi\rangle\to \langle W(t),\phi\rangle.$$
 Since $0\le m_{\gamma(t)}^{\hbar_n}\le 1$, we deduce that $0\le W\le 1$ as well. 
\end{proof}

\begin{lemma}\label{lem:density-husimi}
 Let $\gamma$ be a non-negative, bounded, locally trace-class operator on $L^2(\R^d)$ such that $\rho_\gamma\in L^1_{\rm loc}(\R^d)$ is also in $\cS'(\R^d)$. Let $\hbar>0$ and $\chi\in\cS_z(\R^d)$. Then $\int_{\R^d}m^\hbar_\gamma(\cdot,v)\,dv$ is a non-negative Radon measure on $\R^d$ for all $\hbar>0$, and we have
 $$\rho_{m^\hbar_\gamma}=(2\pi)^{-d}\int_{\R^d}m^\hbar_\gamma(\cdot,v)\,dv=\hbar^d\rho_\gamma * \hbar^{-d/2}\chi^2(-\cdot/\sqrt{\hbar}).$$
\end{lemma}

\begin{proof}
 For all $(x,v)\in\R^d\times\R^d$ and all $\hbar>0$, we have 
 $$m_\gamma^\hbar(x,v)=\tr\sqrt{\gamma}|\chi_{x,v}^\hbar\rangle\langle\chi_{x,v}^\hbar|\sqrt{\gamma}=\sum_j\|\chi_{x,v}^\hbar\sqrt{\gamma}\phi_j\|_{L^2}^2,$$
 where $(\phi_j)$ is any orthonormal basis of $L^2(\R^d)$. Next, we can use that for any $x\in\R^d$ and any $\phi\in L^2(\R^d)$ we have 
 $$\int_{\R^d}\|\chi_{x,v}^\hbar\phi\|_{L^2}^2\,dv = (2\pi\hbar)^d \int_{\R^d} |\phi(z)|^2(\chi^\hbar)^2(z-x)\,dx$$
 to infer that for any $x\in\R^d$,
 $$\int_{\R^d}m_\gamma^\hbar(x,v)\,dv = (2\pi\hbar)^d\tr \sqrt{\gamma} (\chi^\hbar)^2(\cdot-x)\sqrt{\gamma} = (2\pi\hbar)^d(\rho_\gamma * (\chi^\hbar)^2(-\cdot))(x).$$ 
\end{proof}

\begin{corollary}\label{coro:limit-density-husimi}
 Let $\hbar_n\to0$ the sequence given by Proposition \ref{prop:sc-limit-wigner}. Then, we have 
 $$\rho_{m^{\hbar_n}_{\gamma(t)}}\to\rho_{m_{\rm ref}}+\nu(t)$$
 as $n\to\ii$ in $\cD'_x$, for a.e. $t\in\R$.
\end{corollary}

\begin{proof}
 We know that 
 $\hbar_n^d\rho_{\gamma(t)}-\rho_{m_{\rm ref}}\to\nu(t)$ as $n\to\ii$ weakly-$*$ in $L^2+L^{\min(1+2/d,2)}$, for a.e. $t\in\R$. Let $\phi\in C^\ii_0(\R^d)$ and $\chi\in\cS_z(\R^d)$. Since $\gamma(t)\in\cK_S^\hbar$, we deduce that $\gamma(t)$ is a non-negative, bounded, operator with $\rho_{\gamma(t)}\in\cS'_x$ by the Lieb-Thirring inequality. By Lemma \ref{lem:density-husimi}, we deduce that for all $n$ and for a.e. $t\in\R$,
 \begin{multline*}
    \langle\rho_{m^{\hbar_n}_{\gamma(t)}},\phi\rangle\\
    =\langle( \hbar_n^d\rho_{\gamma(t)}-\rho_{m_{\rm ref}}), \phi * \hbar_n^{-d/2}\chi^2(-\cdot/\sqrt{\hbar_n})\rangle
    +\langle\rho_{m_{\rm ref}}, \phi * \hbar_n^{-d/2}\chi^2(-\cdot/\sqrt{\hbar_n})\rangle.
 \end{multline*}
 Since $\int\chi^2=1$, $(\hbar^{-d/2}\chi^2(-\cdot/\sqrt{\hbar}))_{\hbar>0}$ is an approximation of the identity and thus
 $$\phi * \hbar_n^{-d/2}\chi^2(\cdot/\sqrt{\hbar_n})\to\phi$$
 as $n\to\ii$ in $L^1\cap L^{1+d/2}$. This gives the result.
\end{proof}

The following provides a tightness condition in velocity which allows to pass to the limit for the density. 

\begin{proposition}[Identifying the density]\label{prop:ident-density}
 Let $(m_n)$ a sequence of non-negative Radon measures on $\R^d_x\times\R^d_v$ such that $\rho_{m_n}$ is a Radon measure on $\R^d_x$ for all $n$. Assume that $m_n\to m_\ii$ as $n\to\ii$ in $\cD'_{x,v}$ and that $\rho_{m_n}\to\rho_\ii$ in $\cD'_x$ as $n\to\ii$. Assume furthermore that for all non-negative $\phi\in C^\ii_0(\R^d)$ we have 
 $$\lim_{R\to\ii}\limsup_{n\to\ii}\int_{\R^d\times\R^d}\1(|v|\ge R)\phi(x)dm_n(x,v)=0.$$
 Then, we have 
 $$\rho_{m_\ii}=\rho_\ii.$$
\end{proposition}

\begin{proof}
 Let $\phi\in C^\ii_0(\R^d)$, $\phi\ge0$. Let $\chi:\R^d\to[0,1]$ be a smooth, radially decreasing function such that $\chi(x)\equiv1$ for $|x|\le 1$ and $\chi(x)\equiv0$ for $|x|\ge2$. Let $R>0$. Then, for all $n$ we have 
 \begin{align*}
    &(2\pi)^d\int_{\R^d}\phi(x)d\rho_{m_n}(x) \\
    &= \int_{\R^d\times\R^d}\phi(x)dm_n(x,v) \\
    &= \int_{\R^d\times\R^d}\chi(v/R)\phi(x)dm_n(x,v)+\int_{\R^d\times\R^d}(1-\chi(v/R))\phi(x)dm_n(x,v).
 \end{align*}
 Taking the limsup as $n\to\ii$ with $R$ fixed, we find that 
 $$(2\pi)^d\int_{\R^d}\phi(x)d\rho_\ii(x) = \int_{\R^d\times\R^d}\chi(v/R)\phi(x)dm_\ii(x,v)+o_{R\to\ii}(1).$$
 Taking now the limit $R\to\ii$, we find by monotone convergence that
 $$(2\pi)^d\int_{\R^d}\phi(x)d\rho_\ii(x)=\int_{\R^d\times\R^d}\phi(x)dm_\ii(x,v),$$
 which is the desired result.
\end{proof}

Proposition \ref{prop:ident-density} shows that in order to identify the density $\nu(t)$ in Proposition \ref{prop:sc-limit-wigner}, we have to control the high velocity part of the Husimi transform $m^\hbar_{\gamma(t)}$, uniformly in $\hbar$. We obtain such control in two steps: (i) by controlling the classical relative entropy of the Husimi transform by the quantum relative entropy in the next section, and (ii) by showing that the classical relative entropy controls high velocity, in Section \ref{sec:high-freq}.

\section{Berezin-Lieb inequalities}\label{sec:BL}

The uniform control in $\hbar$ of the relative entropy given by Proposition \ref{prop:sc-bounds} will allow us to control the classical relative entropy of the Husimi transforms (as done in Proposition \ref{prop:BL}). The main tool to control classical by quantum objects are \emph{Berezin-Lieb inequalities}~\cite{Simon-81,Berezin-72,Lieb-73}.  

\begin{definition}[Berezin-Lieb property]\label{def:BL}
Let $I\subset\R$ be an interval and $S$ an admissible entropy functional on $I$ (meaning that it satisfies Definition \ref{def:entropy} with $[0,1]$ replaced by $I$). We say that $S$ satisfies the \emph{Berezin-Lieb property on $I$ (with constant $C_S^{BL}>0$)}  if for any resolution of the identity of the form 
$$\int_{\R^N}|n_x\rangle\langle n_x|\,d\zeta(x)=\text{id}_\gH$$
on any Hilbert space $\gH$ where $d\zeta$ is a non-negative Borel measure on $\R^N$ and $n:x\in\R^N\to n_x\in\{u\in\gH,\ \|u\|=1\}$ is continuous, we have 
\begin{equation}
\int_{\R^N}\cH_S\big(\langle n_x,An_x\rangle,\langle n_x Bn_x\rangle\big)\,d\zeta(x)\le C_S^{BL}\;\cH_S(A,B), 
 \label{eq:def_Berezin-Lieb}
\end{equation}
for any bounded self-adjoint operators $A,B$ on $\gH$ with spectrum in $I$.
\end{definition}

We recall that $\cH_S(y,y_0)$ was defined for numbers $y$, $y_0$ in Definition \ref{def:class-rel-ent}. We call~\eqref{eq:def_Berezin-Lieb} a \emph{Berezin-Lieb inequality}.

\begin{remark}
 Since the Berezin-Lieb inequality is an equality if $\dim\gH=1$, we always have $C_S^{BL}\ge1$ so $C_S^{BL}=1$ is the best case. Only in the bosonic case $S=S_b$ we will have $C_S^{BL}>1$ (it may not be sharp, however).
\hfill$\diamond$\end{remark}

In the context of coherent states, the Berezin-Lieb property will be used to obtain the following bound on the control of the classical entropy of the Husimi transform by the quantum entropy.

\begin{proposition}\label{prop:BL}
 Let $S$ an admissible entropy functional having the Berezin-Lieb property on $[0,1]$ with constant $C^{BL}_S$, and define $\gamma_{\rm ref}=(S')^{-1}(-\hbar^2\Delta)$ with $\hbar>0$. Let $\gamma\in\cK_S^\hbar$. Defining the Husimi transform $m^\hbar_\gamma$ as in Definition \ref{def:husimi}, we then have the Berezin-Lieb inequality
 $$\hcl_S(m^\hbar_\gamma,m^\hbar_{\gamma_{\rm ref}})\le C_S^{BL}\hbar^d\cH_S(\gamma,\gamma_{\rm ref}).$$
\end{proposition}

\begin{proof}
 Use the resolution of the identity
 $$\int_{\R^d}\int_{\R^d}|\chi^\hbar_{x,v}\rangle\langle\chi^\hbar_{x,v}|\,\frac{dx\,dv}{(2\pi\hbar)^d}=1$$
 together with the definition of the Berezin-Lieb inequality.
\end{proof}

The rest of this section is devoted to the proof that the three physical entropies in the class $\ccS$ defined in \eqref{eq:entropies-berezin} have the Berezin-Lieb property on $[0,1]$, so that we may apply Proposition \ref{prop:BL}. 

First, the Berezin-Lieb property is stable under the following operations.

\begin{lemma}
 \begin{enumerate}
  \item If $S$ has the Berezin-Lieb property on $I$, then for any $t\in\R$, $S(t+\cdot)$ has the Berezin-Lieb property on $I-t$.
  \item If $S$ has the Berezin-Lieb property on $[0,1]$, then for any $0\le a,b\le 1$, the function $S_{a,b}(x)=S(ax+(1-x)b)$ has the Berezin-Lieb property on $[0,1]$.
  \item If $S_1$ and $S_2$ have the Berezin-Lieb property on $I$ and if $\mu_1,\mu_2\ge0$, then $\mu_1S_1+\mu_2S_2$ has the Berezin-Lieb property on $I$.
  \item If there exist $\alpha,\beta>0$ so that $S_2-S_1/\alpha$ and $\beta\,S_1-S_2$ are both concave on $I$ and if $S_1$ has the Berezin-Lieb property with constant $C_1$, then so does $S_2$ with constant $C_2=C_1\alpha\beta$. 
  \end{enumerate} 
 \label{lem:BL}
\end{lemma}

\begin{proof}
For \textit{(4)} note that the concavity of the two functions implies $\cH_{S_2}\geq \cH_{S_1}/\alpha$ and $\cH_{S_2}\leq \beta \cH_{S_1}$. 
\end{proof}

For $S=S_0$ the Berezin-Lieb property follows from the monotonicity of the Umegaki relative entropy under completely monotone trace-preserving maps~\cite{LieRus-73b,OhyPet-93}. This was proved in~\cite[Lemma 7.2]{LewNamRou-15} under the additional assumption that $\tr(A)=\tr(B)=1$. The following result shows that the trace normalization is indeed not necessary.

\begin{proposition}\label{prop:BL-umegaki}
 The function $S_0(x)=-x\log(x)+x$ has the Berezin-Lieb property on $[0,+\ii)$ with constant $1$. 
\end{proposition}

\begin{proof}
\emph{Step 1: Discrete Berezin-Lieb property.}
Assume that $\gH$ is finite-dimensional, and that $(X_k)_{k=1}^K$ is a family of linear maps on $\gH$, satisfying 
$$\sum_{k=1}^K \alpha_k\,X_k^* X_k=\text{id}_{\gH}$$
for some $\alpha_k>0$.
The linear map 
\begin{equation}
\Phi:A\mapsto \begin{pmatrix}
 \alpha_1 X_1AX_1^*&&0\\
 &\ddots&\\
0 &&\alpha_K X_KAX_K^*
 \end{pmatrix} 
 \label{eq:CPTP}
\end{equation}
from $\cL(\gH)$ to $\cL(\gH^K)$  is completely monotone and trace-preserving~\cite{OhyPet-93}. It is a classical fact that $\cH_{S_0}$ is monotone under such maps, that is,  $\cH_{S_0}(A,B)\geq \cH_{S_0}(\Phi(A),\Phi(B))$ for all $A=A^*,B=B^*\geq0$ satisfying the additional condition that $\tr(A)=\tr(B)=1$~\cite{OhyPet-93,Carlen-10}. To see that the trace normalization is not necessary, we remark that for general $A=A^*,B=B^*\geq0$
$$\cH_{S_0}(A,B)=\tr(A)\;\cH_{S_0}\left(\frac{A}{\tr(A)},\frac{B}{\tr(B)}\right)+\cH_{S_0}\Big(\tr(A)\,,\,\tr(B)\Big).$$
Since $\tr(\Phi(A))=\tr(A)$ and $\tr(\Phi(B))=\tr(B)$ it is then clear that 
$$\cH_{S_0}(A,B)\geq \cH_{S_0}(\Phi(A),\Phi(B))$$ 
for all non-negative hermitian matrices $A,B$, without any trace condition. Using the homogeneity of $\cH_{S_0}$, this gives for our $\Phi$ in~\eqref{eq:CPTP}
\begin{align*}
\cH_{S_0}(A,B)&\geq \sum_{k=1}^K\cH_{S_0}(\alpha_kX_kAX_k^*,\alpha_kX_kBX_k^*)\\
&=\sum_{k=1}^K\alpha_k\;\cH_{S_0}(X_kAX_k^*,X_kBX_k^*). 
\end{align*}
If $\sum_{k=1}^K \alpha_k\,X_k^* X_k\leq\text{id}_{\gH}$ then we may introduce $Y=Y^*$ so that $Y^2=\text{id}_{\gH}-\sum_{k=1}^K \alpha_k\,X_k^* X_k$ and, using that $\cH_{S_0}(YAY,YBY)\geq0$, we obtain 
\begin{equation}
 \sum_{k=1}^K\alpha_k\;\cH_{S_0}(X_kAX_k^*,X_kBX_k^*)\leq \cH_{S_0}(A,B)
 \label{eq:BL-discrete}
\end{equation}
which is a kind of discrete version of~\eqref{eq:def_Berezin-Lieb}. We can even allow $K=+\ii$ after passing to the limit in the inequality at finite $K$. 

 \medskip
 
\noindent \emph{Step 2: Berezin-Lieb property in finite-dimensions.}
 We next show the Berezin-Lieb property when $\gH$ is finite-dimensional, deducing it from the discrete version \eqref{eq:BL-discrete} (approximating integrals by sums). To do so, let $n\in\N$ and for any $k\in(1/n)\Z^N$, let 
 $$Q_k^{(n)}:=k+(1/n)[0,1)^N.$$
 We define the operator
 $$J_n:=\sum_{k\in(1/n)\Z^N}\zeta(Q^{(n)}_k)|n_k\rangle\langle n_k|.$$
 Note that $\langle v,J_n w\rangle\to\langle v,w\rangle$ as $n\to\ii$ for any $x,y\in\gH$ by continuity of $x\mapsto n_x$. Since $\gH$ is finite-dimensional, we deduce that $\|J_n\|\to1$ as $n\to\ii$.  In particular, we have
 $$\sum_{k\in(1/n)\Z^N}\frac{\zeta(Q^{(n)}_k)}{\|J_n\|}|n_k\rangle\langle n_k|\le 1$$
 and we may apply \eqref{eq:BL-discrete} with
 $$X_k:=\langle n_k|,\quad \alpha_k:=\frac{\zeta(Q^{(n)}_k)}{\|J_n\|}.$$
We obtain 
\begin{align*}
\cH_S(A,B)&\ge\sum_{k\in(1/n)\Z^N}\frac{\zeta(Q^{(n)}_k)}{\|J_n\|}\cH_S(\langle n_k,An_k\rangle,\langle n_k,Bn_k\rangle)\\
&=\frac{1}{\|J_n\|}\int_{\R^N}\cH_S\big(a_n(x),b_n(x)\big)\,d\zeta(x) 
\end{align*}
where 
$$a_n(x)= \sum_{k\in(1/n)\Z^N}\1(x\in Q_k^{(n)})\langle n_k, An_k\rangle$$
and
$$b_n(x)= \sum_{k\in(1/n)\Z^N}\1(x\in Q_k^{(n)})\langle n_k, Bn_k\rangle $$
are the piecewise constant approximations to $a(x)=\langle n_x, An_x\rangle$ and $b(x)=\langle n_x, Bn_x\rangle$.
 Since the map $x\mapsto n_x$ is continuous, the functions $a_n$ and $b_n$ converge pointwise to $a$ and $b$. By lower-semi-continuity of the relative entropy, we obtain that
 $$\liminf_{n\to\ii}\cH_S(a_n(x),b_n(x))\ge \cH_S(a(x),b(x))$$
 for all $x$, and we deduce the result after passing to the limit $n\to\ii$ by Fatou's lemma.
 
 \medskip
 
\noindent \emph{Step 3: Extension to infinite-dimensional spaces}. Finally, we extend the Berezin-Lieb property to infinite-dimensional $\gH$. To do so, let $(P_k)$ a sequence of finite-rank orthogonal projections on $\gH$, such that $P_k\to1$ strongly in $\gH$ as $k\to\ii$. Then, by definition of $\cH_S$ \cite{LewSab-13}, we have
 $$\cH_S(A,B)=\lim_{k\to\ii}\cH_S(P_kAP_k,P_kBP_k).$$
 Defining the finite dimensional space $\gH_k=P_k\gH$, we have
 $$\int_{\R^N}|n_x^{(k)}\rangle\langle n_x^{(k)}|\,d\zeta_k(x)=\text{id}_{\gH_k},$$
 where 
 $$n_x^{(k)}:=\frac{P_k n_x}{\|P_k n_x\|},\quad d\zeta_k(x)=\|P_k n_x\|^2\,d\zeta(x).$$
 Using Step 2 on $\gH_k$, we deduce that for all $k$,
 \begin{multline*}
    \cH_S(P_k A P_k,P_k B P_k)\\
    \ge\int_{\R^N}\cH_S\left(\frac{\langle n_x,P_k A P_k n_x\rangle}{\|P_k n_x\|^2},\frac{\langle n_x,P_k A P_k n_x\rangle}{\|P_k n_x\|^2}\right)\,\|P_kn_x\|^2\,d\zeta(x).
 \end{multline*}
 Using that for all $x$,
 $$\|P_kn_x\|\to1,\quad \langle n_x,P_kAP_kn_x\rangle\to\langle n_x,An_x\rangle$$
 as $k\to\ii$, we deduce the Berezin-Lieb inequality by Fatou's lemma.
 \end{proof} 

\begin{corollary}[Bosonic and fermionic cases]
\ 
\begin{enumerate}
 \item The function $S_f(x)=-x\log(x)-(1-x)\log(1-x)$ has the Berezin-Lieb property on $[0,1]$ with constant $1$.
 \item For any $M>0$, the function $S_b(x)=-x\log(x)+(1+x)\log(1+x)$ has the Berezin-Lieb property on $[0,M]$ with constant $1+M$.
\end{enumerate}
\label{coro:BL-boson-fermion}
\end{corollary}

\begin{proof}
The fermionic case $S=S_f$ is an immediate consequence of \textit{(1)--(3)} in Lemma \ref{lem:BL} and of Proposition \ref{prop:BL-umegaki}. The bosonic case $S=S_b$ is not as easy since the term $(1+x)\log(1+x)$ comes with the wrong sign. Here we instead use the last point \textit{(4)} of Lemma~\ref{lem:BL}, noticing that $-x\log(x)-S_b(x)=-(1+x)\log(1+x)$ is concave on $\R_+$ and that $S_b+(1+M)^{-1} x\log(x)$ is concave on $[0,M]$. 
\end{proof}

\begin{lemma}
 The function $S(x)=-x^2$ has the Berezin-Lieb property on $\R$ with constant $1$.
\end{lemma}

\begin{proof}
 For this $S$, we have $\cH_S(x,y)=(x-y)^2$. As a consequence, 
 \begin{align*}
  \int_{\R^N}\cH_S\big(\langle n_x,An_x\rangle,\langle n_x Bn_x\rangle\big)\,d\zeta(x) &= \int_{\R^N}\langle n_x,(A-B)n_x\rangle^2\,d\zeta(x)\\
  &\le \int_{\R^N}\langle n_x,(A-B)^2n_x\rangle\,d\zeta(x)\\
  &= \tr(A-B)^2=\cH_S(A,B).
 \end{align*}
\end{proof}

\begin{remark}\label{rk:BL-family}
 We deduce that functions of the type
 $$S(x)=-ax^2+bx+c-\int_0^\ii(t+x)\log(t+x)\,d\nu(t),$$
 for $a\ge0$, $b,c\in\R$ and $d\nu(t)$ a non-negative measure on $[0,+\ii)$ (with appropriate integrability conditions) satisfy a Berezin-Lieb inequality on $\R_+$. It turns out that these are exactly the functions $S$ which are concave and such that $S''$ is numerically non-decreasing and operator concave \cite{PitVir-15}.
\hfill$\diamond$\end{remark}

\begin{remark}\label{rmk:BL}
 A natural open question, which is outside of the scope of this paper, is to determine the class of entropies satisfying the Berezin-Lieb property, on the interval $[0,1]$. 
 
 Note that it is very important that the definition involves a resolution of the identity with only rank-one projections $|n_x\rangle\langle n_x|$. One could consider more general (sub-)resolutions of the identity of the type 
 $$\int_{\R^N}X(x)^* X(x)\,d\zeta(x)\le 1,$$
 and ask the stronger property
 \begin{equation}
\cH_S(A,B)\ge\int_{\R^N}\cH_S\big(X(x)AX(x)^*,X(x) B X(x)^*\big)\,d\zeta(x).
\label{eq:stronger-BL}
 \end{equation}
This is indeed satisfied by the Umegaki entropy $S_0(x)=-x\log(x)+x$, as is seen from the proof of Proposition~\ref{prop:BL-umegaki}. But we claim that this is the only entropy for which~\eqref{eq:stronger-BL} holds in $\R_+$, up to a multiplicative factor and an additive affine function (which of course does not matter in $\cH_S$). So requiring~\eqref{eq:stronger-BL} is not a good idea.

To prove the claim, take any admissible function $S$ satisfying~\eqref{eq:stronger-BL}. Using the trivial ``resolution of the identity'' $t\times(1/t)=1$ for all $t>0$, we deduce that 
 $\cH_S(tx,ty)\le t\cH_S(x,y)$.
Replacing $x$ and $y$ by $x/t$ and $y/t$ and then $t$ by $1/t$, we deduce that the reverse inequality holds, hence we obtain that $\cH_S$ is homogeneous:
 \begin{equation}
 \cH_S(tx,ty)= t\cH_S(x,y),\qquad \forall x,y,t>0.
  \label{eq:S_homogeneous}
 \end{equation}
The stronger property~\eqref{eq:stronger-BL} turns out to imply monotonicity in the sense of~\cite{LewSab-13}, hence that $-S'$ is operator monotone on $(0,\ii)$ by~\cite[Thm.~1]{LewSab-13}. In particular, $S$ is $C^\ii$ on $(0,\ii)$. Differentiating~\eqref{eq:S_homogeneous} with respect to $y$, we find $tS''(ty)=S''(y)$, hence $S''(y)=-C/y$. This proves, as we claimed, that $S(x)=Cx(1-\log(x))+ax+b$.
 \hfill$\diamond$\end{remark}

\section{High-velocity estimates}\label{sec:high-freq}

In this section, we prove that the classical relative entropy controls high velocities, similarly to the ``zero temperature'' case where the entropy degenerates to $\int|v|^2 f$. This control on high velocities is then used to identify the density in the Vlasov equation, as shown by Proposition \ref{prop:ident-density}. The main result of this section is the following.

\begin{theorem}[High velocity control by the relative entropy]\label{thm:high-freq}
  Let $S$ be an admissible entropy functional and define $m_{\rm ref}(x,v):=(S')^{-1}(|v|^2)$. Let $\hat{\chi}\in C^\ii_c(\R^d)$ be such that $\int|\hat{\chi}|^2=1$, and let $R>0$ such that $\supp\hat{\chi}\subset B(0,R)$. For any $\hbar>0$, define
  $$\tilde{m_{\rm ref}}(x,v):=m_{\rm ref}*_v \hbar^{-d/2}|\hat{\chi}|^2(\cdot/\sqrt{\hbar}),$$
  Let $m:\R^d\times\R^d\to[0,1]$ be any function, and let $\phi\in (L^1\cap L^\ii)(\R^d)$. Then, there exists $C>0$ such that for all $A>32R^2$ and for all $\hbar\in(0,1)$ we have 
  \begin{multline*}
    \int_{\R^d\times\R^d}|\phi(x)|\1(|v|^2\ge 4A)m(x,v)\,dx\,dv\\
    \le C\|\phi\|_{L^\ii}\frac{\hcl_S(m,\tilde{m_{\rm ref}})}{A}+C\|\phi\|_{L^1}\int_{|v|^2\ge A/32}m_{\rm ref}\,dv.
  \end{multline*}
 \end{theorem}
 
 This result indeed shows that a uniform control over $\hcl_S(m,\tilde{m_{\rm ref}})$ implies a control for high $|v|$ which is uniform in $\hbar$. Combined with Proposition \ref{prop:BL}, it implies the following result, part of the proof of Theorem \ref{thm:limit-sc}. 

 \begin{corollary}\label{coro:ident-density}
  Assume that $S\in\ccS$. Then, there exists $C_S>0$ such that the pair $(W,\nu)$ built in Proposition \ref{prop:sc-limit-wigner} satisfies
  \begin{equation}\label{eq:control-entropy}
    \hcl_S(W(t),m_{\rm ref})\le C_S\liminf_{\hbar\to0}\hbar^d\cH_S(\gamma(t),\gamma_{\rm ref})
  \end{equation}
  and 
  $$\nu(t)=\rho_{W(t)}-\rho_{m_{\rm ref}}$$
  for all $t\in\R$.
 \end{corollary}
 
 \begin{proof}
  Any $S\in\ccS$ has the Berezin-Lieb property (on $[0,1]$) by Proposition \ref{prop:BL-umegaki} and Corollary \ref{coro:BL-boson-fermion}. By Proposition \ref{prop:BL}, we deduce that for any $\hbar>0$ and any $t\in\R$ we have 
  \begin{equation}\label{eq:control-entropy-2}
    \hcl_S(m^\hbar_{\gamma(t)},m^\hbar_{\gamma_{\rm ref}})\le C^{BL}_S\hbar^d\cH_S(\gamma(t),\gamma_{\rm ref}).
  \end{equation}
  To go from \eqref{eq:control-entropy-2} to \eqref{eq:control-entropy}, recall that  $m_{\gamma(t)}^{\hbar_n}\to W(t)$ as $n\to\ii$ in $\cD'_{x,v}$ by 
  Corollary \ref{coro:limit-density-husimi}. Hence, \eqref{eq:control-entropy} follows if we know that the classical relative entropy decreases under weak limits, a result that we recall in Lemma \ref{lem:relative-entropy-weak-limit} below. This may introduce another multiplicative constant, as explained in Remark \ref{rk:weak-limit-bosons}. The identification of the density $\nu(t)$ then follows from Proposition \ref{prop:ident-density}, Proposition \ref{prop:BL}, Theorem \ref{thm:high-freq} and the uniform relative entropy bounds \eqref{eq:control-entropy-time}.
 \end{proof}
 
 \begin{lemma}[Decrease of the classical relative entropy under weak limits]
\label{lem:relative-entropy-weak-limit}
Let $S:[0,1]\to\R$ be a continuous, concave function, differentiable on $(0,1)$. Assume that the function $\cH_S:\R\times\R\to[0,+\ii]$ is convex. Let $N\ge1$ and $(a_n)$, $(b_n)$ two sequences of measurable functions from $\R^N$ to $\R$. Assume that for all $n$, $0\le a_n,b_n\le 1$ and that $\hcl_S(a_n,b_n)<+\ii$. Assume also that $a_n\to a$ and $b_n\to b$ in $\cD'(\R^N)$. Then, we have 
$$\liminf_{n\to+\ii}\hcl_S(a_n,b_n)\ge\hcl_S(a,b).$$
\end{lemma}

\begin{proof}
 Since $\cH_S$ is non-negative and lower semi-continuous, as the representation
 $$\cH_S(x,y)=\sup_{0<\theta<1}\frac1\theta\Big(S(\theta x + (1-\theta)y)-\theta S(x)-(1-\theta)S(y)\Big)$$
 shows, the statement is clear for poinwise limits by Fatou's lemma. The convexity of $\cH_S$ will allow us to go from weak convergence to pointwise convergence. Indeed, let $\phi\in C^\ii_0(\R^N)$ with $\phi\ge0$ and $\int\phi=1$. Let $\epsilon>0$ and define $\phi_\epsilon:=\epsilon^{-N}\phi(\cdot/\epsilon)$. Then, since $\cH_S$ is convex, we deduce that for all $\epsilon>0$,
 $$\hcl_S(a_n,b_n)\ge \hcl_S(a_n*\phi_\epsilon,b_n*\phi_\epsilon).$$
 Now since $a_n\to a$, $b_n\to b$ in $\cD'(\R^N)$ and since $\phi_\epsilon\in\cD(\R^N)$ for any $\epsilon>0$, we deduce that $a_n*\phi_\epsilon\to a*\phi_\epsilon$ and that $b_n*\phi_\epsilon\to b*\phi_\epsilon$ as $n\to\ii$, pointwise on $\R^N$. Applying Fatou's lemma, we deduce that for all $\epsilon>0$,
 $$\liminf_{n\to+\ii}\hcl_S(a_n,b_n)\ge\hcl_S(a*\phi_\epsilon,b*\phi_\epsilon).$$
 Now since $a_n\to a$ and $b_n\to b$ in $\cD'(\R^N)$ with the uniform bounds $0\le a_n,b_n\le 1$ for all $n$, we deduce that $0\le a,b\le 1$ as well. Consequently, we have $a*\phi_\epsilon\to a$ and $b*\phi_\epsilon\to b$ as $\epsilon\to0$, for instance in $L^1_{{\rm loc}}(\R^N)$ and thus almost everywhere up to a subsequence. Using again Fatou's lemma, we deduce the result.
\end{proof}

\begin{remark}\label{rk:weak-limit-bosons}
Notice that $\cH_{S_0}$ and $\cH_{S_f}$ are convex but $\cH_{S_b}$ is not. As shown in the proof of Corollary \ref{coro:BL-boson-fermion} we have $\cH_{S_0}/2\leq \cH_{S_b}\leq \cH_{S_0}$ on $[0,1]$
and therefore the bosonic relative entropy satisfies
 $$\frac{\hcl_{S_b}(a,b)}2\le \liminf_{n\to\ii}\hcl_{S_b}(a_n,b_n).$$
\hfill$\diamond$\end{remark}

 The rest of this section is devoted to the proof of Theorem \ref{thm:high-freq}. The first step is inspired by the proof of Theorem 8 in \cite{LewSab-13a}. 

\begin{proposition}\label{prop:control-freq}
 Let $S$ be an admissible entropy functional. Let $u,u_0:\R^d\to[0,1]$ be measurable functions with $u_0\in L^1(\R^d)$. Then, for any $A>0$ we have
 \begin{multline}\label{eq:control-high-freq}
  \left|\int_{\R^d}\1(S'(u_0(v))\ge A)(u(v)-u_0(v))\,dv\right|\le C\frac{\int_{\R^d}\cH_S(u(v),u_0(v))\,dv}{A}\\
  +C\int_{\R^d}\1(S'(u_0(v))\ge A))(S')^{-1}\left(\frac{S'(u_0(v))}{2}\right)\,dv.
 \end{multline}
\end{proposition}

\begin{proof}
 We write 
 $$u-u_0=\int_0^1(\1(u\ge\lambda)-\1(u_0\ge\lambda))\,d\lambda=\int_0^1 a_\mu\,d\lambda,$$
 where $a_\mu=\1(u\ge\lambda)-\1(u_0\ge\lambda)=\1(u\ge\lambda)-\1(S'(u_0)\le \mu)$ with $\mu=S'(\lambda)$. There is a similar representation for the relative entropy
 $$\cH_S(u,u_0)=\int_0^1(S'(u_0)-\mu)a_\mu\,d\lambda=\int_0^1|S'(u_0)-\mu|(a_\mu^+-a_\mu^-)\,d\lambda,$$
 where 
 $$a_\mu^+:=\1(S'(u_0)>\mu)a_\mu\ge0,\qquad a_\mu^-=\1(S'(u_0)\le\mu)a_\mu\le0.$$
 We now distinguish two cases. When $\mu\le A/2$ and $S'(u_0(v))\ge A$ we have
 $$S'(u_0(v))-\mu\ge A-\mu\ge \frac A2.$$
 Hence,
 $$\int_{\R^d}\1(S'(u_0(v))\ge A)a_\mu(v)\,dv\le\frac{2}{A}\int_{\R^d}(S'(u_0)-\mu)a_\mu\,dv,$$
 the left side being non-negative in this case since $S'(u_0)\ge A$ implies $u_0<\lambda$. When $\mu\ge A/2$, we further split
 $$a_\mu=a_\mu^{\ge2\mu}+a_\mu^{<2\mu},$$
 where $$a_\mu^{\ge2\mu}=a_\mu\1(S'(u_0(v))\ge2\mu).$$
 From the same estimate that we just did,
 \begin{multline*}
    \int_{\R^d}\1(S'(u_0(v))\ge A)a_\mu^{\ge 2\mu}(v)\,dv\\
    \le\frac{1}{\mu}\int_{\R^d}(S'(u_0)-\mu)a_\mu^{\ge2\mu}\,dv\le\frac{2}{A}\int_{\R^d}(S'(u_0)-\mu)a_\mu\,dv.
 \end{multline*}
 Finally, we estimate roughly the last term (which is the only one which does not have a sign a priori) using that $0\le u,u_0\le1$:
 $$\int_{\R^d}\1(S'(u_0(v))\ge A)|a_\mu^{<2\mu}|\,dv\le 2|\{ A\le S'(u_0(v))\le 2\mu\}|.$$
 Integrating over $\lambda$, we get the result.
 \end{proof}
 
 \begin{lemma}\label{lem:unif-h}
  Let $S$ be an admissible entropy functional and define $m_{\rm ref}(x,v)=(S')^{-1}(|v|^2)$. Let $\hat{\chi}\in C^\ii_0(\R^d)$ such that $\int|\hat{\chi}|^2=1$, and let $R>0$ such that $\supp\hat{\chi}\subset B(0,R)$. For any $\hbar>0$, define
  $$u_0:=m_{\rm ref}*\hbar^{-d/2}|\hat{\chi}|^2(\cdot/\sqrt{\hbar}).$$
  Then, for all $A>8R^2$ and for all $\hbar\in(0,1)$ we have 
  $$
    \int_{\R^d}\1(S'(u_0(v))\ge A))(S')^{-1}\left(\frac{S'(u_0(v))}{2}\right)\,dv\\
    \le2^{d+1}\int_{|v|^2\ge A/8}m_{\rm ref}\,dv.
  $$
 \end{lemma}
 
 \begin{proof}
  We first bound $u_0$ above and below, independently of $\hbar$. We have 
  $$u_0(v)=\int_{\R^d}(S')^{-1}(|v-\sqrt{\hbar}\eta|^2)|\hat{\chi}(\eta)|^2\,d\eta.$$
  Since for all $\eta\in\supp\hat{\chi}$ we have
  $$|v|-R\le|v-\sqrt{\hbar}\eta|\le|v|+R,$$
  and since $S'$ is decreasing, we deduce the bounds
  $$(S')^{-1}((|v|+R)^2)\le u_0(v)$$
  for all $v$ and 
  $$u_0(v)\le (S')^{-1}((|v|-R)^2)$$
  for all $|v|\ge R$. In the region where $u_0(v)\le (S')^{-1}(2|v|^2)$, we have
  $$(S')^{-1}\left(\frac{S'(u_0(v))}{2}\right)\le m_{\rm ref}$$
  and $S'(u_0(v))\ge A$ implies
  $$ (S')^{-1}((|v|+R)^2)\le u_0(v)\le (S')^{-1}(A)$$
  and thus $|v|+R\ge\sqrt{A}$ implying that $|v|^2\ge A/8$ if $A>8R^2$. As a consequence,
  \begin{multline*}
    \int_{u_0(v)\le (S')^{-1}(2|v|^2)}\1(S'(u_0(v))\ge A))(S')^{-1}\left(\frac{S'(u_0(v))}{2}\right)\,dv\\
    \le\int_{|v|^2\ge A/8}m_{\rm ref}\,dv.
  \end{multline*}
  In the region where $u_0(v)>(S')^{-1}(2|v|^2)$, $S'(u_0(v))\ge A$ implies that $|v|^2\ge A/2$, and thus $A>8R^2$ implies $R\le|v|/2$, so that 
  $$S'(u_0(v))\ge(|v|-R)^2\ge|v|^2/4.$$
  As a consequence,
  \begin{multline*}
    \int_{u_0(v)>(S')^{-1}(2|v|^2)}\1(S'(u_0(v))\ge A))(S')^{-1}\left(\frac{S'(u_0(v))}{2}\right)\,dv\\
    \le\int_{|v|^2\ge A/2}(S')^{-1}(|v|^2/4)\,dv=2^d\int_{|v|^2\ge A/8}m_{\rm ref}\,dv,
  \end{multline*}
  which implies the result.
 \end{proof}
 
 \begin{corollary}\label{coro:high-freq}
  Let $S$ be an admissible entropy functional and define $m_{\rm ref}(x,v)=(S')^{-1}(|v|^2)$. Let $\hat{\chi}\in C^\ii_0(\R^d)$ such that $\int|\hat{\chi}|^2=1$, and let $R>0$ such that $\supp\hat{\chi}\subset B(0,R)$. For any $\hbar>0$, define
  $$u_0:=m_{\rm ref}*\hbar^{-d/2}|\hat{\chi}|^2(\cdot/\sqrt{\hbar}).$$
  Then, there exists $C>0$ such that for all $u:\R^d\to[0,1]$, for all $A>32R^2$ and for all $\hbar\in(0,1)$ we have 
  $$
  \int_{|v|^2\ge4A}u(v)\,dv\le C\frac{\int_{\R^d}\cH_S(u(v),u_0(v))\,dv}{A}+C\int_{|v|^2\ge A/32}m_{\rm ref}\,dv.
  $$
 \end{corollary}
 
 \begin{proof}
  We have seen in the previous proof that for all $|v|\ge R$,
  $$(|v|+R)^2\ge S'(u_0(v))\ge(|v|-R)^2.$$
  Using furthermore that $A>8R^2$, we deduce 
  $$|v|^2\ge4A \Longrightarrow S'(u_0(v))\ge A \Longrightarrow |v|^2\ge A/8,$$
  since in this case we indeed have $|v|\ge R$. We then write
  \begin{align*}
   \int_{|v|^2\ge4A}u(v)\,dv & \le \int_{S'(u_0(v))\ge A}u(v)\,dv \\
   &\le \int_{S'(u_0(v))\ge A}(u(v)-u_0(v))\,dv+\int_{S'(u_0(v))\ge A}u_0(v)\,dv\\
   &\le \int_{S'(u_0(v))\ge A}(u(v)-u_0(v))\,dv+\int_{|v|^2\ge A/8}u_0(v)\,dv.
  \end{align*}
 Using again the  bound $u_0(v)\le(S')^{-1}((|v|-R)^2)$ and the fact that $A>32R^2$ and $|v|^2\ge A/8$ imply $R<|v|/2$, we deduce that 
 \begin{align*}
    \int_{\R^d}\1(|v|^2\ge A/8)u_0(v)\,dv &\le\int_{\R^d}\1(|v|^2\ge A/8)(S')^{-1}(|v|^2/4)\,dv\\
    &=2^d\int_{\R^d}\1(|v|^2\ge A/32)(S')^{-1}(|v|^2)\,dv
 \end{align*}
  The result then follows from combining Proposition \ref{prop:control-freq} and Lemma \ref{lem:unif-h}.
 \end{proof}

\begin{proof}[Proof of Theorem \ref{thm:high-freq}]
 Apply Corollary \ref{coro:high-freq} to $u(v)=m(x,v)$ for $x$ fixed, multiply by $\phi(x)$, and integrate over $x$.
\end{proof}

\section{Uniqueness of solutions to the Vlasov equation}\label{sec:uniqueness}

In the previous sections, we have proved that for any $\gamma_0\in\cK_S^\hbar$ such that $\hbar^d\cH_S(\gamma_0,\gamma_{\rm ref})$ is bounded as $\hbar\to0$, the Wigner transform $W_{\gamma(t)}^\hbar$ of the solution to the semi-classical Hartree equation \eqref{eq:hartree} converges in a weak sense as $\hbar\to0$ (up to a subsequence) to $W(t)\in C^0_t\cD'_{x,v}$ which satisfies
$$0\le W(t)\le 1,\ W-m_{\rm ref}\in L^\ii_t(\R,L^2_{x,v}),\ \hcl_S(W(t),m_{\rm ref})\in L^\ii_t(\R),$$
$$\rho_W-\rho_{m_{\rm ref}}\in L^\ii_t(\R,L^2_x+L^{\min(1+2/d,2)}_x),$$
and is a solution in the distributional sense to the nonlinear Vlasov equation 
\begin{equation}\label{eq:vlasov2}
\begin{cases}
\partial_t W+2v\cdot\nabla_xW-\nabla_x \left(w*\rho_W\right)\cdot\nabla_v W=0,\\
W_{|t=0}=W(0).
\end{cases}
\end{equation}
Our next goal is to show that such solutions are unique. To prove uniqueness, we follow the method of Loeper \cite{Loeper-06} based on optimal transportation. Since we deal here with measures of infinite mass, we need to extend the tools from optimal transportation to this case. First, we go from the Eulerian formulation \eqref{eq:vlasov} to the Lagrangian formulation in terms of solutions to the Newton equations of motions. To do so, we need the following equivalence between Eulerian and Lagrangian solutions:

\begin{lemma}[Equivalence between Eulerian and Lagrangian solutions]\label{lem:ode}
 Let $N\ge1$ and define the space
$$A(\R^N)=\{a\in C^1(\R^N,\R^N),\ D_x a\in L^\ii_x(\R^N)\},$$
endowed with norm
$$\|a\|_A:=\|D_x a\|_{L^\ii}+\|(1+|x|)^{-1} a \|_{L^\ii}.$$
Let $F:\R_t\times\R_x^N\to\R^N$ a locally integrable function such that 
$$\sup_{t\in\R}\|F(t)\|_A<+\ii,\qquad \Div_x F(t)=0\ \text{in}\ \cD'_{t,x}.$$
Then, there exists a unique continuous function $\Phi:\R_t\times\R_s\times\R^N_x\to\R^N_x$ such that for any $x\in\R^N$ and for any $s\in\R$, 
$$\Phi^t_s(x)=x+\int_s^t F(s',\Phi_s^{s'}(x))\,ds',\ \forall t\in\R.$$
Furthermore, for any $s,t\in\R$, $x\in\R^N\mapsto\Phi_s^t(x)$ is a $C^1$-diffeomorphism, bi-Lipschitz, which preserves the Lebesgue measure. 

Finally, for any $u_0\in L^\ii(\R^N)$,  $u_t:=(\Phi_0^t)_\sharp u_0$ is the unique distributional solution in $C^0_t(L^\ii_x,*)$ to the continuity equation
$$\partial_t u + F\cdot\nabla_x u=0$$
such that $u(0)=u_0$.
\end{lemma}

\begin{remark}
 We recall that $u\in C^0_t(L^\ii_x,*)$ means that for all $t\in\R$, $u(t)\in L^\ii_x(\R^N)$ and that for any $\phi\in L^1_x(\R^N)$, the map $t\mapsto\langle u(t),\phi\rangle$ is continuous.
\hfill$\diamond$\end{remark}

Lemma \ref{lem:ode} is standard and its proof is given in Appendix \ref{app:ode}. Applying it to 
$$F(t,x,v)=\left(2v,-\nabla_x \left(w*\rho_W\right)\right),$$
which belongs to $L^\ii_t(\R,A_x(\R^d\times\R^d))$ if 
$$D^2 w\in (L^2\cap L^{\max(2,1+d/2)})(\R^d),$$
and using that $W\in C^0_t\cD'_{x,v}$, $0\le W\le 1$ implies that $W\in C^0_t(L^\ii_{x,v},*)$, we deduce that 
$$W_t=(\Phi_0^t)_\sharp W_0,$$
where $\Phi_0^t$ is the (Newton) flow associated to $F_t$. We prove that such solutions are unique.

\begin{theorem}[Uniqueness of weak solutions to the positive density Vlasov equation]\label{thm:uniqueness}
 Assume that $d\ge1$ and let $w\in L^1(\R^d)$ such that $D^2 w\in (L^1\cap L^{\max(2,1+d/2)})(\R^d)$ and $\nabla w\in (L^1\cap L^{\max(d+2,4)/3})(\R^d)$. Let $\rho_0\ge0$. Then, there exists at most one non-negative solution $W\in C^0_t(L^\ii_{x,v},*)$ to \eqref{eq:vlasov2} such that $\rho_W\in\rho_0+L^\ii_t(L^2_x+L^{\min(1+2/d,2)}_x)$.
\end{theorem}

\begin{remark}
 Non-negativity is only assumed here to give a meaning to the density $\rho_W$.
\hfill$\diamond$\end{remark}

\begin{remark}
 Our uniqueness result is stated on the whole time interval $\R$, but it is also valid on any time interval containing the initial time.
\hfill$\diamond$\end{remark}

Theorem \ref{thm:uniqueness} is proved in several steps. As we said, it requires tools from optimal transportation with infinite mass. Since we have not found such tools in the literature, we present them in Appendix \ref{app:opt-transp}. They might have an independent interest. Before going into the proof of uniqueness, we provide a simple estimate on the difference of two solutions to the Newton equations.

Consider $W_1,W_2\in C^0_t(L^\ii_{x,v},*)$ two non-negative solutions to \eqref{eq:vlasov}. We will denote by $\rho_1$ and $\rho_2$ their respective densities that are assumed to belong to $\rho_0+L^\ii_t(L^2_x+L^{\min(1+2/d,2)}_x)$. We also define $F_1=-\nabla_x(w*\rho_1)$ and $F_2=-\nabla_x(w*\rho_2)$ their respective force fields, which satisfy $D_x F_i\in L^\ii(\R^d)$ due to our assumptions on $w$ and $\rho_i$. We will use the notation
$$\Phi_i(t,x,v)=(X_i(t,x,v),V_i(t,x,v))$$
for the Newton flow associated to $F_i$; that is the unique solution to
$$\begin{cases}
   (\dot{X}_i(t,x,v),\dot{V}_i(t,x,v))=(2V_i(t,x,v),F_i(t,X(t,x,v))),\\
   (X,V)_{|t=0}=(x,v)
  \end{cases}$$
for any $(x,v)\in\R^d\times\R^d$ given by Lemma \ref{lem:ode}. This result also gives that 
$$W_i(t)=\Phi_i(t)_\sharp W_0.$$
Then, we have the following simple estimate.

\begin{lemma}
 For any $t\in\R$ and any $(x,v)\in\R^d\times\R^d$ we have the estimate
 \begin{multline}\label{eq:est-newton}
    |X_1(t,x,v)-X_2(t,x,v)|+|V_1(t,x,v)-V_2(t,x,v)|\le\\
    e^{\int_0^t(2+\|\nabla F_2(s)\|_{L^\ii})\,ds}\int_0^t|F_1(s,X_1(s,x,v))-F_2(s,X_1(s,x,v))|\,ds.
 \end{multline}
\end{lemma}

\begin{proof}
 By Newton's equations we have
 \begin{multline*}
    |X_1(t,x,v)-X_2(t,x,v)|+|V_1(t,x,v)-V_2(t,x,v)|\\
    \le \int_0^t(2|V_1(s,x,v)-V_2(s,x,v)|+|F_1(s,X_1(s,x,v))-F_2(s,X_2(s,x,v))|)\,ds
 \end{multline*}
 which we estimate by splitting the term with the fields into
 \begin{multline*}
  |X_1(t,x,v)-X_2(t,x,v)|+|V_1(t,x,v)-V_2(t,x,v)|\\
    \le \int_0^t(2|V_1(s,x,v)-V_2(s,x,v)|+|F_2(s,X_1(s,x,v))-F_2(s,X_2(s,x,v))|)\,ds\\
    + \int_0^t|F_1(s,X_1(s,x,v))-F_2(s,X_1(s,x,v))|\,ds
 \end{multline*}
 Using the Lipschitz character of $F_2$, we deduce
 \begin{multline*}
  \int_0^t(2|V_1(s,x,v)-V_2(s,x,v)|+|F_2(s,X_1(s,x,v))-F_2(s,X_2(s,x,v))|)\,ds\\
  \le \int_0^t(2+\|\nabla F_2(s)\|_{L^\ii})(|X_1-X_2|(s,x,v)+|V_1-V_2|(s,x,v))\,ds.
 \end{multline*}
 The result then follows by Gronwall's lemma. 
\end{proof}

Loeper's proof of uniqueness relies on introducing the quantity
\begin{equation}\label{eq:def-Q}
Q(t):=\int_{\R^d}\int_{\R^d}W_0(x,v)|X_1(t,x,v)-X_2(t,x,v)|^2\,dx\,dv. 
\end{equation}
From \eqref{eq:est-newton} we deduce that for any $t\in\R$,
$$Q(t)\le C_t\int_0^t \int_{\R^d}|F_1(s,y)-F_2(s,y)|^2\rho_1(s,y)\,dy\,ds,$$
with 
$$C_t:=t\exp\left(2\int_0^t(2+\|\nabla F_2(s)\|_{L^\ii})\,ds\right).$$
Using our assumption on $\rho_1$, we deduce that 
\begin{equation}\label{eq:est-Q-above}
    Q(t)\le C_t\|\rho_1\|_{L^\ii_t(L^{\min(1+2/d,2)}+L^\ii)}\int_0^t\|F_1(s)-F_2(s)\|_{L^2\cap L^{\max(d+2,4)}}^2\,ds.
\end{equation}
In particular, $Q(t)$ is finite for all $t\in\R$ since
\begin{multline*}
    \|F_1(s)-F_2(s)\|_{L^2\cap L^{\max(d+2,4)}}\\
    \le\|\nabla w\|_{L^1\cap L^{\max(d+2,4)/3}}\|\rho_1(s)-\rho_2(s)\|_{L^2+L^{\min(1+2/d,2)}}
\end{multline*}
for all $s$. The next part of the proof consists in estimating $Q(t)$ from below by some norm on $\rho_1(t)-\rho_2(t)$ that itself controls $\|F_1(t)-F_2(t)\|_{L^2\cap L^{\max(d+2,4)}}$, in order to perform a Gronwall-type argument. While Loeper shows that $Q(t)$ controls the $H^{-1}$-norm of $\rho_1(t)-\rho_2(t)$, we have a slightly weaker control.

\begin{proposition}\label{prop:control-Q}
 Assume $Q(t)$ is defined by \eqref{eq:def-Q}. For any $\rho:\R^d\to\R$, define the norm
 $$\|\rho\|_X:=\sup\left\{\int_{\R^d}\rho(x)\phi(x)\,dx,\ \phi\in L^1(\R^d),\  \|\nabla\phi\|_{L^2\cap L^{\max(d+2,4)}}\le 1\right\}.$$
 Then, for any $T>0$, there exists $C_T>0$ such that for any $t\in[-T,T]$ one has
 $$\|\rho_1(t)-\rho_2(t)\|_X^2\le C_T Q(t).$$
\end{proposition}

Before proving Proposition \ref{prop:control-Q}, we show how it implies Theorem \ref{thm:uniqueness}.

\begin{proof}[Proof of Theorem \ref{thm:uniqueness}]
 Let $\rho\in X$ and $1\le j\le d$. We start by estimating $\partial_jw * \rho$ in $L^2\cap L^{\max(d+2,4)}$ by $\|\rho\|_X$. To do so, let $V\in  (L^1\cap L^2)(\R^d)$. By definition of the $X$-norm, we have
 \begin{align*}
    \int_{\R^d}(\partial_j w * \rho)(x) V(x)\,dx &= \int_{\R^d}\rho(y)(\partial_j w*V)(y)\,dy\\
    &\le \|\rho\|_X\|\nabla\partial_j w * V\|_{L^2\cap L^{\max(d+2,4)}}\\
    &\le C\|\rho\|_X\|D^2w\|_{L^1\cap L^{\max(\frac{2(d+2)}{d+4},\frac 43)}}\|V\|_{L^2},
 \end{align*}
 hence 
 $$\|\partial_j w * \rho\|_{L^2}\le C\|\rho\|_X\|D^2w\|_{L^1\cap L^{\max(\frac{2(d+2)}{d+4},\frac 43)}}.$$
 By the same method, we have
 $$\|\partial_jw*\rho\|_{L^{\max(d+2,4)}}\le C\|\rho\|_X\|D^2w\|_{L^{\max(2,1+d/2}\cap L^{\max(\frac{2(d+2)}{d+4},\frac 43)}}.$$
 We thus have proved that 
 $$\|\nabla w * \rho\|_{L^2\cap L^{\max(d+2,4)}}\le C\|D^2 w\|_{L^1\cap L^{\max(2,1+d/2)}}\|\rho\|_X.$$
 Inserting this estimate into \eqref{eq:est-Q-above}, we deduce the bound
 $$Q(t)\le C_T\int_0^t\|\rho_1(s)-\rho_2(s)\|_X^2\,ds.$$
 By the lower bound on $Q(t)$ given by Proposition \ref{prop:control-Q} and Gronwall's estimate, we deduce that $\rho_1\equiv\rho_2$ which then implies that $F_1\equiv F_2$ and $\Phi_1\equiv\Phi_2$, showing that $W_1\equiv W_2$.
\end{proof}

To prove Proposition \ref{prop:control-Q}, we need the following result from optimal transportation.

\begin{proposition}\label{prop:brenier}
 For any $t\in\R$, there exists $\psi(t):\R^d\to\R\cup\{+\ii\}$ a convex lower semi-continuous function, differentiable almost everywhere on the support on $\rho_1(t)$, such that $\rho_2(t)=(\nabla\psi(t))_\sharp\rho_1(t)$, and such that
 $$Q(t)\ge\int_{\R^d}|x-\nabla\psi(t,x)|^2\rho_1(t,x)\,dx.$$
\end{proposition}

When $\int\rho_1(t,x)\,dx=\int\rho_2(t,x)\,dx<\ii$, Proposition \ref{prop:brenier} is a standard result of optimal transportation due to Brenier and McCann. Since here these integrals are infinite, we need to extend this optimal transportation result and this is done in Appendix \ref{app:opt-transp}, where Proposition \ref{prop:brenier} is proved. We now explain how it implies Proposition \ref{prop:control-Q}, following the same strategy as Loeper.

\begin{proof}[Proof of Proposition \ref{prop:control-Q}]
In the following, we drop the dependence in $t$ from our notation. Once the existence of $\psi$ given by Proposition \ref{prop:brenier} is known, we introduce the interpolation between $\rho_1$ and $\rho_2$ invented by McCann \cite{McCann-97}:
$$\rho_\theta:=((\theta-1)\nabla\psi+(2-\theta)\text{id})_\sharp\rho_1,\quad \theta\in[1,2].$$
This path has nice properties \cite[Thm. 2.2]{McCann-97}, in particular that for any function $A:\R_+\to\R_+$, satisfying that $A(0)=0$ and that $\lambda>0\mapsto\lambda^d A(\lambda^{-d})$ is convex, it holds that the map
$$\theta\in[1,2]\mapsto\int_{\R^d}A(\rho_\theta(x))\,dx$$
is convex. We apply this result to the map
$$A:\rho\ge0\mapsto\1(\rho\ge\rho_0)\times
\begin{cases}
 \rho^{1+2/d}-\rho_0^{1+2/d}-\frac{d+2}{d}\rho_0^{2/d}(\rho-\rho_0) &\text{if}\ d\ge3, \\
 (\rho-\rho_0)^2 &\text{if}\ d=1,2,
\end{cases}
$$
which can be shown to satisfy the desired properties. In particular, we have for all $\theta\in[1,2]$,
$$\int_{\R^d}A(\rho_\theta(x))\,dx\le\max\left(\int_{\R^d}A(\rho_1(x))\,dx,\int_{\R^d}A(\rho_2(x))\,dx\right).$$
Since $A(\rho)\sim\rho^{\min(1+2/d,2)}$ as $\rho\to+\ii$, we have that 
$$A(\rho)\ge\frac12\1(\rho\ge R)\rho^{\min(1+2/d,2)}$$
for some large $R>0$. We can then estimate
\begin{align*}
    \|\1(\rho_\theta\ge R)\rho_\theta\|_{L^{\min(1+2/d,2)}}^{\min(1+2/d,2)}&\le 2\int_{\R^d}A(\rho_\theta(x))\,dx\\
    &\le \max\left(\int_{\R^d}A(\rho_1(x))\,dx,\int_{\R^d}A(\rho_2(x))\,dx\right).
\end{align*}
Since on the other hand we have
$$A(\rho)\le C(\1(|\rho-\rho_0|\le R)(\rho-\rho_0)^2+\1(|\rho-\rho_0|\ge R)(\rho-\rho_0)^{\min(1+2/d,2)},$$
we deduce that 
\begin{multline*}
    \int_{\R^d}A(\rho_i(x))\,dx\\
    \le C(\|\1(|\rho_i-\rho_0|\le R)(\rho_i-\rho_0)\|_{L^2}^2+\|\1(|\rho_i-\rho_0|\ge R)(\rho_i-\rho_0)\|_{L^{\min(1+2/d,2)}}^{\min(1+2/d,2)}).
\end{multline*}
Now $\rho_i-\rho_0\in L^2+L^{\min(1+2/d,2)}$, hence $\1(|\rho_i-\rho_0|\le R)(\rho_i-\rho_0)\in L^2$ and $\1(|\rho_i-\rho_0|\ge R)(\rho_i-\rho_0)\in L^{\min(1+2/d,2)}$. 
This shows that
$$\sup_{\theta\in[1,2]}\|\rho_\theta\|_{L^{\min(1+2/d,2)}+L^\ii}<+\ii.$$
Now let $\phi\in L^1(\R^d)$. We write 
\begin{align*}
  &\int_{\R^d}(\rho_2(x)-\rho_1(x))\phi(x)\,dx\\
  &=\int_1^2\frac{d}{d\theta}\int_{\R^d}\rho_\theta(x)\phi(x)\,dx\,d\theta\\
  &=\int_1^2\frac{d}{d\theta}\int_{\R^d}\rho_1(x)\phi((\theta-1)\nabla\psi(x)+(2-\theta)x)\,dx\\
  &=\int_1^2\int_{\R^d}\rho_1(x)(\nabla\psi(x)-x)\cdot\nabla\phi((\theta-1)\nabla\psi(x)+(2-\theta)x)\,dx\\
  &\le\sup_{\theta\in[1,2]}\left(\int_{\R^d}\rho_1(x)|x-\nabla\psi(x)|^2\,dx\right)^{1/2}\left(\int_{\R^d}\rho_\theta(x)|\nabla\phi(x)|^2\,dx\right)^{1/2},
\end{align*}
which shows that for all $\phi$, 
\begin{multline*}
    \left|\int_{\R^d}(\rho_2(t,x)-\rho_1(t,x))\phi(x)\,dx\right|^2\\
    \le Q(t)\sup_{\theta\in[1,2]}\|\rho_\theta\|_{L^{\min(1+2/d,2)}+L^\ii}\|\nabla\phi\|_{L^2\cap L^{\max(d+2,4)}}^2.
\end{multline*}
By the definition of the space $X$, this finishes the proof.
\end{proof}

At this step we also have completed the proof of Theorem~\ref{thm:limit-sc}.

\section{Accessible initial data}\label{sec:wp-vlasov}

In this section, we build families of initial data $(\gamma_{0,\hbar})_\hbar$ for the Hartree equation \eqref{eq:hartree} so that $\hbar^d\cH(\gamma_{0,\hbar},\gamma_{\rm ref})$ remains bounded and such that their Wigner transform converges as $\hbar\to0$ towards a given phase-space distribution. This will have two important consequences: (i) this shows that the assumptions of Theorem \ref{thm:limit-sc} are not empty and (ii) using these particular initial states, we can show the well-posedness of the Vlasov equation in the energy space (Theorem \ref{thm:wp-vlasov}). We start by constructing appropriate initial data for Theorem \ref{thm:limit-sc}.

\begin{lemma}\label{lem:quantiz}
 Let $S$ be an admissible entropy functional and define $f:=(S')^{-1}$, $S'(1^-):=\lim_{x\to 1}S'(x)<0$. Assume that $f$ and $y\mapsto yf(y)$ are continuous on $]S'(1^-),+\ii[$ and vanish at $+\ii$, together with their first and second derivatives. Let $a\in C^\ii_0(\R^d\times\R^d)$ such that $|v|^2+a(x,v)>S'(1^-)$ for all $(x,v)\in\R^d\times\R^d$. Define $m_{\rm ref}(x,v)=f(|v|^2)$, $m(x,v)=f(|v|^2+a(x,v))$. Let the operators
 $$\gamma_{\rm ref}=f(-\hbar^2\Delta),\quad \gamma_\hbar=f\Big(-\hbar^2\Delta+\opw^\hbar(a)\Big),$$
 which is well-defined for $\hbar>0$ small enough by G{\aa}rding's inequality. Then, we have $W^\hbar_{\gamma_\hbar}\to m$ as $\hbar\to0$ in $\cD'_{x,v}$ and 
 $$\lim_{\hbar\to0}\hbar^d\cH_S(\gamma_\hbar,\gamma_{\rm ref})=\hcl_S(m,m_{\rm ref}).$$
\end{lemma}

\begin{remark}
 This choice of 'quantization' has several advantages. First, we need that $0\le\gamma_\hbar\le 1$ which is the reason why we put the quantization as an argument of the function $f$. Another way to obtain this condition would be to use the 'Husimi quantization',
 $$\oph^\hbar(m)=\int_{\R^d}\int_{\R^d}m(x,v)|\chi^\hbar_{x,v}\rangle\langle\chi^\hbar_{x,v}|\frac{dx\,dv}{(2\pi\hbar)^d},$$
 but it is then not obvious how to show that this state has a finite relative entropy with respect to $\gamma_{\rm ref}$. Using a reverse Berezin-Lieb inequality
 $$\cH_S(\oph^\hbar(m),\oph^\hbar(m_{\rm ref}))\le\hbar^{-d}\hcl_S(m,m_{\rm ref}),$$
 which holds for $S_0$ and $S_f$ (and thus for $S_b$ again up to a factor), $\oph^\hbar(m)$ has a finite relative entropy with respect to $\oph^\hbar(m_{\rm ref})$ which is the wrong reference state:
 $$\oph^\hbar(m_{\rm ref})=(f(|\cdot|^2)*\hbar^{-d/2}|\hat{\chi}(\cdot/\sqrt{\hbar})|^2)(-i\hbar\nabla)\neq f(-\hbar^2\Delta).$$
 It is not clear that this new reference state is of the type $(S')^{-1}(-\hbar^2\Delta)$ for some admissible entropy functional $S$ (that would anyway depend on $\hbar$), and this is why we do not use this approach. In addition, it is easier to compute the Wigner transform using symbolic calculus with the Weyl quantization. 
\hfill$\diamond$\end{remark}

\begin{proof}
  Let $\eta:=\min_{(x,v)}(|v|^2+a(x,v))$. Let $g:\R\to[0,1]$ a $C^2$ function which coincides with $(S')^{-1}$ on $[\eta',+\ii)$ for some $S'(1^-)<\eta'<\min(0,\eta)$ and such that $g\equiv0$ on $(-\ii,\eta'-1]$, in order to have that $\gamma_\hbar=g(\opw^\hbar(|v|^2+a))$ and $\gamma_{\rm ref}=g(\opw^\hbar(|v|^2))$ for $\hbar>0$ small enough (using that $-\hbar^2\Delta+\opw^\hbar(a)\ge \eta -c\hbar$ by G{\aa}rding's inequality \cite[Thm. 4.32]{Zworski-book}). The fact that $W^\hbar_{\gamma_\hbar}\to m$ as $\hbar\to0$ is a standard fact from semi-classical analysis, that we show in Proposition \ref{coro:wigner} in Appendix \ref{sec:sc}. Since $m$ differs from $m_{\rm ref}$ only on a compact set, it is clear that $\hcl_S(m,m_{\rm ref})<\ii$. To show that the quantum relative entropy converges towards the classical relative entropy, we will use the integral representation for the relative entropy
  $$\cH_S(\gamma_\hbar,\gamma_{\rm ref})=\int_0^1\tr(g(\opw^\hbar(|v|^2+ta))-g(\opw^\hbar(|v|^2+a))\opw^\hbar(a)\,dt,$$
  that we prove in Appendix \ref{app:int-rep}. Now let $\epsilon>0$ and take $g_\epsilon\in C^\ii_0(\R^d)$ such that $\|g_\epsilon-g\|_{L^\ii}\le\epsilon$. Using that $\|\opw^\hbar(a)\|_{\gS^1}\le C\hbar^{-d}$, we infer that for all $\epsilon>0$, 
  \begin{multline*}
    (2\pi\hbar)^d\cH_S(\gamma_\hbar,\gamma_{\rm ref})=\\
    \int_0^1(2\pi\hbar)^d\tr(g_\epsilon(\opw^\hbar(|v|^2+ta))-g_\epsilon(\opw^\hbar(|v|^2+a))\opw^\hbar(a)\,dt+\cO(\epsilon),
  \end{multline*}
  where the error $\cO(\epsilon)$ is uniform as $\hbar\to0$. In the proof of Proposition \ref{coro:wigner}, we prove that for any $g\in C^\ii_0(\R)$ and for any $b,c\in C^\ii_0(\R^d\times\R^d)$ we have 
 \begin{multline*}
    (2\pi\hbar)^d\tr(g(|v|^2+\opw^\hbar(b))\opw^\hbar(c))\\
    =\int_{\R^d}\int_{\R^d}g(|v|^2+b(x,v))c(x,v)\,dx\,dv+\cO(\hbar).
 \end{multline*}
 As a consequence, we have 
 \begin{multline*}
    (2\pi\hbar)^d\tr(g_\epsilon(\opw^\hbar(|v|^2+ta))-g_\epsilon(\opw^\hbar(|v|^2+a))\opw^\hbar(a)\\
    =\int_{\R^d}\int_{\R^d}(g_\epsilon(|v|^2+ta(x,v))-g_\epsilon(|v|^2+a(x,v)))a(x,v)\,dx\,dv+\cO(\hbar),
 \end{multline*}
 where the $\cO(\hbar)$ may depend badly on $\epsilon$ but is uniform in $t\in[0,1]$. Taking $\hbar\to0$ with $\epsilon$ fixed, we deduce that 
 \begin{multline*}
    \lim_{\hbar\to0}(2\pi\hbar)^d\cH_S(\gamma_\hbar,\gamma_{\rm ref})=\\
    \int_0^1\int_{\R^d}\int_{\R^d}(g_\epsilon(|v|^2+ta(x,v))-g_\epsilon(|v|^2+a(x,v)))a(x,v)\,dx\,dv+\cO(\epsilon).
 \end{multline*}
 for all $\epsilon>0$. Now the limit $\epsilon\to0$ may be taken by dominated convergence since $a$ is compactly supported.
\end{proof}

\begin{remark}\label{rk:quantiz-singular}
 In the case $S'(1^-)=0$, the previous construction is not well-defined since $(S')^{-1}$ is only defined on $[0,+\ii)$ and hence the operator $(S')^{-1}(\opw^\hbar(|v|^2+a))$ may not make sense since $\opw^\hbar(|v|^2+a)$ may not be non-negative. However, if we assume that $S$ can be extended from $[0,1]$ to $[0,1+\delta]$ for some $\delta>0$ as an admissible entropy functional, the previous result holds for any $a\in C^\ii_0(\R^d\times\R^d)$ such that $|v|^2+a(x,v)\ge0$ for all $(x,v)$. This assumption on $S$ is satisfied for any $S\in\ccS$ built on $S_0$ and $S_b$, but of course not on $S_f$. However, if $S$ is built on $S_f$, we have $S'(1^-)=-\ii$ and hence $S'(1^-)=0$ never happens in this case.
\end{remark}

The next lemma shows that, at the classical level, putting the perturbation inside the function $f$ is the same as putting it outside.

\begin{lemma}\label{lem:redr}
 Let $S$ and admissible entropy functional such that $S'(1^-)<0$. Let $b\in C^\ii_0(\R^d\times\R^d)$ such that $0<(S')^{-1}(|v|^2)+b(x,v)<1$. Then, there exists $a\in C^\ii_0(\R^d\times\R^d)$ with $|v|^2+a(x,v)>S'(1^-)$ such that $(S')^{-1}(|v|^2)+b(x,v)=(S')^{-1}(|v|^2+a(x,v))$.
\end{lemma}

\begin{proof}
  Since $0<(S')^{-1}(|v|^2)+b(x,v)<1$ for all $(x,v)$ and since $S$ is $C^\ii$ on $(0,1)$, the function 
  $$a:(x,v)\mapsto S'((S')^{-1}(|v|^2)+b(x,v))-|v|^2$$
  is well-defined and $C^\ii$ on $\R^d\times\R^d$. It is furthermore clearly compactly supported with $\supp a=\supp b$.
\end{proof}

\begin{remark}\label{rk:redr-singular}
 If $S'(1^-)=0$, the previous lemma still holds for all $b\in C^\ii_0(\R^d\times\R^d)$ such that $0<(S')^{-1}(|v|^2)+b(x,v)\le1$, provided that $(S')^{-1}$ has first and second derivatives that have limits at $0$ (which is the case if $S$ can be extended from $[0,1]$ to $[0,1+\delta]$ for some $\delta>0$ as an admissible entropy functional). One then obtains $a\in C^\ii_0(\R^d\times\R^d)$ with $|v|^2+a(x,v)\ge0$ for all $(x,v)$.
\end{remark}

The last lemma shows that the class of states $(S')^{-1}(|v|^2+a(x,v))$ includes all reasonable smooth and compactly supported perturbations of $(S')^{-1}(|v|^2)$. These states are actually dense in the sense of the relative entropy, as we next show.

\begin{lemma}\label{lem:approx-smooth}
 Let $S$ an admissible entropy functional such that $\lim_{x\to1}S'(x)=:S'(1^-)<0$. Let $b:\R^d\times\R^d\to\R$ a measurable function. Define $m_{\rm ref}(x,v)=(S')^{-1}(|v|^2)$ and $m=m_{\rm ref}+b$. Assume that $0\le m\le 1$ and $\hcl_S(m,m_{\rm ref})<\ii$. Then, there exists a sequence $(b_n)\subset C^\ii_c(\R^d\times\R^d)$ satisfying $0<m_n:=m_{\rm ref}+b_n<1$ for all $n$, $\hcl_S(m_n,m_{\rm ref})\to\hcl_S(m,m_{\rm ref})$ as $n\to\ii$ and $b_n\to b$ as $n\to\ii$ in $L^2(\R^d\times\R^d)$.
\end{lemma}

\begin{proof}
 Since $m$ has a finite relative entropy with respect to $m_{\rm ref}$, we deduce that $b\in L^2(\R^d\times\R^d)$ by the (classical) Klein inequality. Now we proceed in two steps: first we localize $b$ and then we smooth it. Let $R>0$. Defining $b_R=\1_{B(0,R)}\1(1/R\le m_{\rm ref}+b\le 1-1/R) b$, we clearly have $0<m_{\rm ref}+b_R<1$ for all $R$ and that $b_R\to b$ in $L^2$ as $R\to\ii$. Furthermore,
 \begin{multline*}
    \hcl(m_{\rm ref}+b_R,m_{\rm ref})=-\int_{B(0,R)}\,dx\,dv\,\1(\frac1R\le m_{\rm ref}+b\le 1-\frac1R)\times\\
    \times\Big(S(m_{\rm ref}(x,v)+b(x,v))-S(m_{\rm ref}(x,v))-|v|^2b(x,v)\Big)
 \end{multline*}
 and thus $\hcl(m_{\rm ref}+b_R,m_{\rm ref})\to\hcl(m_{\rm ref}+b,m_{\rm ref})$ as $R\to\ii$. Replacing $b$ by $b_R$, we may thus assume that $b$ is compactly supported and that $0<m_{\rm ref}+b<1$ (with even $0<\delta_K\le m_{\rm ref}+b\le 1-\delta_K$ on any compact set $K$). Next, we approach $b$ by $b_\epsilon:=b*\eta_\epsilon$, where $\eta_\epsilon:=\epsilon^{-2d}\eta(\cdot/\epsilon)$ with $\epsilon\in(0,1)$ and $\eta\in C^\ii_0(\R^d\times\R^d)$, $\supp\eta\subset B(0,1)$, $\int_{\R^{2d}}\eta=1$, $\eta\ge0$. It also satisfies $0<m_{\rm ref}(x,v)+b_\epsilon(x,v)<1$: indeed, we have $\supp b_\epsilon\subset\supp b + B(0,1)$ hence for all $X\in\supp b_\epsilon$ and for all $Y\in\supp\eta_\epsilon$ we have $X-Y\in K:=\supp b+B(0,2)$ a compact set, hence 
 $$0<\delta_K\eta_\epsilon(Y)\le(m_{\rm ref}(X-Y)+b(X-Y))\eta_\epsilon(Y)\le(1-\delta_K)\eta_\epsilon(Y).$$
 Integrating over $Y$ shows that for all $X\in\supp b_\epsilon$, we have 
 $$0<\delta_K\le m_{\rm ref}*\eta_\epsilon+b_\epsilon\le 1-\delta_K,$$
 and hence
 $$\delta_K+m_{\rm ref}-m_{\rm ref}*\eta_\epsilon\le m_{\rm ref}+b_\epsilon \le 1-\delta_K+m_{\rm ref}-m_{\rm ref}*\eta_\epsilon.$$
 Since $m_{\rm ref}-m_{\rm ref}*\eta_\epsilon\to0$ as $\epsilon\to0$ in $L^\ii$, we deduce that $0<m_{\rm ref}+b_\epsilon<1$ for $\epsilon>0$ small enough, on $\supp b_\epsilon$. Away from $\supp b_\epsilon$, we have also $0<m_{\rm ref}+b_\epsilon=m_{\rm ref}<1$. Furthermore, we clearly have $b_\epsilon\to b$ in $L^2$ as $\epsilon\to0$, as well as almost everywhere (perhaps up to a subsequence). As a consequence,
 \begin{multline*}
    S(m_{\rm ref}(x,v)+b_\epsilon(x,v))-S(m_{\rm ref}(x,v))-|v|^2b_\epsilon(x,v)\\
    \to S(m_{\rm ref}(x,v)+b(x,v))-S(m_{\rm ref}(x,v))-|v|^2b(x,v)
 \end{multline*}
 as $\epsilon\to0$, almost everywhere in $(x,v)$ by continuity of $S$. Now the function 
 $$(x,v)\mapsto S(m_{\rm ref}(x,v)+b_\epsilon(x,v))-S(m_{\rm ref}(x,v))-|v|^2b_\epsilon(x,v)$$
 is supported in $\supp b+B(0,1)$ and is uniformly bounded in $L^\ii_{\text{loc}}(\R^d\times\R^d)$ as $\epsilon\to0$ since $S$ is continuous on $[0,1]$ and thus bounded, and since $b$ is also bounded due to the constraint $m_{\rm ref}+b\in(0,1)$. By dominated convergence, we deduce that $\hcl(m_{\rm ref}+b_\epsilon,m_{\rm ref})\to\hcl(m_{\rm ref}+b,m_{\rm ref})$ as $\epsilon\to0$. This concludes the proof.
\end{proof}

\begin{remark}\label{rk:approx-smooth-singular}
 In the case $S'(1^-)=0$, we obtain the same result (replacing $0<m_{\rm ref}+b_n<1$ by $0<m_{\rm ref}+b_n\le1$) by defining 
 $$b_R=\1_{B(0,R)}\1(|v|\ge1/R)\1(1/R\le m_{\rm ref}+b\le 1-1/R) b$$
 which now satisfies that $0<\delta_K\le m_{\rm ref}+b_R\le 1-\delta_K$ on any compact subset $K$ of $\R^d\times\R^d\setminus(\R^d\times\{0\})$. Since $\supp b_\epsilon\subset\supp b+B(0,\epsilon)\subset\R^d\times\R^d\setminus(\R^d\times\{0\})$ for $\epsilon$ small enough, we deduce again that $0<\delta\le m_{\rm ref}+b_\epsilon\le1-\delta$ on $\supp b_\epsilon$, while away from $\supp b_\epsilon$ we have $0<m_{\rm ref}+b_\epsilon=m_{\rm ref}\le1$.
\end{remark}

We will now ask the question of the well-posedness of the Vlasov equation: given any $b_0$ with $0\le m_{\rm ref}+b_0\le 1$ and $\hcl(m_{\rm ref}+b_0,m_{\rm ref})<+\ii$, does there exist a solution $W$ of the nonlinear Vlasov equation with $\hcl(W(t),m_{\rm ref})\le C$ for all $t$ and such that $W_0=m_{\rm ref}+b_0$? We answer this question by going to the quantum evolution. We begin with the case of smooth perturbations of $m_{\rm ref}$.

\begin{proposition}\label{prop:gwp-smooth}
 Assume that $S\in\ccS$ and define $m_{\rm ref}(x,v)=(S')^{-1}(|v|^2)$. Let $b_0\in C^\ii_0(\R^d\times\R^d)$ such that $0<m_{\rm ref}+b_0\le1$. Then, there exists a unique $b\in L^\ii_t(\R, L^2_{x,v})\cap C^0_t(\R,\cD'_{x,v}(\R^d\times\R^d))$ such that $0\le m_{\rm ref}+b(t)\le 1$ and $\hcl_S(m_{\rm ref}+b(t),m_{\rm ref})\le C(\hcl_S(m_{\rm ref}+b_0,m_{\rm ref}))$ for all $t\in\R$, with $W(t):=m_{\rm ref}+b(t)$ solving the Cauchy problem for the Vlasov equation \eqref{eq:vlasov} with initial condition $W(0)=m_{\rm ref}+b_0$.
\end{proposition}

\begin{proof}
 By Lemma \ref{lem:redr} (and Remark \ref{rk:redr-singular} in the case $S'(1^-)=0$), let $a_0\in C^\ii_0(\R^d\times\R^d)$ such that 
 $$m_{\rm ref}+b_0=(S')^{-1}(|v|^2+a_0),$$
 and by Lemma \ref{lem:quantiz} (and Remark \ref{rk:quantiz-singular}) in the case $S'(1^-)=0$) let
 $$\gamma_\hbar=(S')^{-1}(-\hbar^2\Delta+\opw^\hbar(a_0))$$
 such that 
 $$\lim_{\hbar\to0}\hbar^d\cH_S(\gamma_\hbar,\gamma_{\rm ref})=\hcl_S(m_{\rm ref}+b_0,m_{\rm ref}).$$
 Applying Theorem \ref{thm:gwp-hartree-h}, Theorem \ref{thm:limit-sc}, and the estimate \eqref{eq:control-entropy-time}, there exists a solution to the Hartree equation with initial data $\gamma_\hbar$ with $\hbar^d\cH_S(\gamma_\hbar(t),\gamma_{\rm ref})\le C( \hbar^d\cH_S(\gamma_\hbar,\gamma_{\rm ref}))$ for all $t$ and $\hbar\in(0,1)$, which converges as $\hbar\to0$ to the desired solution to the Vlasov equation. We have already seen that $W^\hbar_{\gamma_\hbar}\to m_{\rm ref}+b_0$ as $\hbar\to0$ in the sense of distributions, so that $W(t)$ has the correct initial condition.
\end{proof}

We now deduce Theorem \ref{thm:wp-vlasov} by approaching any energy initial condition by smooth data. 

\begin{proof}[Proof of Theorem \ref{thm:wp-vlasov}]
 By Lemma \ref{lem:approx-smooth} (and Remark \ref{rk:approx-smooth-singular} in the case $S'(1^-)=0$), let $(b_{n,0})\subset C^\ii_0(\R^d\times\R^d)$ a sequence approaching $W_0-m_{\rm ref}$ in the relative entropy sense. Applying Proposition \ref{prop:gwp-smooth} to $b_{n,0}$, we have a sequence $W_n(t)=m_{\rm ref}+b_n(t)$ of solutions to the Vlasov equations such that 
 $$\sup_{n,t}\hcl_S(W_n(t),m_{\rm ref})\le C(\hcl_S(m_{\rm ref}+b_0,m_{\rm ref})).$$
 Since $(b_n)$ is bounded in $L^\ii_t L^2_{x,v}=(L^1_t L^2_{x,v})'$ (by the classical Klein inequality \eqref{eq:klein-classical}), we may extract a sequence such that $b_n\to b$ weakly in $L^\ii_t L^2_{x,v}$ (and similar weak limits for the density $\rho_{b_n}$ using the classical Lieb-Thirring inequality \eqref{eq:LT-classical}). The fact that $b$ is a solution to the Vlasov equation follows from the same arguments as in the proof of Theorem \ref{thm:limit-sc} (the linear terms are easy, and one can take the limit in the nonlinear terms thanks to the high frequency estimates of the relative entropy). We can also show that $b\in C^0_t\cD'_{x,v}$ again by the Ascoli-Arzela theorem and the uniform entropy bounds. In particular, $b_n(0)\to b(0)$ as $n\to\ii$ in $\cD'_{x,v}$, and since $b_n(0)=b_{n,0}\to b_0$ in $L^2_{x,v}$, we deduce that $b(0)=b_0$.
\end{proof}

\appendix

\section{Proof of Lemma \ref{lem:ode}}\label{app:ode}

 \begin{proof}[Proof of Lemma \ref{lem:ode}]

  Fixing $s$ and $x$, the existence of $\Phi^t_s(x)$ follows from a fixed point theorem in $C^0([s-T,s+T])$ for $T=T(F)>0$ small enough (independent of $s$ and $x$), using the fact that $F$ is globally Lipschitz. For fixed $x$ and $s$, such a $\Phi$ is unique (in $C^0([s-T',s+T''])$ for any $T',T''>0$) by a standard Gronwall argument, using again the fact that $F$ is Lipschitz. Using now that for any $s+T/2<t<s+3T/2$, we have 
  \begin{align*}
    \Phi_{s+T/2}^{t}&(\Phi_{s}^{s+T/2}(x)) \\
    &= \Phi_{s}^{s+T/2}(x)+\int_{s+T/2}^{t} F(s',\Phi_{s+T/2}^{s'}(\Phi_{s}^{s+T/2}(x))\,ds' \\
    &= x + \int_s^{s+T/2}F(s',\Phi_s^{s'}(x))\,ds'+\int_{s+T/2}^{t} F(s',\Phi_{s+T/2}^{s'}(\Phi_{s}^{s+T/2}(x))\,ds', 
  \end{align*}
  we deduce that the function 
  $$t\in[s,s+3T/2]\mapsto\begin{cases}
                          \Phi_s^t(x) & \text{if}\ s\le t\le s+T/2,\\
                          \Phi_{s+T/2}^t\circ\Phi_s^{s+T/2} & \text{if}\ s+T/2\le t\le s+3T/2,
                         \end{cases}
  $$
  is a continuous solution to the integral equation at $s$ and $x$ fixed. By uniqueness, this shows that $\Phi$ can be extended to $[s-T,s+3T/2]$ and iterating this process it can be extended to a continuous function on $\R$. The same argument then shows the semi-group property
  $$\Phi_{t_1}^{t_3}=\Phi_{t_2}^{t_3}\circ\Phi_{t_1}^{t_2}.$$
  By a Gronwall argument, we also get the uniform bound
  $$|\Phi_s^t(x)|\le e^{C|t-s|}(|x|+C|t-s|).$$  
  The semi-group property implies that $\Phi_s^t$ is a bijection of $\R^N$ for any $s,t$ (with inverse $\Phi_t^s$). Now Gronwall again implies that 
  $$|\Phi_s^t(x)-\Phi_s^t(y)|\le e^{C|t-s|}|x-y|,$$
  and using the semi-group property we deduce the lower bound
  $$|\Phi_s^t(x)-\Phi_s^t(y)|\ge e^{-C|t-s|}|x-y|.$$
  This shows that $\Phi_s^t$ is bi-Lipschitz for any $s,t$. Using the integral equation, this proves that $\Phi$ is continuous as a function of $(s,t,x)$. To show that $\Phi^t_s$ is $C^1$, let $x\in\R^N$, $1\le j\le N$ and $\alpha\neq0$. Defining
  $$f_\alpha(t):=\frac{\Phi_s^t(x+\alpha e_j)-\Phi_s^t(x)}{\alpha},$$
  and using the continuity of $\Phi_s^t$, we find that
  $$f_\alpha(t)=e_j+\int_s^t D_{\Phi_s^{s'}(x)}F(s',f_\alpha(s'))\,ds'+o_{\alpha\to0}(1),$$
  hence for any $\alpha,\alpha'\neq0$ we have
  $$f_\alpha(t)-f_{\alpha'}(t)=\int_s^t D_{\Phi_s^{s'}(x)}F(f_\alpha(s')-f_{\alpha'}(s'))\,ds'+o_{\alpha\to0}(1).$$
  Using again Gronwall, we deduce that 
  $$|f_\alpha(t)-f_{\alpha'}(t)|\le e^{C|t-s|}o_{\alpha\to0}(1),$$
  showing that $f_\alpha(t)$ admits a limit as $\alpha\to0$. It satisfies the equation
  $$\partial_{x_j}\Phi_s^t(x)=e_j+\int_s^t D_{\Phi_s^{s'}(x)}F(s',\partial_{x_j}\Phi_s^{s'}(x))\,ds'.$$
  For any $h\in\R^N$, writing the equation satisfied by 
  $$\partial_{x_j}\Phi_s^t(x+h)-\partial_{x_j}\Phi_s^t(x)$$
  and using Gronwall again, we find that 
  $$\partial_{x_j}\Phi_s^t(x+h)-\partial_{x_j}\Phi_s^t(x)\to0$$
  as $h\to0$, showing that $\partial_{x_j}\Phi_s^t$ is continuous. Since it holds for any $j$, this shows that $\Phi_s^t$ is $C^1$. By the semi-group property, $\Phi_s^t$ is a $C^1$-diffeomorphism of $\R^N$ for each $s,t$. We have the equation for all $s,t,x$
  $$D_x\Phi_s^t = \text{id}+\int_s^t D_{\Phi_s^{s'}(x)}F(s')\circ D_x\Phi_s^{s'}\,ds'.$$
  Taking the determinant of the right side which is a Lipschitz function of $t$, and taking its distributional derivative, one finds that 
  $$\frac{d}{dt}\det D_x\Phi_s^t=\tr D_{\Phi_s^t(x)} F(s) \det D_x\Phi_s^t=0,$$
  in the sense of distributions over $t$, showing that 
  $$\det D_x\Phi_s^t =\det D_x\Phi_s^s = 1.$$
  By the previous estimates, we have that 
  $$(t,x)\mapsto F(t,\Phi_0^t(x))\in L^\ii_{{\rm loc}}(\R\times\R^N).$$
  More precisely, we even have
  $$|F(t,\Phi_0^t(x))|\le Ce^{C|t|}(1+|x|+C|t|).$$
  Hence, the identity
  $$\frac{d}{dt}\Phi_0^t(x)=F(t,\Phi_0^t(x))$$
  holds in $L^\ii_{{\rm loc}}(\R\times\R^N)\subset\cD'_{t,x}$. As a consequence, for any $\phi\in\cD_x$, we have
  $$\frac{d}{dt}\phi(\Phi_0^t(x))=\nabla\phi(\Phi_0^t(x))\cdot F(t,\Phi_0^t(x))$$
  also in $L^\ii_{{\rm loc}}(\R\times\R^N)$, implying in particular that for any $t,s,x$,
  $$\phi(\Phi_0^t(x))-\phi(\Phi_0^s(x))=\int_s^t \nabla\phi(\Phi_0^{s'}(x))\cdot F(s',\Phi_0^{s'}(x))\,ds'.$$
  From this identity, we deduce the bound
  $$\int_{\R^N}|\phi(\Phi_0^t(x))-\phi(\Phi_0^s(x))|\,dx\le \int_s^t \int_{\R^N}|\nabla\phi(x)||F(s',x)|\,dx\,ds',$$
  showing that 
  $$\phi(\Phi_0^s(x))\to \phi(\Phi_0^t(x))$$
  as $s\to t$ in $L^1_x$. By density of $\cD_x$ in $L^1_x$, this convergence also holds for $\phi\in L^1_x$, since
  $$\int_{\R^N}|\phi(\Phi_0^s(x))|\,dx=\int_{\R^N}|\phi(y)|\,dy$$
  for all $\phi\in L^1_x$. Now let $u_0\in L^\ii(\R^N)$. By the previous considerations, we deduce that $u\in C^0_t(L^\ii,*)$, which means that 
  $$\langle u(t),\phi\rangle\in C^0_t$$
  for any $\phi\in L^1_x$. By the same identity as before, we also infer that for any $\phi\in\cD_{t,x}$, for any $t\in\R$, and any $h\neq0$,
  $$\left\langle\frac{u(t+h)-u(t)}{h},\phi(t)\right\rangle_x=\frac{1}{h}\int_t^{t+h}\int_{\R^N}u(s,x)\nabla\phi(t,x)\cdot F(s,x)\,dx\,ds.$$
  Since
  $$s\mapsto\int_{\R^N}u(s,x)\nabla\phi(t,x)\cdot F(s,x)\,dx\in L^\ii_s,$$
  we deduce by the Lebesgue differentiation theorem that 
  $$\frac{1}{h}\int_t^{t+h}\int_{\R^N}u(s,x)\nabla\phi(t,x)\cdot F(s,x)\,dx\,ds\to\int_{\R^N}u(t,x)\nabla\phi(t,x)\cdot F(t,x)\,dx$$
  as $h\to0$ in $L^1_t(\R)$. Since 
  $$\langle\frac{u(t+h)-u(t)}{h},\phi(t)\rangle_{t,x}\to\langle \partial_t u,\phi\rangle$$
  as $h\to0$, this shows that $u$ satisfies the continuity equation. It remains to show the opposite direction. Thus, let $u\in C^0_t(L^\ii_x,*)$ any solution to the continuity equation such that $u(0)=u_0$. It is enough to show that $u_t\circ\Phi_0^t$ does not depend on $t$. First, notice that $u_t\circ\Phi_0^t\in C^0_t(L^\ii_x,*)$ and thus makes sense as a distribution: indeed, 
  $$\langle u_t\circ\Phi_0^t,\phi\rangle = \langle u_t,\phi\circ\Phi_t^0\rangle,$$
  and we have already shown that $(t,x)\mapsto \phi(\Phi_t^0(x))\in C^0_t L^1_x$ if $\phi\in L^1_x$. Now let $\phi\in\cD_x$, $t\in\R$, and $h\neq0$. We have 
  $$\langle u_t\circ\Phi_0^t,\phi\rangle_x=\langle u_t,\phi_t\rangle,$$
  where $\phi_t:=\phi\circ\Phi_t^0$. As a consequence,
  $$\frac1h(\langle u_{t+h},\phi_{t+h}\rangle-\langle u_t,\phi_t\rangle)=\langle\frac{u(t+h)-u(t)}{h},\phi_t\rangle+\langle u(t+h),\frac{\phi_{t+h}-\phi_t}{h}\rangle.$$
  The identity
  $$\langle\frac{u(t+h)-u(t)}{h},\phi\rangle_x=\frac{1}{h}\int_t^{t+h}\int_{\R^N}u(s,x)\nabla\phi(x)\cdot F(s,x)\,dx\,ds$$
  has been shown for any $\phi\in\cD_{x}$, which is a priori not the case here. However, by density one can extend it to any $\phi\in L^1_x$ such that $(1+|x|)\nabla\phi(x)\in L^1_x$, which is the case for $\phi=\phi_t$ since 
  $$\nabla\phi_t(x)=-F(t,\Phi_t^0(x))\cdot\nabla\phi(\Phi_t^0(x))$$
  and thus 
  $$\int_{\R^N}(1+|x|)|\nabla\phi_t(x)|\,dx\le \int_{\R^N}(1+|\Phi_t^0(x)|)|F(t,x)||\nabla\phi(x)|\,dx$$
  which is finite since $|\Phi_t^0(x)|\le Ce^{C|t|}(|x|+C|t|)$. By Lebesgue's differentiation theorem, we deduce that 
  $$\langle\frac{u(t+h)-u(t)}{h},\phi_t\rangle\to\int_{\R^N}u(t,x)\nabla\phi_t(x)\cdot F(t,x)\,dx$$
  as $h\to0$ in $L^1_{{\rm loc}}(\R_t)$. On the other hand, the continuity equation also holds for $\phi_t$: we also have that
  $$\partial_t\phi_t+F_t\cdot\nabla\phi_t=0$$
  in $\cD'_{t,x}$. Since the right side belongs to $L^\ii_t L^1_x$ as we have already seen, we deduce that
  $$\phi_{t+h}(x)-\phi_t(x)=-\int_t^{t+h} \nabla\phi_s(x)\cdot F(s,x)\,ds,$$
  as an identity in $C^0_tL^1_x$. Again by Lebesgue's differentiation theorem, we deduce that
  $$\frac{\phi_{t+h}(x)-\phi_t(x)}{h}\to-\nabla\phi_t(x)\cdot F(t,x)$$
  as $h\to0$, strongly in $L^1_{{\rm loc},t}L^1_x$. Hence, by continuity of $u$ in $C^0_t(L^\ii_x,*)$, we have
  $$\langle u(t+h),\frac{\phi_{t+h}-\phi_t}{h}\rangle\to-\int_{\R^N}u(t,x)\nabla\phi_t(x)\cdot F(t,x)\,dx,$$
  as $h\to0$ in $L^1_{{\rm loc},t}$, showing that 
  $$\frac{d}{dt}u_t\circ\Phi_0^t=0$$ 
  in $\cD'_{t,x}$. This shows that 
  $$u_t\circ\Phi_0^t=u_0\circ\Phi_0^0=u_0,$$
  which is the desired identity.
\end{proof}

\section{Infinite-mass optimal transportation}
\label{app:opt-transp}

Recall that a \emph{Radon measure} on $\R^N$ is a non-negative Borel measure on $\R^N$ which is regular and locally finite. 
A sequence $(\mu_n)$ of Radon measures on $\R^N$ converges weakly to another Radon measure $\mu$ on $\R^N$ if for all $f\in C_c(\R^N)$ we have
$$\int_{\R^N}f(x)\,d\mu_n(x)\to\int_{\R^N}f(x)\,d\mu(x)$$
as $n\to\ii$. When this holds, we write $\mu_n\rightharpoonup\mu$. We also recall (see, e.g.,~\cite[Thm. 1.41]{EvaGar-15}) that if a sequence $(\mu_n)$ of Radon measures is such that the sequence
 $$\left(\int_{\R^N}f(x)\,d\mu_n(x)\right)_{n\in\N}$$
 is bounded for every fixed $f\in C_c(\R^N)$, then there exists a Radon measure $\mu$ and a subsequence of $(\mu_n)$ such that $\mu_n\rightharpoonup\mu$
 as $n\to\ii$ along this subsequence. For any Radon measure $\mu$ on $\R^d\times\R^d$, its \emph{marginals} $\mu^{(1)}$ and $\mu^{(2)}$ are non-negative Borel measures on $\R^d$ (not necessarily Radon), such that for all Borel sets $A\subset\R^d$ we have $\mu^{(1)}(A)=\mu(A\times\R^d)$ and $\mu^{(2)}(A)=\mu(\R^d\times A)$. With these definitions at hand, we are ready to state our first result about optimal transport.

\begin{lemma}\label{lem:minimization}
 Let $\rho_1,\rho_2\in L^1_{{\rm loc},+}(\R^d)$. Define $\Gamma(\rho_1,\rho_2)$ as the set of all  Radon measures on $\R^d\times\R^d$ having marginals $\rho_1$ and $\rho_2$. Define the optimal transport problem
 $$C(\rho_1,\rho_2)=\inf_{\Pi\in\Gamma(\rho_1,\rho_2)}\int_{\R^d\times\R^d}|x-y|^2\,d\Pi(x,y).$$
 If $C(\rho_1,\rho_2)<+\ii$, then the infimum defining $C(\rho_1,\rho_2)$ is attained.
\end{lemma}

\begin{remark}
 The assumption $C(\rho_1,\rho_2)<+\ii$ ensures in particular that $\Gamma(\rho_1,\rho_2)$ is not empty, implying that $\int_{\R^d}\rho_1=\int_{\R^d}\rho_2\in[0,+\ii]$. When $\int_{\R^d}\rho_1<+\ii$, Lemma \ref{lem:minimization} is standard and its proof can be found for instance in \cite[Thm. 4.1]{Villani-book}. We did not find this result for $\int_{\R^d}\rho_1=+\ii$ in the literature, so we give a proof below. Contrary to the case $\int_{\R^d}\rho_1<+\ii$ where the set $\Gamma(\rho_1,\rho_2)$ is compact so that any sequence has a converging subsequence, we exploit here the fact that we have a \emph{minimizing} sequence to obtain compactness. 
\hfill$\diamond$\end{remark}

\begin{proof}
 Let $(\Pi_n)$ a minimizing sequence for $C(\rho_1,\rho_2)$. In particular, the sequence 
 $$\int_{\R^d\times\R^d}|x-y|^2d\Pi_n(x,y)$$
 is bounded. For any $\phi\in C_c(\R^d\times\R^d)$, there exists $K\subset\R^d$ a compact set and $C>0$ such that for all $(x,y)\in\R^d\times\R^d$,
 $$|\phi(x,y)|\le C\1_K(x).$$
 As a consequence, we have the bound
 $$\left|\int_{\R^d\times\R^d}\phi(x,y)\,d\Pi_n(x,y)\right|\le C\int_K\rho_1(x)\,dx,$$
 and thus $(\int\phi\,d\Pi_n)_n$ is bounded. By the compactness property of Radon measures recalled above, this implies the existence of a Radon measure $\Pi$ such that $\Pi_n\rightharpoonup\Pi$ up to a subsequence. Now let $\chi\in C_c(\R_+,\R_+)$ be a non-increasing function such that $\chi(0)=1$, and let $R>0$. Define $\chi_R:(x,y)\in\R^d\times\R^d\mapsto \chi(|(x,y)|/R)\in C_c(\R^d\times\R^d)$. Then, $\chi_R(x,y)\to1$ as $R\to+\ii$ non-decreasingly, for all $(x,y)\in\R^d\times\R^d$. By monotone convergence, we deduce that
 \begin{align*}
    \int_{\R^d\times\R^d}|x-y|^2d\Pi(x,y) &= \lim_{R\to\ii}\int_{\R^d\times\R^d}\chi_R(x,y)|x-y|^2d\Pi(x,y),\\
    &= \lim_{R\to\ii}\lim_{n\to\ii}\int_{\R^d\times\R^d}\chi_R(x,y)|x-y|^2d\Pi_n(x,y)\\
    &\le\liminf_{n\to\ii}\int_{\R^d\times\R^d}|x-y|^2d\Pi_n(x,y).
 \end{align*}
 As a consequence, we have proved that 
 $$\int_{\R^d\times\R^d}|x-y|^2d\Pi(x,y)\le\liminf_{n\to\ii}\int_{\R^d\times\R^d}|x-y|^2d\Pi_n(x,y)=C(\rho_1,\rho_2).$$ 
 To show that $\Pi$ is a minimizer for $C(\rho_1,\rho_2)$, it remains to show that it has the right marginals. To do so, let first $f\in C_c(\R^d)$ such that $f\ge0$. Let also $R>0$ and $\chi_R(x)=\chi(|x|/R)$ for all $x\in\R^d$, with $\chi$ the same function as before (assuming furthermore that $\chi\equiv1$ on $[0,1]$). Then, we have for all $n$,
 \begin{align*}
    &\int_{\R^d}f(x)\rho_1(x)\,dx \\
    &= \int_{\R^d\times\R^d}f(x)d\Pi_n(x,y)\\
    &= \int_{\R^d\times\R^d}\chi_R(y)f(x)d\Pi_n(x,y) + \int_{\R^d\times\R^d}(1-\chi_R(y))f(x)d\Pi_n(x,y).
 \end{align*}
 Now for all $x\in\supp f$ we have $|x|\le R'$ for some $R'>0$, while for all $y\in\supp(1-\chi_R)$ we have $|y|\ge R$. In particular, for all $(x,y)\in\supp f\otimes(1-\chi_R)$ we have 
 $$|x-y|\ge R-R'$$
 and thus
 \begin{align*}
    \left|\int_{\R^d\times\R^d}(1-\chi_R(y))f(x)d\Pi_n(x,y)\right|&\le \frac{C}{(R-R')^2}\int_{\R^d\times\R^d}|x-y|^2d\Pi_n(x,y)\\
    &\le\frac{C}{R^2}
 \end{align*}
 for all $R>R'$ and for some $C>0$ independent of $n$. Letting $n\to\ii$ we find that for all $R>R'$,
 $$\int_{\R^d}f(x)\rho_1(x)\,dx=\int_{\R^d}\chi_R(y) f(x)d\Pi(x,y)+\cO_{R\to+\ii}(R^{-2}),$$
 and letting then $R\to\ii$ shows that
 $$\int_{\R^d}f(x)\rho_1(x)\,dx=\int_{\R^d}f(x)\,d\Pi^{(1)}(x),$$
 for all $f\in C_c(\R^d)$, $f\ge0$. This implies in particular that $\Pi^{(1)}$ is finite on compact sets, and using \cite[Thm. 2.18]{Rudin} we deduce that it is also regular. Writing any $f\in C_c(\R^d)$ as a difference of two non-negative compactly supported continuous functions, we also deduce the identity 
 $$\int_{\R^d}f(x)\rho_1(x)\,dx=\int_{\R^d}f(x)\,d\Pi^{(1)}(x),$$
 for all $f\in C_c(\R^d)$. This implies that $\Pi^{(1)}=\rho_1$ by the uniqueness part of Riesz's representation theorem. Of course, the proof for the other marginal is the same.
\end{proof}

The proofs of \cite[Thm. 2.3]{GanMcc-96} and \cite[Prop. 10]{McCann-95} carry on to the infinite mass case to obtain that the minimizers have the Monge form:

\begin{proposition}\label{prop:brenier-2}
 Any minimizer of Lemma \ref{lem:minimization} has the form $\Pi=(\text{id}\times\nabla\psi)_\sharp\rho_1$ for some convex function $\psi$ differentiable $\rho_1$-a.e.
\end{proposition}

We finally give the
\begin{proof}[Proof of Proposition \ref{prop:brenier}]
Defining the measure $\Pi(t)$ on $\R^d\times\R^d$ by
$$\Pi(t)=(X_1(t),X_2(t))_\sharp W_0,$$
we have the identity
$$Q(t)=\int_{\R^d\times\R^d}|x-y|^2d\Pi(t),$$
while $\Pi(t)$ has marginals $\Pi^{(1)}(t)=\rho_1(t)$, $\Pi^{(2)}(t)=\rho_2(t)$. As a consequence, we have the lower bound
$$Q(t)\ge\inf_{\Pi':\, (\Pi')^{(1)}=\rho_1(t), (\Pi')^{(2)}=\rho_2(t)}\int_{\R^d\times\R^d}|x-y|^2d\Pi'$$
in terms of an optimal transport problem. It remains to apply Lemma \ref{lem:minimization} and Proposition \ref{prop:brenier-2} to finish the proof.
\end{proof}

\section{Some tools from semi-classical analysis}\label{sec:sc}

We recall \cite[Sec. 4.4]{Zworski-book} that $m:\R^d\times\R^d\to(0,+\ii)$ is called an order function if there exists $C>0$ and $N>0$ such that for all $w,z\in\R^d\times\R^d$, 
$$m(z)\le C (1+|z-w|^2)^Nm(w).$$
We will use only the order functions of the type
$$m(x,v)=(1+|x|^2)^k(1+|v|^2)^\ell.$$
Given an order function $m$, the class of symbols $S(m)$ is defined as
$$S(m):=\{ a\in C^\ii(\R^d\times\R^d),\ \forall\alpha\in\N^d,\ \exists C_\alpha>0,\ |\partial^\alpha a|\le C_\alpha m\}.$$
We recall the following standard results in semiclassical pseudo-differential calculus. The first one is contained in \cite[Thm. 7.9 and Prop. 7.7]{DimSjo-book}.

\begin{proposition}\label{prop:dimsjo1}
 Let $a_1\in S(m_1)$ and $a_2\in S(m_2)$. Then, we have
 $$\opw^\hbar(a_1)\opw^\hbar(a_2)=\opw^\hbar(c)$$
 where $c\in S(m_1m_2)$ has the expansion
 $$c=ab+\cO(\hbar)$$
 in the topology of $S(m_1m_2)$.
\end{proposition}

The second one is contained in \cite[Thm. 8.7 and following remarks]{DimSjo-book}.

\begin{proposition}\label{prop:dimsjo2}
 Let $p\in S(m)$ such that $p+i$ is elliptic in $S(m)$ (that is, such that $|p+i|\ge Cm$ for some $C>0$). Let $f\in C^\ii_0(\R)$. Then, $f(P)=\opw^\hbar(c)$ for some $c\in S(1/m)$ with the expansion 
 $$c=f\circ p+\cO(\hbar)$$
 in the topology of $S(1/m)$.
\end{proposition}

Finally, we will use \cite[Thm. 9.4]{DimSjo-book}:

\begin{proposition}\label{prop:dimsjo3}
 Let $a\in\cS'(\R^d\times\R^d)$ such that $\partial^\alpha a_{x,v}\in L^1(\R^d\times\R^d)$ for all $|\alpha|\le 2d+1$. Then, $\opw^\hbar(a)$ is trace-class and 
 $$\tr\opw^\hbar(a)=\int_{\R^d}\int_{\R^d}a(x,v)\frac{dx\,dv}{(2\pi\hbar)^d}.$$
\end{proposition}

All these results allow us to prove:

\begin{proposition}\label{prop:wigner-function}
 Let $a\in C^\ii_0(\R^d\times\R^d)$ and $f\in C^\ii_0(\R)$. Then, the Wigner transform of $f(\opw^\hbar(|v|^2+a))$ satisfies 
 $$W^\hbar[f(\opw^\hbar(|v|^2+a))]\to f(|v|^2+a)$$
 as $\hbar\to0$, in $\cD'(\R^d\times\R^d)$. 
\end{proposition}

\begin{proof}
 Let $\phi\in C^\ii_0(\R^d\times\R^d)$, and let us compute
 $$\langle W^\hbar[f(\opw^\hbar(|v|^2+a))],\phi\rangle = (2\pi\hbar)^d \tr f(\opw^\hbar(|v|^2+a))\opw^\hbar(\phi).$$
 Using Proposition \ref{prop:dimsjo2} with $p(x,v)=|v|^2+a(x,v)$ and $m(x,v)=m_1(x,v)=1+|v|^2$, we have $f(\opw^\hbar(|v|^2+a))=\opw^\hbar(c_1)$ with $c_1=f(|v|^2+a)+\cO(\hbar)$ in $S(1/m_1)$. Using now Proposition \ref{prop:dimsjo1} with $m_2=1+|x|^2+|v|^2$ and using that $\phi\in S(m_2^{-k})$ for all $k\in\N$, we deduce that 
 $$f(\opw^\hbar(|v|^2+a))\opw^\hbar(\phi)=\opw^\hbar(c_2)$$
 with $c_2=f(|v|^2+a)\phi+\cO(\hbar)$ in $S(m_2^{-k})$ for all $k\in\N$. Since any symbol in $S(m_2^{-k})$ for $k$ large enough satisfies the assumptions of Proposition \ref{prop:dimsjo3}, we deduce that 
 \begin{multline*}
    \tr f(\opw^\hbar(|v|^2+a))\opw^\hbar(\phi)\\
    =\frac{1}{(2\pi\hbar)^d}\left(\int_{\R^d}\int_{\R^d}f(|v|^2+a(x,v))\phi(x,v)\,dx\,dv+\cO(\hbar)\right)
 \end{multline*}
 and thus 
 $$\langle W^\hbar[f(\opw^\hbar(|v|^2+a))],\phi\rangle\to\langle f(|v|^2+a),\phi\rangle$$
 as $\hbar\to0$, which is exactly the result.
\end{proof}

\begin{corollary}\label{coro:wigner}
 The conclusions of Proposition \ref{prop:wigner-function} also hold when $f$ is continuous and vanishes at infinity.
\end{corollary}

\begin{proof}
 Let $f_\epsilon\in C^\ii_0(\R^d\times\R^d)$ be such that $f_\epsilon\to f$ as $\epsilon\to0$ in $L^\ii$. Let $\phi\in C^\ii_0(\R^d\times\R^d)$ and decompose 
 \begin{multline*}
  \langle W^\hbar[f(\opw^\hbar(|v|^2+a))],\phi\rangle - \langle f(|v|^2+a),\phi\rangle =\\
  \langle W^\hbar[f_\epsilon(\opw^\hbar(|v|^2+a))],\phi\rangle - \langle f_\epsilon(|v|^2+a),\phi\rangle\\
  +(2\pi\hbar)^d\tr (f-f_\epsilon)(\opw^\hbar(|v|^2+a))\opw^\hbar(\phi)+\langle (f_\epsilon-f)(|v|^2+a),\phi\rangle.
 \end{multline*}
 The second line can be estimated by 
 \begin{multline*}
    \left|(2\pi\hbar)^d\langle\tr (f-f_\epsilon)(\opw^\hbar(|v|^2+a))\opw^\hbar(\phi)+\langle (f_\epsilon-f)(|v|^2+a),\phi\rangle\right|\\
    \le\|f-f_\epsilon\|_{L^\ii}((2\pi\hbar)^d\|\opw^\hbar(\phi)\|_{\gS^1}+\|\phi\|_{L^1}),
 \end{multline*}
 and using that $\|\opw^\hbar(\phi)\|_{\gS^1}\le C_\phi \hbar^{-d}$ by \cite[Eq. (9.2)]{DimSjo-book}, we deduce that
 \begin{multline*}
  \langle W^\hbar[f(\opw^\hbar(|v|^2+a))],\phi\rangle - \langle f(|v|^2+a),\phi\rangle\\
  =\langle W^\hbar[f_\epsilon(\opw^\hbar(|v|^2+a))],\phi\rangle - \langle f_\epsilon(|v|^2+a),\phi\rangle+o_{\epsilon\to0}(1),
  \end{multline*}
  where the $o(1)$ is uniform in $\hbar$. Applying Proposition \ref{prop:wigner-function} to $f_\epsilon$ and taking the limit $\hbar\to0$ with $\epsilon$ fixed, we deduce that 
  $$\limsup_{\hbar\to0}\left|\langle W^\hbar[f(\opw^\hbar(|v|^2+a))],\phi\rangle - \langle f(|v|^2+a),\phi\rangle\right|=o_{\epsilon\to0}(1),$$
  which gives the result as $\epsilon\to0$.
\end{proof}

\section{An integral representation for the relative entropy}\label{app:int-rep}

\begin{lemma}
 Let $S:[0,M]\to\R$ a strictly concave and continuous function, differentiable on $(0,M)$. Let $A$ and $B$ two self-adjoint matrices of same size with spectrum contained in $\ran(S')$. Define $\gamma_A=(S')^{-1}(A)$ and $\gamma_B=(S')^{-1}(B)$. Then, we have the representation of the relative entropy 
 \begin{equation}\label{eq:rep-rel-ent}
    \cH_S(\gamma_A,\gamma_B)=\int_0^1\tr\Big(\left\{(S')^{-1}(tA+(1-t)B)-(S')^{-1}(A)\right\}(A-B)\Big)\,dt.
 \end{equation}
\end{lemma}

\begin{remark}
 Notice that the trace under the integral is non-negative by Klein's inequality since the function $(x,y)\mapsto ((S')^{-1}(x)-(S')^{-1}(y))(y-x)$ is non-negative. 
\hfill$\diamond$\end{remark}

\begin{proof}
 We first write the relative entropy in terms of a difference of free energies with a linear remainder term:
 \begin{multline*}
    \cH_S(\gamma_A,\gamma_B)\\
    =\tr(S'(\gamma_A)\gamma_A-S(\gamma_A))-\tr(S'(\gamma_B)\gamma_B-S(\gamma_B))+\tr(S'(\gamma_B)-S'(\gamma_A))\gamma_A.
 \end{multline*}
 Next, we notice that 
 $$S'(\gamma_A)\gamma_A-S(\gamma_A)=A(S')^{-1}(A)-S((S')^{-1}(A)),$$
 and the key remark is that the derivative of the function $x\mapsto x(S')^{-1}(x)-S((S')^{-1}(x))$ is simply $x\mapsto(S')^{-1}(x)$. As a consequence,
 $$\cH_S(\gamma_A,\gamma_B)=\int_0^1\tr (S')^{-1}(tA+(1-t)B)(A-B)\,dt+\tr(B-A)(S')^{-1}(A),$$
 which is the desired formula.
\end{proof}

\begin{proposition}
 Let $S:[0,M]\to\R$ be a concave and continuous function, differentiable on $(0,M)$, such that $-S'$ is operator monotone and such that $\lim_{x\to0}S'(x)=+\ii$. Assume furthermore that $(S')^{-1}$ and $x(S')^{-1}$ are continuous and vanish at infinity, together with their first and second derivatives. Define $M_-:=\lim_{x\to M}S'(x)$. Then, the formula \eqref{eq:rep-rel-ent} extends to any pair of self-adjoint operators $A$ and $B$ on infinite-dimensional Hilbert spaces such that $A,B\ge m$ for some $m>M_-$ and such that $A-B\in\gS^1$.
\end{proposition}

\begin{proof}
 We will first use the result of Peller \cite{Peller-85} (see also \cite{FraPus-15}) that says that if $C$ and $D$ are self-adjoint operators with $C-D\in\gS^1$ and if $g:\R\to\R$ belongs to the Besov space $B^1_{\ii,1}$ (which is for instance the case if $g$ and $g''$ are bounded), then $g(C)-g(D)\in\gS^1$ and one has the bound
 \begin{equation}\label{eq:peller}
    \|g(C)-g(D)\|_{\gS^1}\le c\|g\|_{B^1_{\ii,1}}\|C-D\|_{\gS^1},
 \end{equation}
 for some $c>0$ independent of $g$, $C$, and $D$. Under our assumptions on $(S')^{-1}$, we deduce from this result that
 $$\gamma_A-\gamma_B=(S')^{-1}(A)-(S')^{-1}(B)\in\gS^1,\ A(S')^{-1}(A)-B(S')^{-1}(B)\in\gS^1$$
 since $A-B\in\gS^1$. Notice that the functions $(S')^{-1}$ and $x(S')^{-1}(x)$ may not be defined on the whole line $\R$; however since $A,B\ge m>M_-$, we may replace extend these functions on $(-\ii,m)$ in order to obtain a function defined on $\R$, which is for instance smooth and compactly supported on $(-\ii,m)$. Now we also have that
 $$S(\gamma_A)-S(\gamma_B)=g(A)-g(B)+A(S')^{-1}(A)-B(S')^{-1}(B),$$
 with $g(x)=S((S')^{-1}(x))-x(S')^{-1}(x)$. The function $g$ is clearly bounded since $S$ and $x(S')^{-1}$ are bounded, and one may check that $g'(x)=-(S')^{-1}$, hence $g''$ is also bounded. By Peller's result, we deduce that $S(\gamma_A)-S(\gamma_B)\in\gS^1$. Finally, using that 
 \begin{align*}
    S'(\gamma_B)(\gamma_A-\gamma_B) &= B((S')^{-1}(A)-(S')^{-1}(B)) \\
    &=(B-A)(S')^{-1}(A)+A(S')^{-1}A-B(S')^{-1}(B),
 \end{align*}
 we deduce that $S'(\gamma_B)(\gamma_A-\gamma_B)\in\gS^1$ since $(S')^{-1}(A)$ is bounded. We thus have proved that $S(\gamma_A)-S(\gamma_B)$, $\gamma_A-\gamma_B$, and $S'(\gamma_B)(\gamma_A-\gamma_B)$ are all trace-class, so that by \cite[Thm. 4]{DeuHaiSei-15} we have
 $$\cH_S(\gamma_A,\gamma_B)=-\tr(S(\gamma_A)-S(\gamma_B)-S'(\gamma_B)(\gamma_A-\gamma_B)).$$
 With the same function $g$ as before, we deduce that 
 $$\cH_S(\gamma_A,\gamma_B)=-\tr(g(A)-g(B))-\tr(S')^{-1}(A)(A-B).$$
 It remains to justify the formula
 $$-\tr(g(A)-g(B))=-\int_0^1\tr g'(tA+(1-t)B)(A-B)\,dt,$$
 which holds in finite dimensions since we may separate the traces. To extend it to infinite dimensions, we apply it to $A_n:=P_nAP_n$ and $B_n=P_nBP_n$, where $(P_n)$ is a sequence of finite-rank orthogonal projections such that $P_n\to1$ strongly as $n\to\ii$ and such that $\ran P_n\subset D(A)$. Since $A-B$ is trace-class, we deduce that $D(A)$ is a core for $B$ and thus $(tA_n+(1-t)B_n)\phi\to (tA+(1-t)B)\phi$ for all $\phi\in D(A)$. By \cite[Thm. VIII.20 \& VIII.25]{ReeSim1}, this implies that $g'(tA_n+(1-t)B_n)\to g'(tA+(1-t)B)$ strongly as $n\to\ii$. Since $A-B\in\gS^1$, we deduce by dominated convergence that 
 $$\int_0^1\tr g'(tA_n+(1-t)B_n)(A_n-B_n)\,dt\to \int_0^1\tr g'(tA+(1-t)B)(A-B)\,dt$$
 as $n\to\ii$. Now we treat the left-side. First, by Peller's estimate \eqref{eq:peller}, we only have to prove that 
 $$\tr \phi(A_n)-\phi(B_n)\to \tr \phi(A)-\phi(B)$$
 for $\phi\in C^\ii_0(\R)$, approaching $g$ in $B^1_{\ii,1}$ by such a $\phi$. For such a nice $\phi$, one may use the Helffer-Sj\"ostrand formula to prove that $\phi(A_n)-\phi(B_n)\to\phi(A)-\phi(B)$ in $\gS^1$ since $A_n-B_n\to A-B$ as $n\to\ii$ in $\gS^1$, so that for all $z\in\C\setminus\R$, 
 $$\frac{1}{z-A_n}-\frac{1}{z-B_n}=\frac{1}{z-A_n}(B_n-A_n)\frac{1}{z-B_n}\to\frac{1}{z-A}(B-A)\frac{1}{z-B}$$
 as $n\to\ii$ in $\gS^1$, together with the uniform estimate
 $$\left\|\frac{1}{z-A_n}(B_n-A_n)\frac{1}{z-B_n}\right\|_{\gS^1}\le\frac{C}{|\text{Im}(z)|^2}.$$
\end{proof}

\begin{remark}
 Our assumptions on $(S')^{-1}$ include the physical cases of boltzons, bosons, and fermions since $(S')^{-1}$ is smooth and decays exponentially in these cases. They also include power like entropies of the form $S(x)=x^m-ax$ with $\max(0,1-2/d)<m<1$, because $(S')^{-1}(y)$ behaves like $y^{-1/(1-m)}$ for $y\to+\ii$, with $1/(1-m)>1$.
\hfill$\diamond$\end{remark}


\begin{thebibliography}{10}

\bibitem{AkiMarSpa-08}
{\sc G.~L. Aki, P.~A. Markowich, and C.~Sparber}, {\em Classical limit for
  semirelativistic {H}artree systems}, J. Math. Phys., 49 (2008), pp.~102110,
  10.

\bibitem{AthPau-11}
{\sc A.~Athanassoulis and T.~Paul}, {\em Strong phase-space semiclassical
  asymptotics}, SIAM J. Math. Anal., 43 (2011), pp.~2116--2149.

\bibitem{BacBrePetPicTza-15}
{\sc V.~Bach, S.~Breteaux, S.~Petrat, P.~Pickl, and T.~Tzaneteas}, {\em Kinetic
  energy estimates for the accuracy of the time-dependent {H}artree-{F}ock
  approximation with {C}oulomb interaction}, J. Math. Pures Appl. (9), 105
  (2016), pp.~1--30.

\bibitem{BarGolGotMau-03}
{\sc C.~Bardos, F.~Golse, A.~Gottlieb, and N.~Mauser}, {\em Mean field dynamics
  of fermions and the time-dependent {H}artree-{F}ock equation}, J. Math. Pures
  Appl. (9), 82 (2003), pp.~665--683.

\bibitem{BedMouMas-16}
{\sc J.~Bedrossian, N.~Masmoudi, and C.~Mouhot}, {\em Landau damping:
  paraproducts and {G}evrey regularity}, Ann. PDE, 2 (2016), pp.~Art. 4, 71.

\bibitem{BedMouMas-18}
\leavevmode\vrule height 2pt depth -1.6pt width 23pt, {\em Landau damping in
  finite regularity for unconfined systems with screened interactions}, Comm.
  Pure Appl. Math., 71 (2018), pp.~537--576.

\bibitem{BenPorSafSch-16}
{\sc N.~Benedikter, M.~Porta, C.~Saffirio, and B.~Schlein}, {\em From the
  {H}artree dynamics to the {V}lasov equation}, Arch. Ration. Mech. Anal., 221
  (2016), pp.~273--334.

\bibitem{BenPorSch-14}
{\sc N.~Benedikter, M.~Porta, and B.~Schlein}, {\em Mean-field evolution of
  fermionic systems}, Comm. Math. Phys., 331 (2014), pp.~1087--1131.

\bibitem{Berezin-72}
{\sc F.~A. Berezin}, {\em Covariant and contravariant symbols of operators},
  Izv. Akad. Nauk SSSR Ser. Mat., 36 (1972), pp.~1134--1167.

\bibitem{BraHep-77}
{\sc W.~Braun and K.~Hepp}, {\em The {V}lasov dynamics and its fluctuations in
  the {$1/N$} limit of interacting classical particles}, Comm. Math. Phys., 56
  (1977), pp.~101--113.

\bibitem{Carlen-10}
{\sc E.~Carlen}, {\em Trace inequalities and quantum entropy: an introductory
  course}, in Entropy and the Quantum, R.~Sims and D.~Ueltschi, eds., vol.~529
  of Contemporary Mathematics, American Mathematical Society, 2010,
  pp.~73--140.
\newblock Arizona School of Analysis with Applications, March 16-20, 2009,
  University of Arizona.

\bibitem{CheHonPav-16}
{\sc T.~Chen, Y.~Hong, and N.~Pavlovi\'c}, {\em Global well-posedness of the
  {NLS} system for infinitely many fermions}, Arch. Ration. Mech. Anal., 224
  (2017), pp.~91--123.

\bibitem{CheHonPav-17}
{\sc T.~Chen, Y.~Hong, and N.~Pavlovi\'{c}}, {\em On the scattering problem for
  infinitely many fermions in dimensions {$d\ge3$} at positive temperature},
  Ann. Inst. H. Poincar\'{e} Anal. Non Lin\'{e}aire, 35 (2018), pp.~393--416.

\bibitem{ColSuz-19_ppt}
{\sc C.~{Collot} and A.-S. {de Suzzoni}}, {\em {Stability of equilibria for a
  Hartree equation for random fields}}, arXiv e-prints,  (2018),
  p.~arXiv:1811.03150.

\bibitem{DeuHaiSei-15}
{\sc A.~Deuchert, C.~Hainzl, and R.~Seiringer}, {\em Note on a family of
  monotone quantum relative entropies}, Lett. Math. Phys., 105 (2015),
  pp.~1449--1466.

\bibitem{DieRadSch-18}
{\sc E.~Dietler, S.~Rademacher, and B.~Schlein}, {\em From {H}artree dynamics
  to the relativistic {V}lasov equation}, J. Stat. Phys., 172 (2018),
  pp.~398--433.

\bibitem{DimSjo-book}
{\sc M.~Dimassi and J.~Sjostrand}, {\em Spectral asymptotics in the
  semi-classical limit}, no.~268, Cambridge university press, 1999.

\bibitem{Dobrushin-79}
{\sc R.~Dobru\v{s}in}, {\em Vlasov equations}, Funct. Anal. Appl., 13 (1979),
  pp.~115--123.

\bibitem{ElgErdSchYau-04}
{\sc A.~Elgart, L.~Erd\H{o}s, B.~Schlein, and H.-T. Yau}, {\em Nonlinear
  {H}artree equation as the mean field limit of weakly coupled fermions}, J.
  Math. Pures Appl., 83 (2004), pp.~1241--1273.

\bibitem{EvaGar-15}
{\sc L.~C. Evans and R.~F. Gariepy}, {\em Measure theory and fine properties of
  functions}, Textbooks in Mathematics, CRC Press, Boca Raton, FL, revised~ed.,
  2015.

\bibitem{FraLewLieSei-12}
{\sc R.~L. Frank, M.~Lewin, E.~H. Lieb, and R.~Seiringer}, {\em A positive
  density analogue of the {L}ieb-{T}hirring inequality}, Duke Math. J., 162
  (2012), pp.~435--495.

\bibitem{FraPus-15}
{\sc R.~L. Frank and A.~Pushnitski}, {\em Trace class conditions for functions
  of {S}chr\"{o}dinger operators}, Comm. Math. Phys., 335 (2015), pp.~477--496.

\bibitem{FroKno-11}
{\sc J.~Fr{\"o}hlich and A.~Knowles}, {\em A microscopic derivation of the
  time-dependent {H}artree-{F}ock equation with {C}oulomb two-body
  interaction}, J. Stat. Phys., 145 (2011), pp.~23--50.

\bibitem{GanMcc-96}
{\sc W.~Gangbo and R.~J. McCann}, {\em The geometry of optimal transportation},
  Acta Math., 177 (1996), pp.~113--161.

\bibitem{GiuVig-05}
{\sc G.~Giuliani and G.~Vignale}, {\em Quantum Theory of the Electron Liquid},
  Cambridge University Press, 2005.

\bibitem{GolMouPau-16}
{\sc F.~Golse, C.~Mouhot, and T.~Paul}, {\em On the mean field and classical
  limits of quantum mechanics}, Comm. Math. Phys., 343 (2016), pp.~165--205.

\bibitem{GolPau-17}
{\sc F.~Golse and T.~Paul}, {\em The {S}chr\"{o}dinger equation in the
  mean-field and semiclassical regime}, Arch. Ration. Mech. Anal., 223 (2017),
  pp.~57--94.

\bibitem{HanNguRou-19}
{\sc D.~Han-Kwan, T.~T. Nguyen, and F.~Rousset}, {\em Asymptotic stability of
  equilibria for screened {V}lasov-{P}oisson systems via pointwise dispersive
  estimates}, arXiv preprint arXiv:1906.05723,  (2019).

\bibitem{HelRob-81}
{\sc B.~Helffer and D.~Robert}, {\em Comportement semi-classique du spectre des
  hamiltoniens quantiques elliptiques}, Ann. Inst. Fourier (Grenoble), 31
  (1981), pp.~xi, 169--223.

\bibitem{Hepp-74}
{\sc K.~Hepp}, {\em The classical limit for quantum mechanical correlation
  functions}, Commun. Math. Phys., 35 (1974), pp.~265--277.

\bibitem{Landau-46}
{\sc L.~D. Landau}, {\em {On the vibrations of the electronic plasma}}, J.
  Phys.(USSR), 10 (1946), pp.~25--34.

\bibitem{LewNamRou-15}
{\sc M.~Lewin, P.~T. Nam, and N.~Rougerie}, {\em Derivation of nonlinear
  {G}ibbs measures from many-body quantum mechanics}, J. \'{E}c. polytech.
  Math., 2 (2015), pp.~65--115.

\bibitem{LewSab-13}
{\sc M.~Lewin and J.~Sabin}, {\em A family of monotone quantum relative
  entropies}, Lett. Math. Phys., 104 (2014), pp.~691--705.

\bibitem{LewSab-13b}
{\sc M.~Lewin and J.~Sabin}, {\em The {H}artree equation for infinitely many
  particles. {II}. {D}ispersion and scattering in {2D}}, Analysis and PDE, 7
  (2014), pp.~1339--1363.

\bibitem{LewSab-13a}
{\sc M.~Lewin and J.~Sabin}, {\em The {H}artree equation for infinitely many
  particles {I}. {W}ell-posedness theory}, Comm. Math. Phys., 334 (2015),
  pp.~117--170.

\bibitem{LieRus-73b}
{\sc E.~Lieb and M.-B. Ruskai}, {\em Proof of the strong subadditivity of
  quantum-mechanical entropy}, J. Math. Phys., 14 (1973), pp.~1938--1941.
\newblock With an appendix by B. Simon.

\bibitem{Lieb-73}
{\sc E.~H. Lieb}, {\em The classical limit of quantum spin systems}, Comm.
  Math. Phys., 31 (1973), pp.~327--340.

\bibitem{LioPau-93}
{\sc P.-L. Lions and T.~Paul}, {\em Sur les mesures de {W}igner}, Rev. Mat.
  Iberoamericana, 9 (1993), pp.~553--618.

\bibitem{Loeper-06}
{\sc G.~Loeper}, {\em Uniqueness of the solution to the {V}lasov-{P}oisson
  system with bounded density}, J. Math. Pures Appl. (9), 86 (2006),
  pp.~68--79.

\bibitem{LunMar-83}
{\sc S.~Lundqvist and N.~March}, eds., {\em Theory of the Inhomogeneous
  Electron Gas}, Physics of Solids and Liquids, Springer US, 1983.

\bibitem{MarMau-93}
{\sc P.~A. Markowich and N.~J. Mauser}, {\em The classical limit of a
  self-consistent quantum-{V}lasov equation in {$3$}{D}}, Math. Models Methods
  Appl. Sci., 3 (1993), pp.~109--124.

\bibitem{Maslov-72}
{\sc V.~P. Maslov}, {\em Th{\'e}orie des perturbations et m{\'e}thodes
  asymptotiques}, Dunod, Paris,  (1972).

\bibitem{McCann-95}
{\sc R.~J. McCann}, {\em Existence and uniqueness of monotone
  measure-preserving maps}, Duke Math. J., 80 (1995), pp.~309--323.

\bibitem{McCann-97}
\leavevmode\vrule height 2pt depth -1.6pt width 23pt, {\em A convexity
  principle for interacting gases}, Adv. Math., 128 (1997), pp.~153--179.

\bibitem{MouVil-11}
{\sc C.~Mouhot and C.~Villani}, {\em On {L}andau damping}, Acta Math., 207
  (2011), pp.~29--201.

\bibitem{NarSew-81}
{\sc H.~Narnhofer and G.~L. Sewell}, {\em Vlasov hydrodynamics of a quantum
  mechanical model}, Comm. Math. Phys., 79 (1981), pp.~9--24.

\bibitem{OhyPet-93}
{\sc M.~Ohya and D.~Petz}, {\em Quantum entropy and its use}, Texts and
  Monographs in Physics, Springer-Verlag, Berlin, 1993.

\bibitem{Peller-85}
{\sc V.~V. Peller}, {\em Hankel operators in the theory of perturbations of
  unitary and selfadjoint operators}, Funktsional. Anal. i Prilozhen., 19
  (1985), pp.~37--51, 96.

\bibitem{Penrose-60}
{\sc O.~Penrose}, {\em Electrostatic instabilities of a uniform non-maxwellian
  plasma}, Phys. Fluids, 3 (1960), pp.~258--265.

\bibitem{PetPic-16}
{\sc S.~Petrat and P.~Pickl}, {\em A new method and a new scaling for deriving
  fermionic mean-field dynamics}, Math. Phys. Anal. Geom., 19 (2016).

\bibitem{PitVir-15}
{\sc J.~Pitrik and D.~Virosztek}, {\em On the joint convexity of the {B}regman
  divergence of matrices}, Lett. Math. Phys., 105 (2015), pp.~675--692.

\bibitem{ReeSim1}
{\sc M.~Reed and B.~Simon}, {\em Methods of Modern Mathematical Physics. I.
  Functional Analysis}, Academic Press, 1972.

\bibitem{Rudin}
{\sc W.~Rudin}, {\em Real and complex analysis}, McGraw-Hill Book Co., New
  York, third~ed., 1987.

\bibitem{Schro-26}
{\sc E.~Schr{\"o}dinger}, {\em Der stetige {\"u}bergang von der mikro- zur
  makromechanik}, Naturwissenschaften, 14 (1926), pp.~664--666.

\bibitem{Simon-81}
{\sc B.~Simon}, {\em The classical limit of quantum partition functions}, Comm.
  Math. Phys., 71 (1980), pp.~247--276.

\bibitem{Spohn-81}
{\sc H.~Spohn}, {\em On the {V}lasov hierarchy}, Math. Methods Appl. Sci., 3
  (1981), pp.~445--455.

\bibitem{Tatarski-83}
{\sc V.~I. Tatarski{\u{\i}}}, {\em The {W}igner representation of quantum
  mechanics}, Soviet Physics Uspekhi, 26 (1983), pp.~311--327.

\bibitem{Villani-book}
{\sc C.~Villani}, {\em Optimal transport}, vol.~338 of Grundlehren der
  Mathematischen Wissenschaften [Fundamental Principles of Mathematical
  Sciences], Springer-Verlag, Berlin, 2009.
\newblock Old and new.

\bibitem{Wigner-32}
{\sc E.~Wigner}, {\em On the quantum correction for thermodynamic equilibrium},
  Phys. Rev., 40 (1932), pp.~749--759.

\bibitem{Zworski-book}
{\sc M.~Zworski}, {\em Semiclassical analysis}, vol.~138 of Graduate Studies in
  Mathematics, American Mathematical Society, Providence, RI, 2012.

\end{thebibliography}
\end{document}